%% file: 0main.tex
\documentclass[12pt]{article}
\def\blind{0}
\def\withsupp{1}
\input{preamble}

\input{1titlepage}
\begin{document}

    \doublespacing
    
    \maketitle
    \makeAbstract
    \input{2intro}

    \input{6methods}

\input{3results}

    \input{5discussion}

    \input{appendix}
    \input{7backmatter}
    \begin{singlespace}
        \setlength{\bibsep}{0pt plus 0.3ex}
        \bibliographystyle{apalike}
        \bibliography{library}
    \end{singlespace}
    \clearpage
    %% Table 1 is now \input inside 3results.tex, next to its first citation.
    %% The ASA style guide permits tables either near first reference or at the
    %% end; embedding it saves the page it occupied on its own after the
    %% bibliography.

    %% ---------------- Supplementary Material (appended) ----------------
    %% Figure and table S-numbering is set inside 1suppl_figs.tex and
    %% 2suppl_tabs.tex; only the section counter is set here.
    \ifnum\withsupp=1
    \clearpage
    \renewcommand{\thesection}{S\arabic{section}}
    \setcounter{section}{0}
    \begin{center}
    {\large\bf Supplementary Material for\\[2pt]
    ``Improving Mendelian Randomization Analysis by Instrument Borrowing from Auxiliary Outcome Traits''}
    \end{center}
    \input{supplement/0suppl_text}

    \input{supplement/1suppl_figs}

    \input{supplement/2suppl_tabs}

    \fi

\end{document}

%% file: preamble.tex
\usepackage[utf8]{inputenc}
\usepackage[english]{babel}
\usepackage{wrapfig}
\usepackage{overpic}
\usepackage[letterpaper,left=1in,right=1in,top=1in,bottom=1.3in]{geometry}
\usepackage{authblk}
\usepackage[colorlinks=true, linkcolor=black, citecolor=black, urlcolor=blue, filecolor=blue,breaklinks=true]{hyperref}
\usepackage{mdframed}
\usepackage{array}
\usepackage{microtype}
\usepackage{float}
\usepackage{setspace}
\usepackage{setspace}
\usepackage{lineno}
\usepackage{csquotes}
\usepackage{ragged2e}
\usepackage{algorithm}
\usepackage{algpseudocode}
\usepackage{amssymb}
\usepackage{graphicx}
\graphicspath{{./main_figs/}{./supplement/suppl_figs/}}
\DeclareGraphicsExtensions{.pdf,.jpeg,.JPG,.png,.PNG, .eps, .tiff}
\usepackage{subcaption}
\DeclareCaptionLabelFormat{bold}{\textbf{(#2)}}
\newcommand{\labelphantom}[1]{
  \parbox{0pt}{\phantomsubcaption\label{#1}}
}
\usepackage{pgffor}
\usepackage{alphalph}
\usepackage{longtable}
\usepackage[labelfont=bf, textfont=bf, singlelinecheck=off, textfont=footnotesize]{caption}
\usepackage{booktabs}
\usepackage{multirow}
\usepackage{rotating}
\usepackage{amsmath}
\usepackage{float}

\usepackage{makecell}
\usepackage{arydshln}

\providecommand{\blind}{1}

\input{helperFunctions}

\usepackage[dvipsnames]{xcolor}
\usepackage{changepage}
\usepackage{booktabs,tabularx}
\usepackage{xcolor}
\usepackage{enumitem}
\usepackage{tikz}
\usetikzlibrary{positioning, shapes, arrows.meta, shadows, fit, calc, backgrounds, shadows.blur, decorations.pathmorphing,patterns}
\definecolor{instBlue}{RGB}{102,153,204}
\definecolor{exposetGreen}{RGB}{153,204,153}
\definecolor{outcomeOrange}{RGB}{255,179,102}
\definecolor{confoundGrey}{RGB}{200,200,200}
\definecolor{arrowColor}{RGB}{80,80,80}

\definecolor{exposure}{RGB}{70,130,180}
\definecolor{trait1}{RGB}{220,20,60}
\definecolor{trait2}{RGB}{50,205,50}
\definecolor{drug}{RGB}{255,140,0}
\definecolor{coheterogeneity}{RGB}{147,112,219}

\tikzset{
    instrument/.style={
        draw=instBlue!70!black,
        fill=instBlue!20,
        ellipse,
        minimum width=2.2cm,
        minimum height=1cm,
        font=\small\sf,
        drop shadow
    },
    exposure/.style={
        draw=exposetGreen!70!black,
        fill=exposetGreen!20,
        rounded rectangle,
        rounded corners=3pt,
        minimum width=2.2cm,
        minimum height=1cm,
        font=\small\sf,
        drop shadow
    },
    outcome/.style={
        draw=outcomeOrange!70!black,
        fill=outcomeOrange!20,
        rounded rectangle,
        rounded corners=3pt,
        minimum width=2.2cm,
        minimum height=1cm,
        font=\small\sf,
        drop shadow
    },
    confounder/.style={
        draw=confoundGrey!70!black,
        fill=confoundGrey!10,
        ellipse,
        minimum width=2.4cm,
        minimum height=1cm,
        font=\small\sf,
        dashed
    },
    causalArrow/.style={
        ->,
        thick,
        draw=arrowColor,
        >=Stealth[length=3mm,width=2mm]
    },
    violationArrow/.style={
        ->,
        thick,
        draw=red!70!black,
        dashed,
        >=Stealth[length=3mm,width=2mm]
    },
    edgeLabel/.style={
        font=\footnotesize\sf,
        fill=white,
        inner sep=1pt
    },
    legendBox/.style={
        draw=black!50,
        fill=white,
        rounded corners=2pt,
        font=\footnotesize\sf
    }
}

\tikzset{
  setCircle/.style 2 args={draw, thick, fill=#1!30, opacity=0.6, drop shadow={shadow blur steps=5,shadow xshift=0.5pt,shadow yshift=-0.5pt,shadow scale=1}} ,
  exposure/.style       ={rectangle, draw, thick, rounded corners, minimum width=1.6cm, minimum height=1cm, fill=white},
  outcome/.style        ={rectangle, draw, thick, rounded corners, minimum width=1.6cm, minimum height=1cm, fill=white},
  confounder/.style     ={ellipse, draw, thick, minimum width=1.8cm, minimum height=1cm, fill=white},
  arrow/.style          ={-{Latex[length=3mm]}, thick},
  violation/.style      ={dashed, -{Latex[length=3mm]}, thick},
  small/.style          ={font=\footnotesize, align=center}
}

\usepackage{newunicodechar}
\newunicodechar{ }{~}
\usepackage{placeins}
\usepackage{algpseudocode}
\usepackage{rmathbr}

\usepackage[round,authoryear]{natbib}
\usepackage{rotating} 
\usepackage{caption}
\usepackage{adjustbox}
\usepackage{lscape}

\usepackage{adjustbox}
\newcommand{\ignore}[1]{}

\newif\ifrevmark
\revmarkfalse
\newcommand{\rev}[1]{\ifrevmark\textcolor{red}{#1}\else#1\fi}
\newenvironment{revblock}{\ifrevmark\color{red}\fi}{}

\usepackage{amsmath,amsthm,amssymb}
\usepackage{enumitem}

\usepackage[capitalise, noabbrev, nameinlink]{cleveref}
\crefdefaultlabelformat{#2\textbf{#1}#3}
\Crefname{figure}{\textbf{Figure}}{\textbf{Figures}}
\Crefname{section}{\textbf{Section}}{\textbf{Sections}}
\Crefname{table}{\textbf{Table}}{\textbf{Tables}}

\newtheorem{theorem}{Theorem}
\newtheorem{lemma}{Lemma}
\newtheorem{slemma}{Lemma}

\crefname{slemma}{Lemma}{Lemmas}
\Crefname{slemma}{Lemma}{Lemmas}

\theoremstyle{definition}

\theoremstyle{remark}
\newtheorem{remark}{Remark}

\DeclareMathOperator{\Var}{Var}
\DeclareMathOperator{\Cov}{Cov}

%% file: helperFunctions.tex
\newcommand{\generateFigSubpanels}[6][0]{
     \expandafter\newcommand\csname#2\endcsname{
      \begin{figure}[htbp]
        \foreach [count=\i] \x in #6{
            \labelphantom{fig:#2:\AlphAlph{\i}}
        }
        \centering
        \includegraphics[width=#4\linewidth,angle=#1]{#3}
        \caption{\textbf{\normalsize #5}}\label{fig:#2}
        \footnotesize
        \justifying
        \foreach [count=\i] \x in #6{
            \noindent\subref{fig:#2:\AlphAlph{\i}}~\x\space
        }
      \end{figure}
    }
}

\newcommand{\generateFig}[6][0]{
     \expandafter\newcommand\csname#2\endcsname{
      \begin{figure}[htbp]
        \centering
        \includegraphics[width=#4\linewidth,angle=#1]{#3}
        \caption{\textbf{\normalsize #5}\label{fig:#2}
        \footnotesize
        \normalfont
        #6
        }
        
      \end{figure}
    }
}

\newcommand{\generateSidewaysFigSubpanels}[6][0]{
     \expandafter\newcommand\csname#2\endcsname{
      \begin{sidewaysfigure}[htbp]
        \foreach [count=\i] \x in #6{
            \labelphantom{fig:#2:\AlphAlph{\i}}
        }
        \centering
        \includegraphics[width=#4\linewidth,angle=#1]{#3}
        \caption{\textbf{\normalsize #5}}\label{fig:#2}
        \footnotesize
        \justifying
        \foreach [count=\i] \x in #6{
            \noindent\subref{fig:#2:\AlphAlph{\i}}~\x\space
        }
      \end{sidewaysfigure}
    }
}

\newcommand{\generateSidewaysFig}[6][0]{
     \expandafter\newcommand\csname#2\endcsname{
      \begin{sidewaysfigure}[htbp]
        \centering
        \includegraphics[width=#4\linewidth,angle=#1]{#3}
        \caption{\textbf{\normalsize #5}\label{fig:#2}
        \footnotesize
        \normalfont
        #6
        }
        
      \end{sidewaysfigure}
    }
}

\newcommand{\generateTab}[6][]{
     \expandafter\newcommand\csname#2\endcsname{ \begin{table}[ht]
          \centering
                \resizebox{#4\textwidth}{!}{
                \input{#3}
            }
            \caption{\textbf{\normalsize #5}\label{tab:#2}
            \footnotesize
            \normalfont
            #6
            }
            
        \end{table}
    }
}

\newcommand{\generateSidewaysTab}[6][]{
     \expandafter\newcommand\csname#2\endcsname{ \begin{sidewaystable}[ht]
          \centering
                \resizebox{#4\textwidth}{!}{
                \input{#3}
            }
            \caption{\textbf{\normalsize #5}\label{tab:#2}
            \footnotesize
            \normalfont
            #6
            }
            
        \end{sidewaystable}
    }
}

%% file: 1titlepage.tex
\ifnum\blind=1
  \title{\bf Improving Mendelian Randomization Analysis by Instrument Borrowing from Auxiliary Outcome Traits}
  \author{}
\else
  \title{\bf Improving Mendelian Randomization Analysis by Instrument Borrowing from Auxiliary Outcome Traits}
  \author[1]{Anagh Chattopadhyay}
  \author[1,2,*]{Nilanjan Chatterjee}
  \affil[1]{Department of Biostatistics, Johns Hopkins Bloomberg School of Public Health, Baltimore, MD}
  \affil[2]{Department of Oncology, School of Medicine, Johns Hopkins University, Baltimore, MD}
  \affil[ ]{ }
  \affil[*]{\textit{Correspondence to }\underline{nilanjan@jhu.edu}}
  
\fi
\date{}

\newcommand{\makeAbstract}{
\begin{center}\textbf{Abstract}\end{center}
\noindent
Mendelian randomization (MR) infers causal effects of exposures on outcomes using genetic variants as instrumental variables, but existing methods remain vulnerable to bias and loss of power when instruments are invalid. We hypothesize that closely related outcome traits share a large overlap in their valid instruments for a given exposure, so that borrowing instruments across such traits can yield more robust MR inference. To operationalize this, we introduce a novel coheterogeneity statistic, together with its asymptotic theory under different paradigms that allows for a growing number of instruments as the sample size increases, to identify secondary outcome traits likely to share valid instruments with the primary trait. We then propose extensions of two popular methods, MR-Mode and MR-PRESSO, that improve inference for the primary trait by exploiting a secondary trait. Extensive simulations show that the coheterogeneity statistic has the expected finite-sample properties and that instrument-borrowing methods consistently outperform their standard counterparts, with gains increasing as the overlap in valid instruments grows. Applications to established positive and negative control hypotheses indicate improved control of type I error and increased power. We further illustrate practical utility by revisiting the long-debated hypothesis of the causal effect of vitamin D on cardiometabolic traits.

\medskip
\noindent\textbf{Key Words:} Coheterogeneity; Horizontal pleiotropy; Instrumental variables;
Invalid instruments; Mendelian randomization; Robust estimation.
}

%% file: 2intro.tex
\section{Introduction}

\noindent
As genome-wide association studies (GWAS) have led to the identification of increasingly large numbers of genetic variants associated with complex traits and molecular biomarkers, the popularity of Mendelian randomization (MR) methods has also increased, which exploit these discoveries to explore downstream causal effects of these traits \citep{smith2004mendelian}. MR has been instrumental in elucidating numerous relationships between risk factors and diseases \citep{pierce2018mendelian, mordi2021type, yarmolinsky2024association, xu2024socioeconomic}, contributing to sociological research \citep{viinikainen2022does, rogne2024mediating} and investigating biological pathways \citep{xiuyun2022network}. In particular, MR approaches based on genetic variants associated with molecular traits such as proteins have increasingly gained traction to identify potential drug targets and therapeutic interventions \citep{walker2017mendelian, gill2021mendelian}. However, it is also increasingly recognized that the indiscriminate use of MR with large numbers of genetic instruments and widely available summary-statistics data leads to biases similar to those of observational studies due to violations of key assumptions \citep{qi2021comprehensive}.

Common genetic variants identified in GWAS often do not have specific effects on index traits but instead exhibit pleiotropic effects on a variety of related and unrelated traits \citep{watanabe2019global, qi2024genome}. The presence of widespread pleiotropy often leads to violations of the key MR assumption that genetic instruments associated with an index ``exposure'' trait do not affect the outcome of interest through underlying confounders or through other independent mechanisms \citep{pickrell2016detection, visscher201710, verbanck2018detection}. A wide variety of MR methods have emerged to tackle bias due to pleiotropy under alternative assumptions regarding the distribution of such pleiotropic effects across instruments \citep{bowden2015mendelian, hartwig2017robust, qi2019mendelian, burgess2020robust}. Most methods assume that there exists a hidden set of valid instruments and use a statistical model to develop a data-driven approach that leverages the valid instruments while ``discarding'' the invalid ones to improve inference.

Large-scale simulation studies, as well as analyses using positive and negative controls, show that many of these methods provide improved bias control while retaining power, although to varying degrees \citep{qi2021comprehensive, slob2020comparison}. Nevertheless, negative-control experiments \citep{qi2019mendelian}, including those reported below, continue to show that even some of the best-performing methods can be susceptible to bias in complex scenarios. In this article, we will show that while exposure–outcome pleiotropy is a ``curse'' for MR analysis, pleiotropy across outcome traits can actually be utilized to improve existing MR methods. Intuitively, related outcome traits such as metabolic disorders and psychiatric conditions are likely to share common confounders with respect to the exposures of interest. Thus, the sets of valid and invalid instruments for exposure–outcome relationships are likely to overlap across related outcome traits (see \cref{graph_summ}). Standard MR methods analyze one outcome trait at a time and, therefore, miss the opportunity to take advantage of such a pleiotropic structure.

In this article, we first introduce a novel measure of coheterogeneity for MR analysis to identify auxiliary outcome trait(s) that will have a large overlap with the primary trait with respect to the underlying valid/invalid instrument structure. We establish an asymptotic inferential theory for such a statistic, as an estimator for a population parameter defined by the selected instruments, in a framework that allows the instrument set size to increase with increasing GWAS sample size. We provide further theoretical results regarding the properties of the estimator in the presence of "weak instruments" and for the estimation of the limiting population parameter defined {\em by all instruments}. Finally, we propose simple extensions of popular MR methods, MR-Mode and MR-PRESSO, for the joint analysis of a primary trait with selected auxiliary trait(s). We used extensive simulation studies to show that instrument borrowing generally leads to an improvement in power without compromising type-I error. In addition, applications of proposed and existing methods for the investigation of a class of well-studied hypotheses, including positive and negative controls, demonstrate that the IB approach can better control bias, increase efficiency, or both. The framework also naturally lends itself to extensions of other existing MR methods for the joint analysis of multiple outcome traits.

%% file: 6methods.tex
\section{Methods}

\subsection{Background Setup and Notations}
\label{subsec:setup}

\cref{graph_summ} provides a graphical schema of the key steps of the proposed method. We assume that GWAS summary-level data are available for index exposure \(X\) and two outcomes: a primary trait \(Y_1\) and an auxiliary trait \(Y_2\). For simplicity, we assume that there is no sample overlap between the GWAS for the exposure and the GWAS for either outcome, but we allow potential sample overlap across the GWASs of the two outcome traits, \(Y_1\) and \(Y_2\), as related traits are often studied together.

Let \(N_X=N\) denote the effective sample size for the GWAS of \(X\). We assume that the effective sample sizes for \(Y_1\) and \(Y_2\) grow in fixed ratios with respect to \(N_X\), and therefore index the overall asymptotic regime by \(N\). More precisely, let \(N_l\) denote the GWAS sample size for \(Y_l\), and define \(\kappa_l:=N_l/N\). We assume that \(K_N\) independent SNPs, \(G_1,\dots,G_{K_N}\), have been chosen as instruments for \(X\), which typically require them to achieve a strict level of genome-wide significance, for example \(p<5\times 10^{-8}\). We write \(\mathcal K_N\) for this selected instrument set, so that \(K_N=|\mathcal K_N|\). Here, we explicitly allow the instrument set size to depend on sample size because the number of available instruments is expected to increase with increasing GWAS sample size. In the following, for simplicity of notation, we suppress the dependence of \(K\) and other quantities on \(N\) unless it is explicitly needed.

Let
\[\hat{\beta}_X = (\hat{\beta}_{X,1}, \ldots, \hat{\beta}_{X,K}), \quad 
\hat{\beta}_{Y_1} = (\hat{\beta}_{Y_1,1}, \ldots, \hat{\beta}_{Y_1,K}), \quad
\hat{\beta}_{Y_2} = (\hat{\beta}_{Y_2,1}, \ldots, \hat{\beta}_{Y_2,K})\]
denote the estimates of the standardized effects of \(G_1,\dots,G_K\) on \(X\), \(Y_1\), and \(Y_2\), respectively. For \(l\in\{1,2\}\) and \(k\in\mathcal K\), define the Wald ratio
\(\hat{\theta}_{l,k}:=\frac{\hat{\beta}_{Y_l,k}}{\hat{\beta}_{X,k}}\).
In vector form, let
\(\hat{\theta}_{l}=(\hat{\theta}_{l,1},\ldots,\hat{\theta}_{l,K})\)
denote the vector of ratio-estimates for outcome \(Y_l\). 

Valid instruments for the \(X\to Y_l\) pathway, \(l\in\{1,2\}\), are those with no direct effect on \(Y_l\) and no association with unmeasured confounders of the \(X\)--\(Y_l\) relationship. Under standard assumptions, the Wald ratios at valid instruments are consistent for the corresponding causal effects, and if all instruments were valid the inverse variance-weighted estimator (IVW) \citep{burgess2013mendelian} would produce efficient estimates. In reality, the set of available instruments is a mixture of valid and invalid instruments, and many alternative MR methods are available to account for invalid instruments under different assumptions about the distribution of effects associated with horizontal pleiotropy, i.e., the direct effects of the instruments on the outcome trait that are not mediated through the exposure.

For related outcomes, validity status may overlap because of shared biological pathways and confounding structures; see \cref{graph_summ}. Let \(\mathcal I_1\) and \(\mathcal I_2\) denote the sets of instruments invalid for \(Y_1\) and \(Y_2\), respectively, with respect to \(X\). We quantify the overlap of invalid instruments by
\(D^{\mathrm{inv}}_{\mathrm{ov}} = \frac{|\mathcal I_1\cap \mathcal I_2|} {|\mathcal I_1\cup \mathcal I_2|}\),
where \(|\cdot|\) denotes the cardinality of the set.
An analogous overlap measure can be defined for valid instruments, and it can be easily shown that they are monotonically related to each other.  For brevity, we write
\(D_{\mathrm{ov}}\equiv D^{\mathrm{inv}}_{\mathrm{ov}}\).

A summary of the notation used throughout the paper is provided in \Cref{tab:notation} in the Supplementary Material.

\subsection{A Measure of Coheterogeneity}
\label{subsec:coheterogeneity}

Our first goal is to identify pairs of outcome traits \((Y_1,Y_2)\) that may have substantial overlap in invalid instruments, indicated by \(D_{\mathrm{ov}}\). Because \(D_{\mathrm{ov}}\) is not directly observable, we introduce a summary-statistic-based measure designed to track the shared pleiotropic structure between outcomes. The construction is motivated by the heterogeneity parameter \(\tau^2\) in the meta-analysis of random-effects \citep{dersimonian1986meta}, which quantifies the variation in the estimates of the ratios between instruments. 

Consider the following first-order decomposition of the ratio estimate
\(\hat\theta_{l,k}=\theta_l+\alpha_{l,k}+\epsilon_{l,k}\),
where \(\theta_l\) is the causal effect of \(X\) on \(Y_l\), \(\alpha_{l,k}\) is the pleiotropic bias on the ratio-scale, and \(\epsilon_{l,k}\) is the sampling error, with \(\mathbb E(\epsilon_{l,k})=0\). Let \(\sigma_{l,k}^2:=\Var(\epsilon_{l,k})\). In practice, \(\sigma_{l,k}^2\) is estimated from GWAS standard errors, typically by using a delta-method approximation.

Define normalized precision weights as
\[w_k=
\frac{(\sigma_{1,k}^2\sigma_{2,k}^2)^{-1/2}}
{\sum_{j\in\mathcal K}(\sigma_{1,j}^2\sigma_{2,j}^2)^{-1/2}},
\qquad
\sum_{k\in\mathcal K}w_k=1.\]
Let \(\bar\alpha_l:=\sum_{k\in\mathcal K}w_k\alpha_{l,k}\) and \(\tilde\alpha_{l,k}:=\alpha_{l,k}-\bar\alpha_l\). Conditional on the selected set of instruments, define the target parameter in the form of a weighted correlation as
\(\rho_{CH}^{(N)} = \frac{C_{12}}{\tau_1\tau_2}\),
where \(C_{12}:=\sum_{k\in\mathcal K}w_k\tilde\alpha_{1,k}\tilde\alpha_{2,k}\) and \(\tau_l^2:=\sum_{k\in\mathcal K}w_k\tilde\alpha_{l,k}^2\). Large values of \(|\rho_{CH}^{(N)}|\) indicate a strong cross-trait alignment of pleiotropic deviations.

To estimate \(\rho_{CH}^{(N)}\), let \(\hat\sigma_{l,k}^2\) be an estimate of \(\sigma_{l,k}^2\), which are readily available from GWAS, and further let \(\hat\sigma_{12,k}\) be an estimate of \(\Cov(\epsilon_{1,k},\epsilon_{2,k})\), which can be obtained from genome-wide GWAS summary-statistics data by application of methods such as bivariate LD-score regression \citep{bulik2015atlas} (see \cref{sec:cov_est} for details). Define the normalized precision weights as
\[\hat w_k 
= 
\frac{(\hat\sigma_{1,k}^2\hat\sigma_{2,k}^2)^{-1/2}}
{\sum_{j\in\mathcal K}(\hat\sigma_{1,j}^2\hat\sigma_{2,j}^2)^{-1/2}}.\]
With \(\bar\theta_l:=\sum_{k\in\mathcal K}\hat w_k \hat\theta_{l,k}\) and \(\Delta_{l,k}:=\hat\theta_{l,k}-\bar\theta_l\), set
\(\hat\rho_{CH}^{(N)} = \frac{\hat C_{12}}{\hat\tau_1\hat\tau_2}\),
where
\begin{align}
\hat C_{12}
&=
\sum_{k\in\mathcal K}
\hat w_k 
\left(
\Delta_{1,k}\Delta_{2,k}
-
\hat\sigma_{12,k}
\right),
\label{eq:C12_est}\\
\hat\tau_l^2
&=
\left[
\sum_{k\in\mathcal K}
\hat w_k 
\left(
\Delta_{l,k}^2-\hat\sigma_{l,k}^2
\right)
\right]_+,
\qquad l=1,2,
\label{eq:tau_est}
\end{align}
and \([a]_+:=\max(a,0)\).
The subtractions in \eqref{eq:C12_est}--\eqref{eq:tau_est} debias the Wald-ratio moments for the sampling error, so that \(\hat C_{12}\) and \(\hat\tau_l^2\) target the underlying pleiotropic moments. 

The estimator \(\hat\rho_{CH}^{(N)}\) is a smooth functional of the stacked GWAS summary statistics \(\hat{\mathbf b}=(\hat{\mathbf b}_1,\dots,\hat{\mathbf b}_K)\), with \(\hat{\mathbf b}_k=(\hat\beta_{X,k},\hat\beta_{Y_1,k},\hat\beta_{Y_2,k})^\top\) and population counterpart \(\mathbf b_k=(\beta_{X,k},\beta_{Y_1,k},\beta_{Y_2,k})^\top\); we write \(\hat\rho_{CH}^{(N)}=T_N(\hat{\mathbf b})\) for this map, and note that the precision weights \(\hat w_k\) are themselves formed from \(\hat{\mathbf b}\). Let \(\Omega_k:=\Cov\{\sqrt N(\hat{\mathbf b}_k-\mathbf b_k)\}\) be the per-SNP sampling covariance of the GWAS estimators, and let \(\nabla_k T_N\) denote the gradient of \(T_N\) in the associations of SNP \(k\). The following theorem gives the conditional large-sample distribution of \(\hat\rho_{CH}^{(N)}\) under the joint limit \(N\to\infty\) and \(K=K_N\to\infty\); the regularity conditions and an outline of the proof are given in Appendix~\ref{app:thm1}, and the proof in full in Supplementary Section~\ref{app:proof_main} of the Supplementary Material.

\begin{theorem}[Consistency and asymptotic normality]
\label{thm:main}
Under conditions \textup{(C1)}--\textup{(C6)} in Appendix~\ref{app:thm1}, the following hold as \(N\to\infty\) with \(K=K_N\to\infty\).
\begin{enumerate}[label=\textup{(\roman*)}, leftmargin=*, itemsep=3pt]
\item \textbf{Consistency.} \(\hat\rho_{CH}^{(N)}-\rho_{CH}^{(N)}\xrightarrow{p}0\).

\item \textbf{Asymptotic normality.}
\[\sqrt{NK_N}\,\big(\hat\rho_{CH}^{(N)}-\rho_{CH}^{(N)}\big)\xrightarrow{d}\mathcal{N}(0,\sigma^{*2}),
\qquad
\sigma^{*2}:=\lim_{N\to\infty}K_N\sigma_N^2,\]
where,
\begin{equation}
\sigma_N^2=\sum_{k=1}^K \nabla_k T_N(\mathbf b)^\top\,\Omega_k\,\nabla_k T_N(\mathbf b).
\label{eq:exact_var}
\end{equation}

\item \textbf{Variance estimation.} The plug-in estimator \(\hat\sigma_N^2\), obtained by evaluating \(\nabla_k T_N\) and \(\Omega_k\) at \(\hat{\mathbf b}\), satisfies \(K_N\hat\sigma_N^2\xrightarrow{p}\sigma^{*2}\). Hence \(\sqrt N(\hat\rho_{CH}^{(N)}-\rho_{CH}^{(N)})/\hat\sigma_N\xrightarrow{d}\mathcal N(0,1)\), and \(\hat\rho_{CH}^{(N)}\pm z_{1-\alpha/2}\hat\sigma_N/\sqrt N\) has asymptotic coverage \(1-\alpha\).
\end{enumerate}
\end{theorem}

\begin{remark}
We use the results of \cref{thm:main} for hypothesis testing and confidence interval construction for \(\rho_{CH}^{(N)}\). The substantive assumptions are Conditions (C1) and (C3), which require \(K_N = o(N)\) and a fixed lower bound on instrument strength, \(|\beta_{X,k}| \geq c_x > 0\) uniformly in \(k\). In Supplementary Section~\ref{app:proof_main} (Lemma~\ref{lem:weak_instrument_consistency}), we further sketch the consistency of the estimator under a weak-instrument regime in which \(|\beta_{X,k}|\) is allowed to vanish at a rate \(c_N^\ast\), provided \(\sqrt{N}\, c_N^\ast / K^{1/4} \to \infty\). 
\end{remark}

\begin{remark}[Fixed-weight variance]
The plug-in estimator \(\hat\sigma_N^2\) differentiates \(T_N\) through all of its data dependence, including the precision weights \(\hat w_k\). A simpler alternative treats the weights as fixed and differentiates only through the Wald ratios, giving a closed-form expression \(\sigma_{N,\mathrm{fix}}^2\), derived in \cref{app:proof_main}. Its plug-in estimator evaluates that expression at the sample associations, weights, and moments,
\[\hat\sigma_{N,\mathrm{fix}}^2
=\frac{1}{\hat\tau_1^2\hat\tau_2^2}\sum_{k=1}^{K}\Big(\frac{\hat w_k}{\hat\beta_{X,k}}\Big)^2
\Big[\big(\hat r_{1,k}\hat{\tilde D}_{2,k}+\hat r_{2,k}\hat{\tilde D}_{1,k}\big)^2
+\frac{\hat{\tilde D}_{2,k}^2}{\kappa_1}+\frac{\hat{\tilde D}_{1,k}^2}{\kappa_2}
+\frac{2\hat{\rho}_{s}\,\hat{\tilde D}_{1,k}\hat{\tilde D}_{2,k}}{\sqrt{\kappa_1\kappa_2}}\Big],\]
with \(\hat{\tilde D}_{l,k}=\Delta_{l,k}-\hat\rho_{CH}^{(N)}(\hat\tau_l/\hat\tau_{3-l})\Delta_{3-l,k}\), \(\hat r_{l,k}=\hat\beta_{Y_l,k}/\hat\beta_{X,k}\), \(\kappa_l=N_l/N\), and \(\hat{\rho}_{s}\) an estimate of \(\rho_{s}\), the correlation between the sampling errors of the two outcome GWAS induced by shared individuals (\(\rho_{s}=0\) for disjoint samples). This alternative is exact when the weights are specified externally rather than estimated from the data; when the weights are estimated, it omits a first-order weight-estimation term and therefore serves as an approximation, which in our simulations was consistently conservative relative to \(\hat\sigma_N^2\). We accordingly use the full-gradient estimator \eqref{eq:exact_var} by default and reserve \(\sigma_{N,\mathrm{fix}}^2\) for user-specified fixed weights.
\end{remark}

\begin{remark}
We also note that \(\rho_{CH}^{(N)}\) is itself a random variable: it is defined conditionally on the selected instrument set \(\mathcal{K}_N\), which in turn depends on the GWAS noise. Under additional assumptions on the distribution of SNP effects and the instrument-selection process, \(\rho_{CH}^{(N)}\) can be shown to converge, at a suitable rate, to a genome-wide coheterogeneity parameter \(\rho_{CH}\) defined over all instruments (see \cref{thm:genome_wide_limit} in \cref{app:proof_genome_wide_limit}). In principle, one can also develop an asymptotic theory for \(\rho_{CH}^{(N)}\) centered around \(\rho_{CH}\) under additional conditions, but because MR analyses are typically conducted on a fixed set of selected SNPs, we recommend performing conditional coheterogeneity inference on the selected \(\mathcal{K}\) and treating \(\rho_{CH}^{(N)}\) as the target estimand.
\end{remark}

Our original motivation for the coheterogeneity statistic is to detect pairs of outcome traits that share a large fraction of underlying valid instruments. The next lemma shows that, under a polygenic trait architecture, \(\rho_{CH}^{(N)}\) is expected to converge to a limiting parameter that is monotone in \(D_{\mathrm{ov}}\).

\begin{lemma}
\label{lem:dov_scaling}
Consider a polygenic architecture for complex traits where a proportion $(\pi)$ of the selected instruments is valid, i.e, those that affect $Y_1$ and $Y_2$ only through $X$, and the remaining fraction, $\bar{\pi}=(1-\pi)$, of the instruments can be further fractionated into $q_0=D_{\mathrm{ov}}$, $q_1$ and $q_2$, with $q_0+q_1+q_2=1$, depending on whether the invalidity of the instruments arises through a confounder that is shared,  $Y_1$-specific or $Y_2$-specific, respectively.  Then, under a random-effect model for effect-sizes for SNPs (see \cref{sec:simulation}) and regularity conditions listed in
Appendix~\ref{subsec:scale}, as \(N\to\infty\), one can show 
\begin{equation}
   \rho_{CH}^{(N)}
\xrightarrow{p} \rho_{CH}=
\frac{
\bar\pi q_{0}(ab)}{
\sqrt{
\left[
c^2+
\bar\pi(q_0+q_1)a^2
\right]
\left[d^2+
\bar\pi(q_0+q_2)b^2
\right]}},
\label{eq:rho_dov_main}
\end{equation}
where $a$, $b$, $c$, and $d$ are constants defined by a combination of variance component parameters of the
random-effect model and effect-sizes for the confounders on the outcome traits. As a result, one can
further show that $\rho_{CH}$ increases or decreases monotonically as a function of $q_0$ depending on whether the
shared confounder affects $Y_1$ and $Y_2$ in the same or opposite directions, respectively.
\end{lemma}

\subsection{IB-Methods}

\subsubsection{IB-Mode}

Next, we consider extending existing MR methods to utilize secondary outcome traits. We first consider extending MR-Mode \citep{hartwig2017robust}, which is one of the most robust methods to handle violation of the InSIDE assumption, although it can suffer power loss compared to alternative methods \citep{qi2021comprehensive, xue2021constrained}. The method requires the zero-modal pleiotropy assumption, which implies that the distribution of underlying horizontal pleiotropic effects have a unique mode at zero corresponding to the set of underlying valid instruments. Intuitively, the ratio-parameters across two traits will also have a unique mode for the set of shared valid instruments, and when there is a large overlap of valid instruments across the two traits, bivariate estimation of mode will be more efficient than its univariate counterpart (see \cref{graph_summ}). Based on this idea, we propose to obtain a joint mode estimator of the two sets of ratio estimates by maximizing the weighted bivariate kernel density estimation (KDE).  For each SNP $k$, denote the bivariate vector of two ratio-estimates $\hat{\boldsymbol{\theta}}^{(k)} = (\hat{\beta}_{Y_1,k}/\hat{\beta}_{X,k}, \; \hat{\beta}_{Y_2,k}/\hat{\beta}_{X,k})^\top$. The causal effect vector $\boldsymbol{\theta} = (\theta_1,\theta_2)^\top$ is then estimated by $\hat{\boldsymbol{\theta}}^{*}$, the maximizer of the weighted bivariate KDE:
\[\hat{f}(\boldsymbol{\theta}; H,w)
= \sum_{k=1}^K w_k \, |H|^{-1/2}(2\pi)^{-1}
\exp\!\left[-\tfrac{1}{2}\big(\boldsymbol{\theta}-\hat{\boldsymbol{\theta}}^{(k)}\big)^\top
    H^{-1}\big(\boldsymbol{\theta}-\hat{\boldsymbol{\theta}}^{(k)}\big)\right],\]
where $H$ is a positive-definite bandwidth matrix $2\times 2$ and weights ($w_k \propto 1/(\sigma_{1,k}\sigma_{2,k}))$. We initialize $H$ by first defining $H_0 = \operatorname{diag}(\sigma_{\theta_1}^2 K^{-1/3}, \; \sigma_{\theta_2}^2 K^{-1/3})$ obtained from Silverman's rule for $d=2$, and set $H = \phi H_0$ with the tuning parameter $\phi > 0$, and
\(\sigma_{\theta_l}^2\) for \(l=1,2\) is the empirical variance of the \(K\) ratio estimates corresponding to trait $l$. The default parameter for standard MR-Mode is $\phi=1$ and we found that the same value also provides strict type-I error rate control for IB-Mode for a wide variety of settings (more details in simulations).  

For inference, IB-Mode adopts the same parametric bootstrap as MR-Mode, modified to respect the dependence between the two outcomes. In each bootstrap replicate, the bivariate ratio vector $\hat{\boldsymbol\theta}^{(k)}=(\hat\theta_{1,k},\hat\theta_{2,k})^\top$ at each instrument $k$ is redrawn jointly from a bivariate normal centred at its observed value, $\hat{\boldsymbol\theta}^{(k)}_{(b)}\sim \mathcal{N}_2(\hat{\boldsymbol\theta}^{(k)},\widehat{\boldsymbol\Sigma}_k)$, where $\widehat{\boldsymbol\Sigma}_k$ has diagonal entries equal to the delta-method variances $\hat\sigma_{1,k}^2$ and $\hat\sigma_{2,k}^2$ and off-diagonal entry equal to the cross-trait sampling covariance $\hat\sigma_{12,k}$ (estimated as in \cref{sec:cov_est}); the weighted KDE is then re-fitted to obtain the modal estimate. The standard error of each effect estimate is the median absolute deviation of its bootstrap replicates; as in MR-Mode \citep{hartwig2017robust}, confidence intervals use a normal reference and $p$-values a $t$ reference with $K-1$ degrees of freedom. When the two outcome GWAS are estimated on non-overlapping samples, $\hat\sigma_{12,k}=0$ and the joint draw reduces to the independent redrawing used by MR-Mode; this is the case in our simulation study, so its results are unaffected by this distinction. IB-Mode therefore generalizes MR-Mode by leveraging auxiliary outcomes through a weighted multivariate KDE, while retaining robustness to pleiotropy via modal estimation.

\subsubsection{IB-PRESSO}

MR-PRESSO is an MR method that implements a series of outlier-detection steps to remove potentially invalid instruments \citep{verbanck2018detection}. Intuitively, since valid/invalid instruments are likely to be shared across related outcomes, the outlier detection can be improved by their joint analysis. For each SNP $k = 1, \ldots, K$ and trait index $l \in \{1, 2\}$, let
$\hat\theta_{\mathrm{IVW},l}^{(-k)}$ denote the IVW estimator for outcome $Y_l$
computed with the $k$-th SNP excluded, and let
\(\hat\beta_{Y_l,-k} := \hat\theta_{\mathrm{IVW},l}^{(-k)} \times \hat\beta_{X,k}\)
denote the predicted effect of the $k$-th SNP on $Y_l$ implied by the
leave-one-out IVW fit. \rev{The outlier detection below operates on the effect estimates
rather than on their ratios. Let $\sigma_{X,k}^2:=\Var(\hat\beta_{X,k})$ and
$\sigma_{Y_l,k}^2:=\Var(\hat\beta_{Y_l,k})$ denote the corresponding sampling variances,
obtained directly from the reported GWAS standard errors; these are distinct from
$\sigma_{l,k}^2$, which is defined on the ratio scale.}

Define the bivariate residual at SNP $k$ as
\[\mathbf{r}^{(k)} = 
\begin{pmatrix}
  \hat\beta_{Y_1,k} - \hat\beta_{Y_1,-k} \\
  \hat\beta_{Y_2,k} - \hat\beta_{Y_2,-k}
\end{pmatrix}.\]

We propose a joint outlier-detection statistic based on the squared Mahalanobis
distance at each SNP,
$
D^2_k = (\mathbf{r}^{(k)})^\top (\Sigma^{r}_k)^{-1} \mathbf{r}^{(k)},
$
\rev{where $\Sigma^{r}_k$ has diagonal entries equal to the residual variances $\sigma_{r_l,k}^2$, $l\in\{1,2\}$, and off-diagonal entry equal to the residual covariance $\sigma_{r_{12},k}$ (given in \cref{sec:cov_est}). Under the null hypothesis that the $k$-th instrument is valid, $\mathbf{r}^{(k)}$ has mean zero and covariance $\Sigma^{r}_k$, so that $D^2_k\sim\chi^2_2$.}  This extends naturally to the global statistic
$
\mathrm{RSS}_{\mathrm{IB\text{-}PRESSO}} 
= 
\sum_{k=1}^K D^2_k,
$
whose significance is assessed against a parametric-bootstrap null distribution, as in the original PRESSO. \rev{When global heterogeneity is detected, individual outliers are flagged by a Bonferroni-corrected test: instrument $k$ is removed if $K p_k \le \alpha$, where $p_k = \Pr(\chi^2_2 > D^2_k)$.}

After the outlying SNPs are removed, the causal effect on the primary outcome $Y_1$ is re-estimated exactly as in MR-PRESSO: by the inverse-variance weighted regression of $\hat\beta_{Y_1,k}$ on $\hat\beta_{X,k}$ through the origin over the retained instruments, with weights $1/\sigma_{Y_1,k}^2$, and its standard error is taken as the standard error of this weighted-regression slope. The auxiliary outcome $Y_2$ thus enters only the joint outlier detection (through the bivariate residual $\mathbf{r}^{(k)}$); the final causal estimate and its standard error pertain to $Y_1$.

The suggested workflow for an instrument-borrowing implementation is summarised as \cref{algo} in the Supplementary Material.

%% file: 3results.tex
\input{4maintextFigs}

\section{Simulation studies}
\label{sec:simulation}

We designed simulations to evaluate the robustness of Mendelian Randomization (MR) estimators under varied genetic architecture, pleiotropy, and trait overlap. Our framework generalizes \citep{qi2021comprehensive} by jointly simulating summary-level GWAS data for an exposure \(X\) and two outcomes, \(Y_1\) and \(Y_2\), incorporating unmeasured confounding with overlapping structure. For computational efficiency, similar to \citep{qi2021comprehensive}, we directly simulate summary statistics throughout the genome instead of individual-level data, with parameters chosen to align with realistic heritability and effect-size distributions for the corresponding traits.

We simulate genome-wide data on \(M=200{,}000\) independent SNPs. Unmeasured confounding is represented through three latent factors: a shared confounder \(U_0\) that affects both outcomes, and two outcome-specific confounders \(U_1\) and \(U_2\) that affect only \(Y_1\) and \(Y_2\), respectively; see \cref{fig:simulation_structure}. For SNP \(m\), let \(\gamma_m\) denote its direct effect on exposure \(X\), and let \(\phi_{0m}\), \(\phi_{1m}\), and \(\phi_{2m}\) denote its direct effects on \(U_0\), \(U_1\), and \(U_2\), respectively. In addition, let \(\delta_{1m}\) and \(\delta_{2m}\) denote the intrinsic direct effects of SNP \(m\) on \(Y_1\) and \(Y_2\), respectively, not mediated by \(X\) or the modelled confounders.

The total effects of SNP on the exposure and outcome traits are generated as
\begin{equation}
\begin{aligned}
\beta_{X,m}   
&= 
\gamma_m 
+ \theta_{U_0\to X}\phi_{0m} 
+ \theta_{U_1\to X}\phi_{1m} 
+ \theta_{U_2\to X}\phi_{2m}, \\
\beta_{Y_1,m} 
&= 
\delta_{1m} 
+ \theta_{U_0\to Y_1}\phi_{0m} 
+ \theta_{U_1\to Y_1}\phi_{1m} 
+ \theta_{X\to Y_1}\beta_{X,m}, \\
\beta_{Y_2,m} 
&= 
\delta_{2m} 
+ \theta_{U_0\to Y_2}\phi_{0m} 
+ \theta_{U_2\to Y_2}\phi_{2m} 
+ \theta_{X\to Y_2}\beta_{X,m}.
\end{aligned}
\label{sim_frame}
\end{equation}
Here, \(\theta_{U\to X}\) and \(\theta_{U\to Y}\) denote the regression effects of confounders on exposure and outcomes, and \(\theta_{X\to Y_1}\) and \(\theta_{X\to Y_2}\) are the true causal effects of \(X\) on \(Y_1\) and \(Y_2\), respectively. Without loss of generality, we set
\[\theta_{U_0\to X}
=
\theta_{U_1\to X}
=
\theta_{U_2\to X}
\equiv
\theta_{U\to X},
\qquad
\theta_{U_0\to Y_l}
=
\theta_{U_l\to Y_l}
\equiv
\theta_{U\to Y_l},
\quad l\in\{1,2\}.\]

We first assume that a fraction $p$ of all SNPs is associated with \(X\), that is, \(\beta_{X,m}\neq0\), and hence forms the underlying pool of potential instruments. This pool is further partitioned into valid instruments, with fraction \(\pi\), and invalid instruments, with fraction $\bar\pi=1-\pi$. For valid instruments, we assume \(\gamma_m\neq0\) and \(\phi_{\ell m}=0\) for \(\ell=0,1,2\), with \(\gamma_m\sim \mathcal{N}(0,\sigma_X^2)\). For invalid instruments, we set the one of \(\phi_{0m}\), \(\phi_{1m}\), or \(\phi_{2m}\) to be nonzero and generate \(\gamma_m\sim \mathcal{N}(0,\widetilde\sigma_X^2)\), where \(\widetilde\sigma_X^2=\sigma_X^2-\theta_{U\to X}^2\sigma_U^2\), so that valid and invalid instruments explain comparable amounts of exposure variation. Invalid instruments are further partitioned according to whether they are associated with \(U_0\), \(U_1\), or \(U_2\), with corresponding fractions \(q_0\), \(q_1\), and \(q_2=1-q_0-q_1\). The nonzero confounder effect, one of \(\phi_{0m}\), \(\phi_{1m}\), or \(\phi_{2m}\), is generated from a mean-zero normal distribution with common variance \(\sigma_U^2\). In this simulation model, the overlap measure for invalid instruments is $
D_{\mathrm{ov}}=q_0.
$

Finally, for each outcome \(Y_l\), the intrinsic direct outcome effect \(\delta_{lm}\) is allowed to be nonzero for SNPs that are invalid instruments for \(X\), or for SNPs that are not associated with \(X\) but have direct effects on \(Y_l\). We generate
\(\delta_{lm}\sim \mathcal{N}(\mu_{Y_l},\sigma_{Y_l}^2)\),
allowing \(\mu_{Y_l}\neq0\) to represent direct directional pleiotropy. SNPs that do not belong to any of the above categories are assigned no effect on \(X\), \(Y_1\), or \(Y_2\). For each SNP, observed GWAS summary statistics are generated as
\[\hat\beta_{Z,m}
=
\beta_{Z,m}
+
\epsilon_{Z,m},
\qquad
\epsilon_{Z,m}\sim \mathcal{N}(0,N_Z^{-1}),
\qquad
Z\in\{X,Y_1,Y_2\},\]
with error variances determined by trait-specific sample sizes. The noise terms \(\epsilon_{Z,m}\) are drawn independently across the three traits \(Z\in\{X,Y_1,Y_2\}\), so that the exposure and the two outcome GWAS are assumed to have no overlapping samples. Candidate MR instruments are defined as SNPs that reach genome-wide significance for exposure, \(p<5\times10^{-8}\), reflecting real GWAS discovery. The resulting set of instruments contains both valid and invalid SNPs according to the simulated mixture. Further implementation details are provided in Supplementary Section~\ref{sec:sim_settings}.

\begin{figure}[!htbp]
    \centering
    \includegraphics[width=1\linewidth]{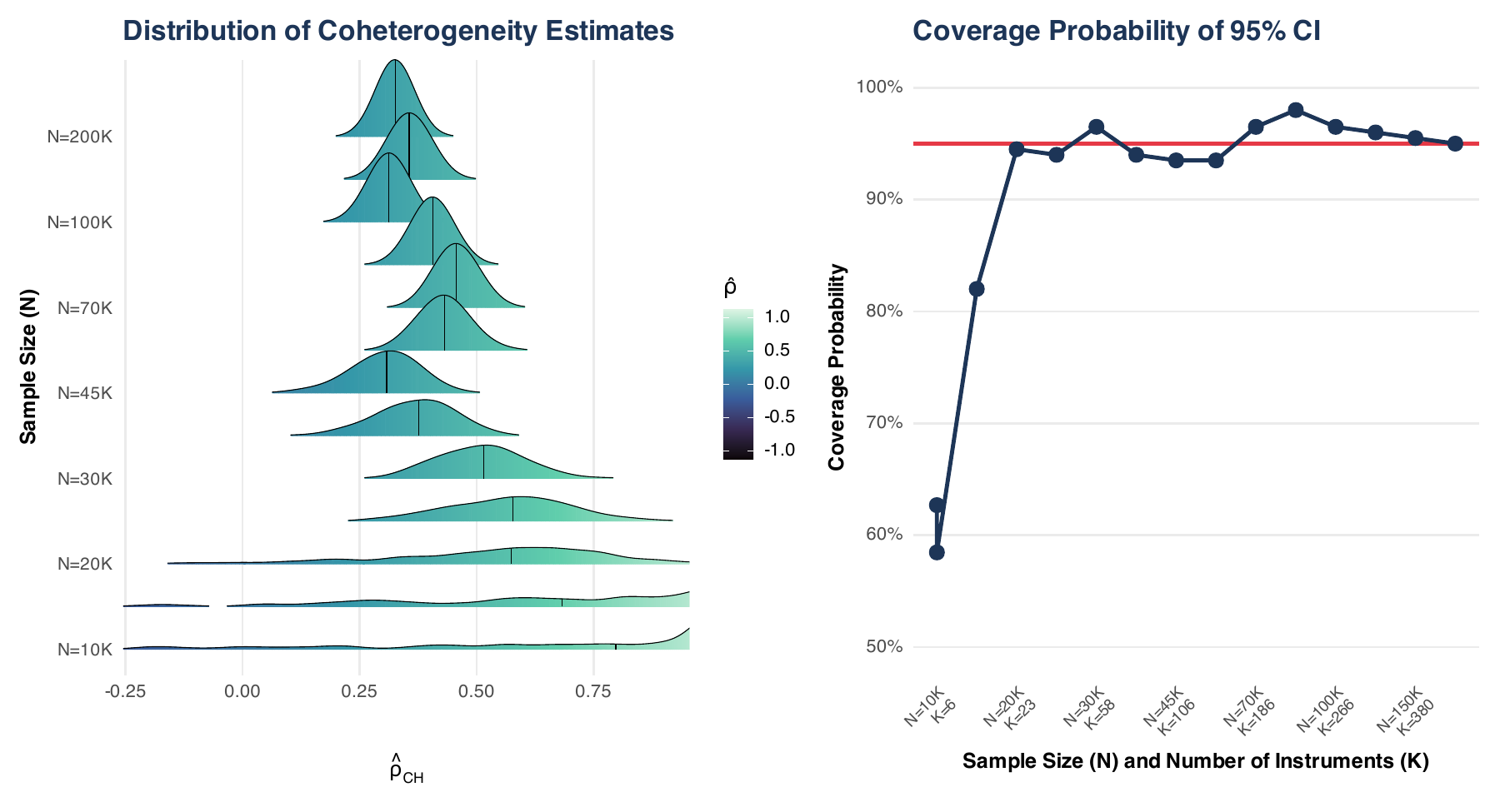}
    \caption{\textbf{Sampling distribution of the coheterogeneity statistics and coverage of 95\% confidence interval.}
\textbf{Left:} Ridge density plots of $\hat{\rho}_{CH}^{(N)}$ across replications with different sample sizes ($N$). \textbf{Right:} Empirical coverage of nominal 95\% confidence intervals for $\hat{\rho}_{CH}^{(N)}$ across sample sizes $N$
(with the corresponding mean number of selected instruments $K$). The horizontal red line marks 95\% nominal coverage.}
\label{fig:cohet_coverage}
\end{figure}

Simulation studies show that the confidence intervals for ${\rho}_{CH}^{(N)}$ based on established asymptotic theory maintain the probability of coverage at the desired level for a large sample size (N $>$ 20K) and the associated number of instruments is also large (K $>$ 20) (\cref{fig:cohet_coverage}).  For larger sample sizes,  the distribution of $\hat{\rho}_{CH}^{(N)}$ follows the normal distribution around ${\rho}_{CH}^{(N)}$ and ${\rho}_{CH}^{(N)}$ has random variation around ${\rho}_{CH}$. We also observe that $\hat{\rho}_{CH}^{(N)}$ follows $D_{\mathrm{ov}}$ in a wide range of scenarios. The relationship between $\hat{\rho}_{CH}^{(N)}$ and $D_{\mathrm{ov}}$ is stronger when the proportion of invalid instruments is larger relative to the total number of instruments (\cref{fig:cohet-Dcov}
). In addition, the trend is also stronger with an increasing sample size for the underlying GWAS. In the absence of common confounders, $\hat{\rho}_{CH}^{(N)}$ can track overlaps in invalid instruments that exhibit directional pleiotropy across traits (\cref{fig:cohet-Dirpleit}). When there are no shared confounders or directional pleiotropy, the $\hat{\rho}_{CH}^{(N)}$ statistics remain close to zero irrespective of the sample size of the underlying GWAS and the proportion of invalid instruments.

Both MR-Mode and IB-Mode involve a key underlying tuning parameter $\phi$, which controls the degree of smoothing in the estimation of the underlying density functions. In general, as $\phi$ increases, both the type-I error and the power of the methods are expected to increase. A default value of $\phi = 1$ has been recommended for standard MR-Mode \citep{yavorska2017mendelianrandomization, hartwig2017robust}. To objectively compare IB-Mode and MR-Mode, we varied the value of $\phi$ over 25 equally spaced values (on logarithmic scale) between 0.1 and 10, and  for each method we selected the largest \(\phi\) that maintains type-I error \(\leq 5\%\). Using this approach, we found that IB-Mode can improve the power of MR-Mode in a wide range of scenarios, including varying sample sizes and the proportion of invalid instruments (\cref{fig:power_comparison1}). The greatest driver of power gain, as expected, is the degree of overlap in the underlying invalid instruments ($D_{\mathrm{ov}}$). In addition, the power gain also depends on the sample size for the GWAS of the auxiliary trait, but not on the underlying causal effect of $X$ for the auxiliary trait  (\cref{fig:power_combined}).  For smaller sample sizes ($N=50$K), we observe that the IB-Mode can sometimes have a modest loss of power compared to the MR-Mode. 

In our simulations, we observed that the choice of $\phi=1$ consistently maintained the type-I error of IB-Mode at the desired level in a wide variety of settings (\cref{fig:type-1-error-calibration}). Although we found some higher values of $\phi$, for example $\phi=2$, also often maintained type-I error in most situations, in the future, we will implement IB-Mode with the more conservative choice of $\phi=1$. We observe that the type-I error of the method for the target trait remains valid even when the exposure has a non-zero causal effect on the secondary trait (\cref{fig:robust_aux_trait}). In other words, instrument borrowing, while leading to an improvement in power, does not result in ``slippage'' of the causal effect of exposure from the secondary to the primary trait. 

\begin{figure}[!htbp]
    \centering
    \includegraphics[width=\textwidth]{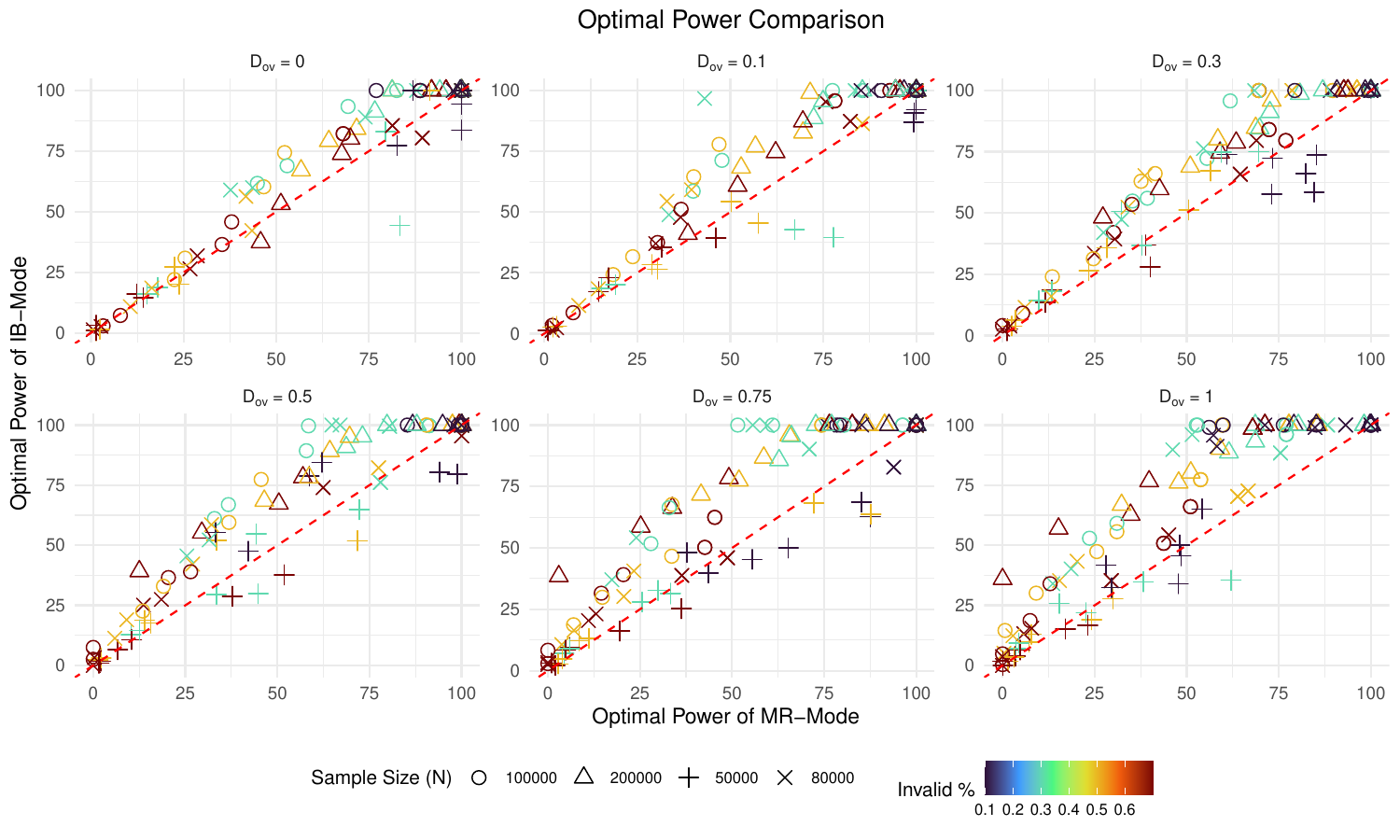}
    \caption{
        \textit{\textbf{Power comparison between MR-Mode and IB-Mode methods across diverse simulation scenarios.} For each method and scenario setting, the largest value for the tuning parameter \(\phi\) that maintains the Type-I error rate at or below 5\% is selected using simulation experiments under null. The statistical power of the two methods at the respective optimal \(\phi\) are then evaluated through simulations under alternative hypotheses. Simulation scenarios vary by sample size, the proportion of invalid instruments, and the level of overlap (\(D_{\mathrm{ov}}\)) for invalid instruments across two outcome traits. In the plot, the color gradient reflects the proportion of invalid instruments, while different point shapes correspond to varying sample sizes. The red dashed identity line ($y = x$) marks equal performance for both methods; points above this line indicate higher power for IB-Mode, whereas points below favour MR-Mode.
    }}
    \label{fig:power_comparison1}
\end{figure}

Consistent with previous reports \citep{qi2021comprehensive, slob2020comparison}, our simulation studies confirm that many existing MR methods, including those designed to account for invalid instruments under various assumptions, can lead to biased inference, i.e., inflated type-I error (\cref{fig:panel2x2_simulation}) or lower coverage of confidence intervals (\cref{fig:Main_simulation_results_coverage}),  in various scenarios, especially when the InSIDE assumptions are violated. In general, IVW, MRMedian, and MR-PRESSO tend to be more sensitive to the presence of invalid instruments.  \rev{Instrument borrowing nonetheless yields a consistent gain: IB-PRESSO attains a lower type-I error than MR-PRESSO in nearly every setting considered, curtailing the inflation induced by InSIDE violation even where it cannot be eliminated.} MR-Egger, MR-Mix, MR-ConMix and MR-cML lead to substantially improved type-I error rate controls, but even these methods can sometimes have substantially inflated type-I error at larger sample sizes or/and in the presence of a large fraction ( $>$ 0.5) of invalid instruments. The MR-Egger method has a large uncertainty associated with it and as a result has poor power (\cref{fig:panel2x2_simulation}) and large MSE (\cref{figsup:panel2x2_simulation}) in most scenarios. Finally, MR-Mode is the only existing method that can maintain type-I error across all scenarios. IB-Mode also maintains the type-I error in all the same scenarios and consistently leads to improved power (\cref{fig:panel2x2_simulation}) and smaller MSE (\cref{figsup:panel2x2_simulation}) except when the sample size of the underlying GWAS is small.  

In a separate analysis, we compared the power of IB-Mode with $\mathrm{MR}^2$, a Bayesian method that allows the joint analysis of multiple exposures and multiple outcomes simultaneously \citep{zuber2023multi}. We implemented $\mathrm{MR}^2$ using the available package ($\mathrm{MR}^2$) with a single exposure and two outcome traits. For fair comparison, we first conducted simulation studies under the null hypothesis of no exposure effect on the target trait to select the threshold for the underlying posterior inclusion probability that corresponds to a fixed type-I error (0.05 and 0.001). We then simulated data under an alternative to evaluate the power of the proposed method at the selected threshold for inclusion probabilities. The results (\cref{fig:power_comparison_0.001}) show that the power of $\mathrm{MR}^2$ is often much lower than IB-Mode, indicating that $\mathrm{MR}^2$ is not suitable for testing simple causal hypotheses related to a single exposure.

\begin{figure}[!htbp]
    \centering
    \includegraphics[width=0.7\linewidth]{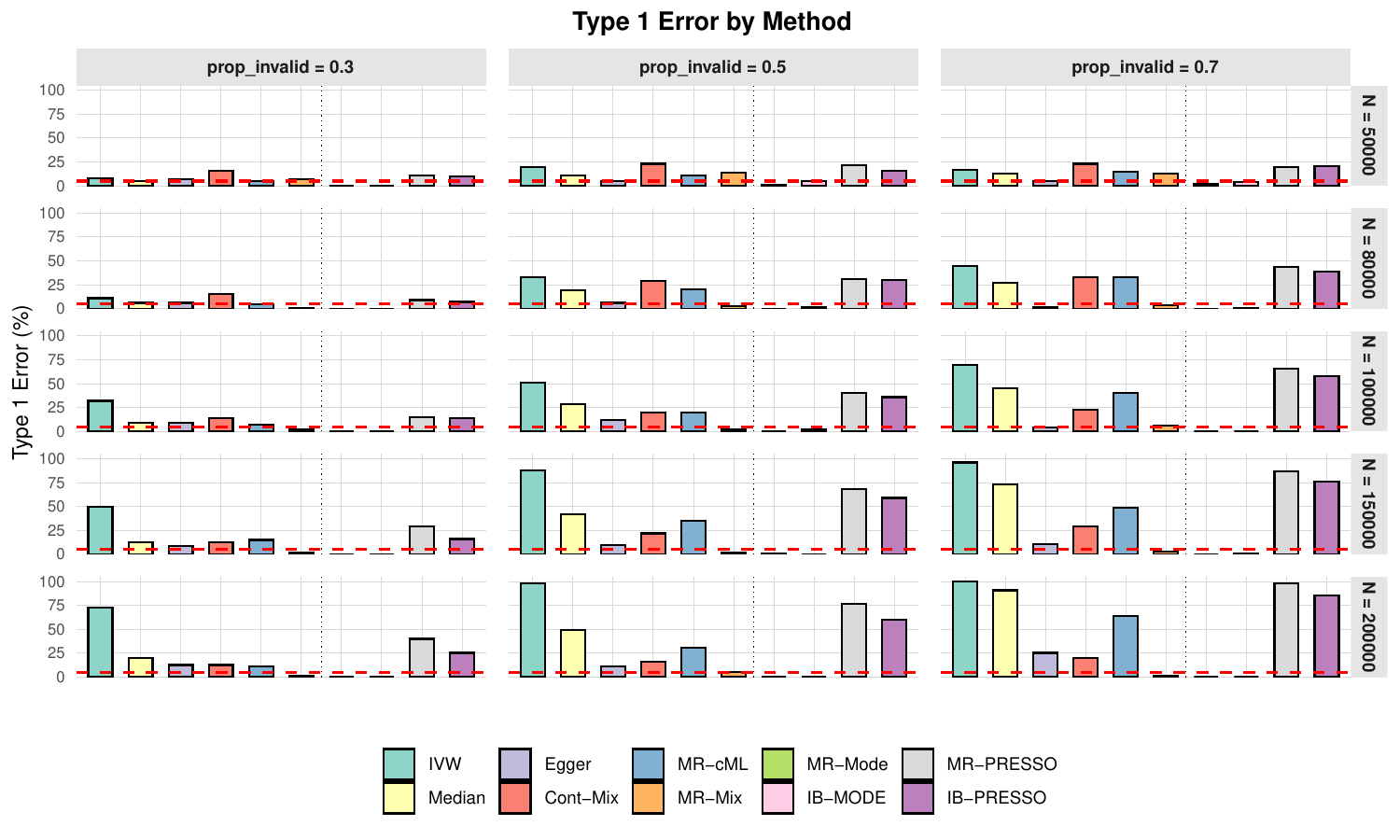}\\[0.5em]
    \includegraphics[width=0.7\linewidth]{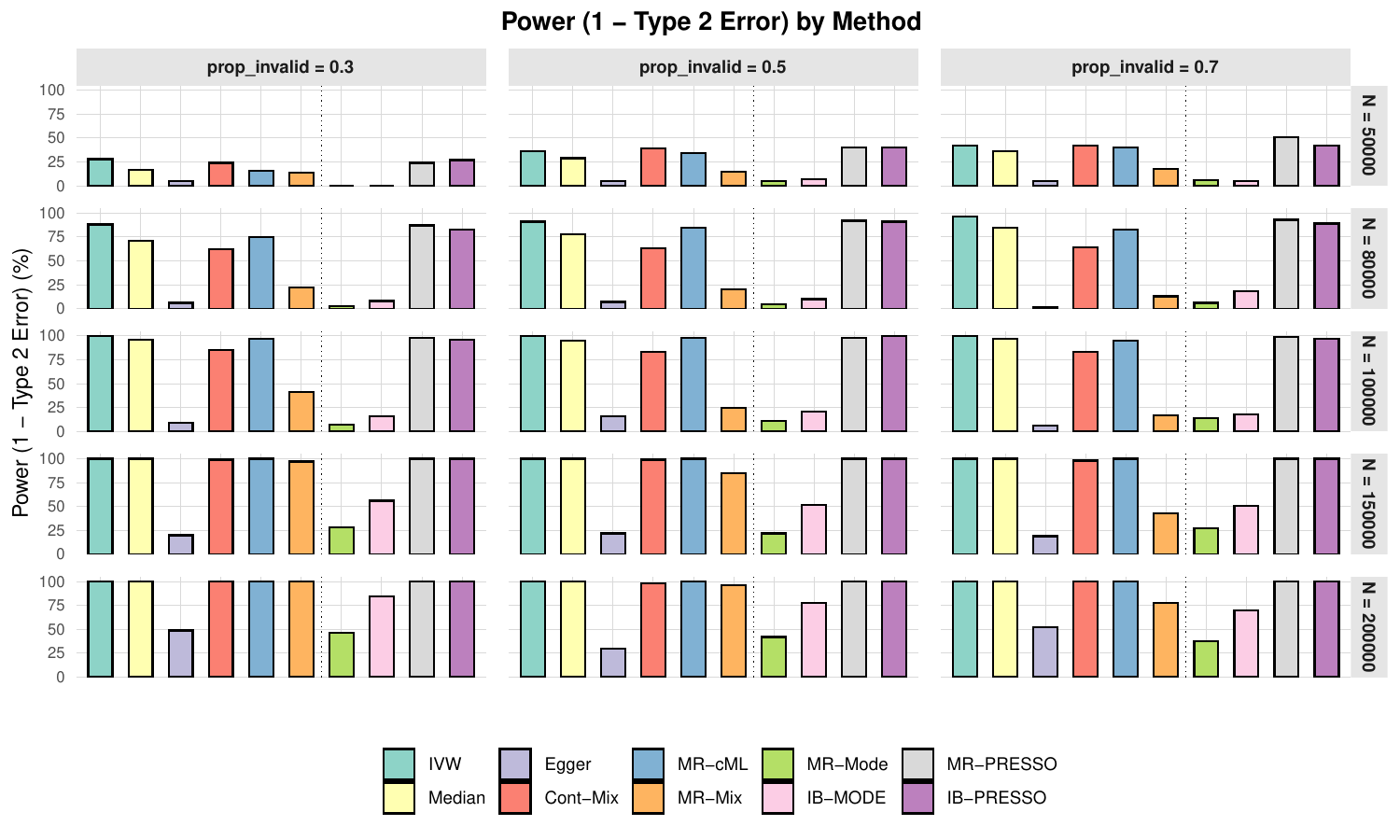}
    
    \caption{
        \textbf{Type-I error and power of different MR methods}. Simulations are conducted under models where invalid instruments arise due to hidden heritable confounders leading to violation of the InSIDE assumption. The overlap in underlying invalid instruments across two outcome traits due to shared confounders is assumed to be $D_{\mathrm{ov}}=0.75$. Additionally, invalid instruments are allowed to have directional pleiotropic effects. The exposure has a causal effect of $\theta_{X\to Y_2}=0.3$ on the secondary (auxiliary) outcome trait.
        \textbf{(Top)} Type-I error when the true causal effect for primary outcome trait is 0.
        \textbf{(Bottom)} Power when the true causal effect for the primary outcome trait is 0.1.
    }
    \label{fig:panel2x2_simulation}
\end{figure}

\section{Real data analysis}

\label{subsec:data_anal}

We first conducted MR analysis using the proposed and existing methods to investigate the causal relationship between cardio-metabolic diseases and various established risk factors. The set of hypotheses we studied included 46 exposure–outcome relationships which have been previously categorized by \citep{morrison2020mendelian} into four categories based on prior evidence: (i) causal (19 pairs), (ii) correlated but with unclear causality or conflicting evidence (15), (iii) implausible or unsupported (10), and (iv) non-causal (2); the full set of exposures and outcomes analysed, including negative controls, is summarized in \cref{tab:traits_summary_compact}. Furthermore, we included vitamin D as an exposure of interest, given the ongoing scientific debate on its potential causal effect on various health outcomes.  

We obtained GWAS summary statistics for relevant outcome traits  from recently available data from the Million Veteran Program (MVP), a large-scale biobank of U.S. military veterans \citep{gaziano2016million}. The analyses were restricted to participants of genetically inferred European ancestry determined by the MVP ancestry pipeline. Specific outcomes included a panel of primarily (but not exclusively) cardio-metabolic traits, selected to include both clinical disease endpoints and quantitative biomarkers. Binary outcomes included coronary artery disease (CAD), type 2 diabetes (T2D), myocardial infarction (MI), hypertension (HTN), stroke, type 1 diabetes (T1D), and asthma. Continuous outcome traits included estimated glomerular filtration rate (eGFR), systolic and diastolic blood pressure (SBP, DBP), lipid levels, etc. A complete list of outcome traits, abbreviations, and other details is provided in \cref{Description}.  Exposures included anthropometric measures (BMI, percentage of body fat, height, birth weight), lipid traits (high-density lipoprotein cholesterol [HDL-C], low-density lipoprotein cholesterol [LDL-C], triglycerides [TG]), blood pressure traits (SBP, DBP), lifestyle and behavioral factors (alcoholic beverages per week, smoking status as regular smoker) and metabolic or biochemical markers (fasting glucose, vitamin D). For each exposure trait, we sought to obtain publicly available summary statistics from the largest meta-analysis  conducted in the European ancestry population but not including any MVP data - the source for our outcome trait GWAS. Data Sources and underlying sample size for all the outcome and exposure GWAS can be found in \cref{Supp-Tab}.

All summary GWAS statistics for exposures and outcomes were aligned to the human reference genome GRCh37 (hg19) and were harmonized for allele orientation, strand alignment, and variant annotation. Variants with a minor allele frequency (MAF $\leq 1\%$), an imputation quality score $\leq 0.8$, or inconsistent allele strand assignments across datasets were excluded. Instrumental variables were identified by LD-based clumping (genome-wide significance $p < 5 \times 10^{-8}$, 1 Mb window, LD threshold $r^2 < 0.1$), retaining the most significant SNP per locus. Harmonization of the data set was performed with the \texttt{TwoSampleMR} package. 

We first examine coheterogeneity statistics ($\hat{\rho}_{CH}^{(N)}$) for all pairs of cardio-metabolic outcome traits (e.g., T2D and CAD) anchored to individual exposures (e.g., BMI)(see \cref{fig:main_multi_panel}A). In all analyses, we adjust for underlying correlations in effect-size estimates for the outcome traits, which are all obtained from MVP (see Methods). We observe a wide-spread presence of statistically significant coheterogeneity across pairs of traits, indicating that the presence of shared confounders is a fairly common phenomenon. However, the pattern is exposure specific. For example, we find highly significant evidence of positive coheterogeneity ($\hat{\rho}_{CH}^{(N)}$) between CAD and T2D when BMI is the underlying exposure, suggesting that there are likely common confounders underlying the relationship between BMI $\rightarrow$ CAD and BMI $\rightarrow$ T2D. However, $\hat{\rho}_{CH}^{(N)}$ was not significant between these two outcome traits when LDL was the underlying exposure. In general, we find more wide-spread shared confounding across cardio-metabolic traits in relation to BMI than LDL. In subsequent analysis, for each outcome trait of interest, we choose the secondary trait that exhibits the largest absolute value of the underlying statistic, $|\hat{\rho}_{CH}^{(N)}|$, but we also perform sensitivity analysis with additional secondary traits.  

\begin{figure}[!htbp]
    \centering
    % The whole composite is scaled as one block: the panels keep their original
    % relative sizes and spacing, and only the overall frame shrinks.  Scaling
    % the panels individually instead leaves gaps between them.
    \resizebox{0.82\textwidth}{!}{%
    \begin{minipage}{\textwidth}
        \centering
        \begin{minipage}[t]{0.40\textwidth}
            \vspace{0pt}
            \begin{overpic}[width=\linewidth]{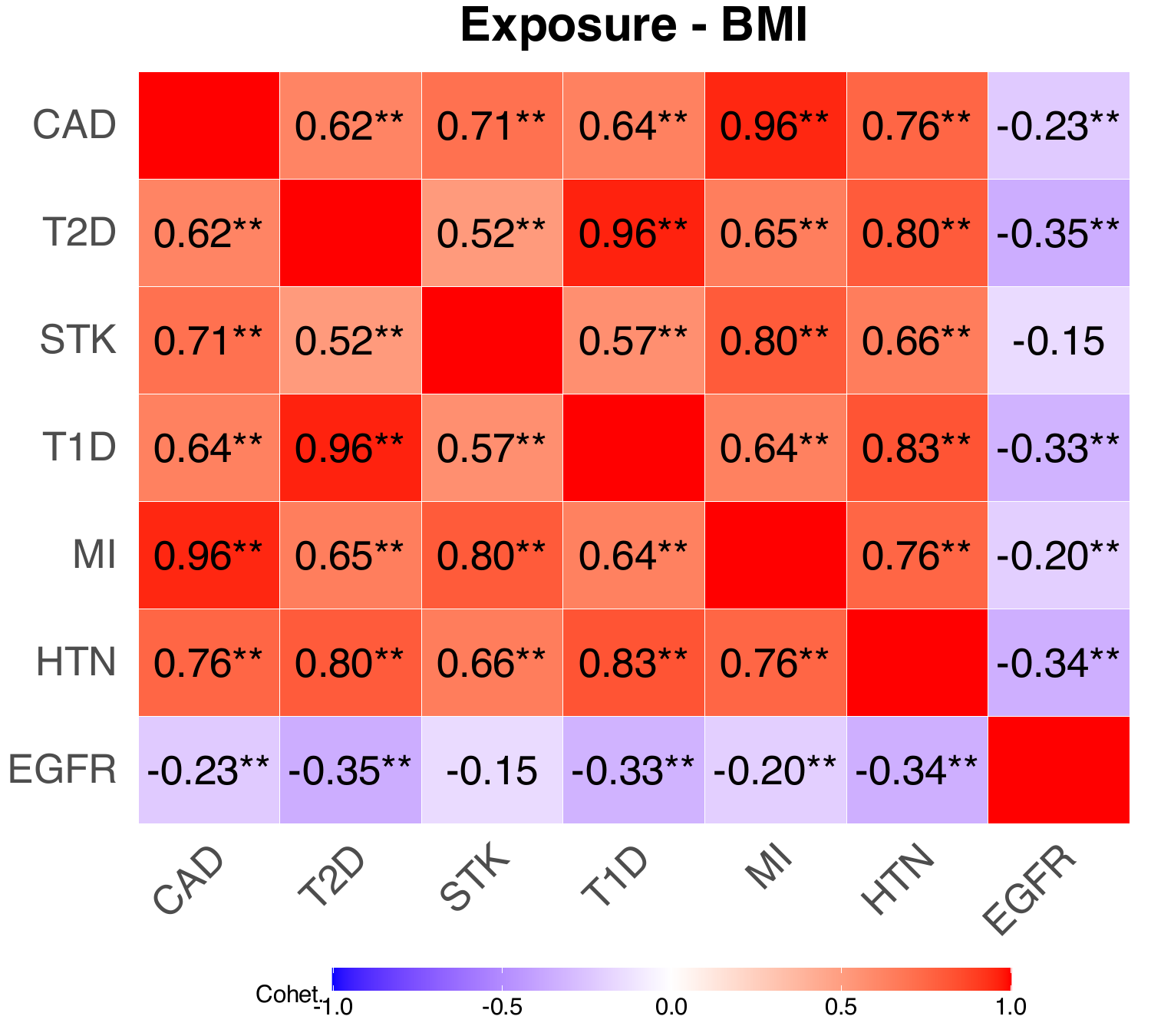}
                \put(2,95){\color{black}\large\textbf{A}}
            \end{overpic}
            \vspace{0.4em}
            \begin{overpic}[width=\linewidth]{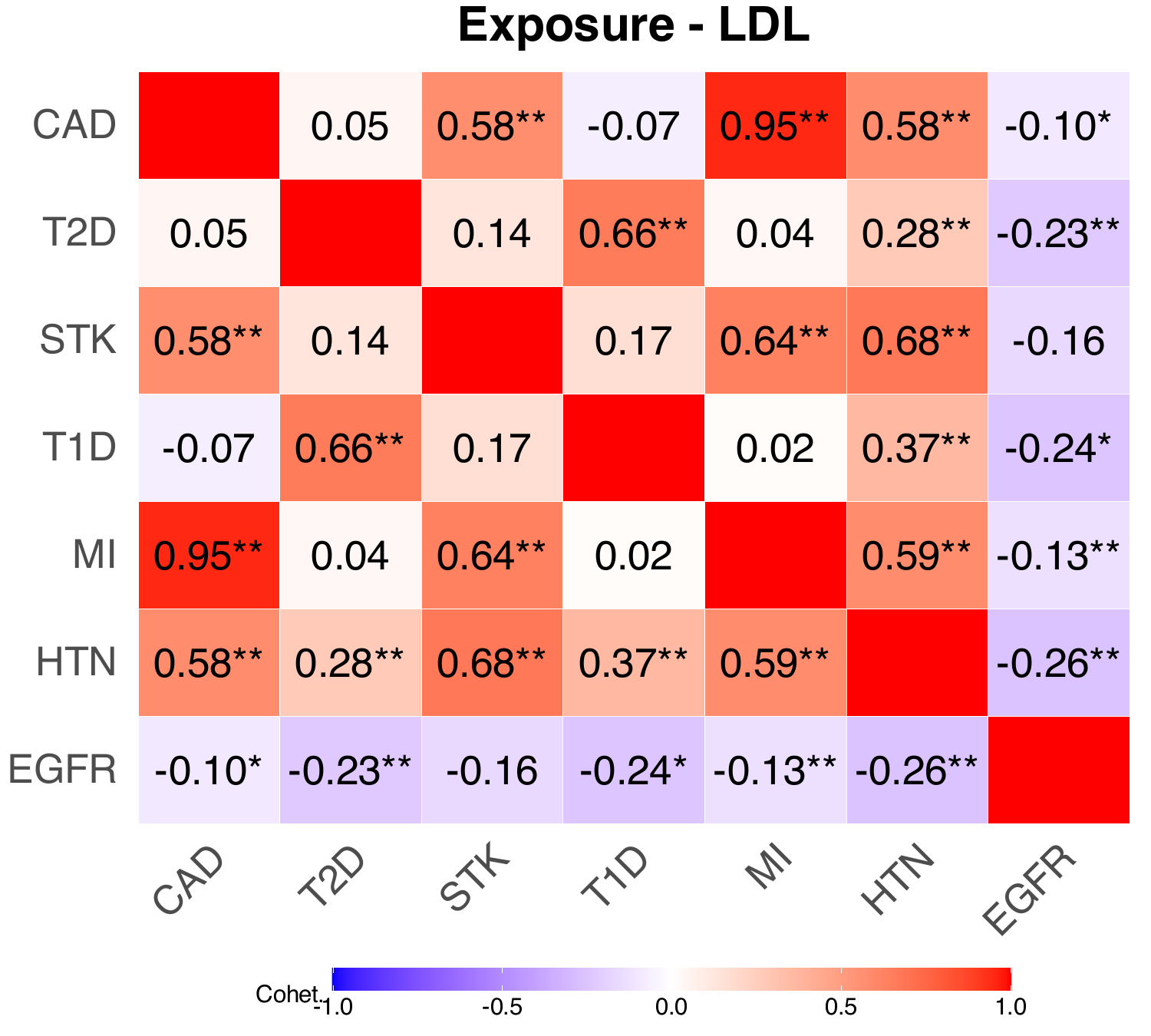}
            \end{overpic}
        \end{minipage}
        \hfill
        \begin{minipage}[t]{0.57\textwidth}
            \vspace{0pt}
            \begin{overpic}[width=\linewidth]{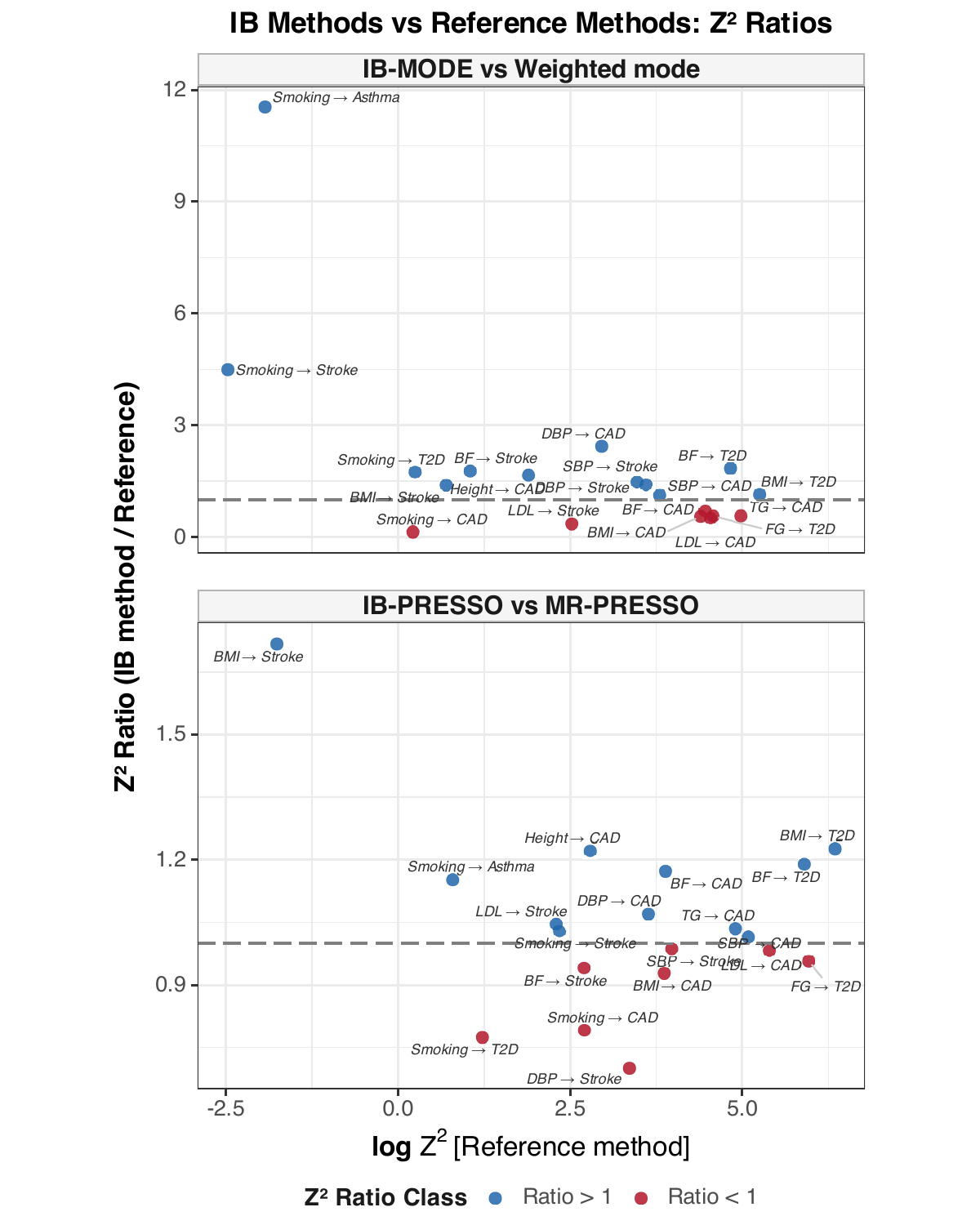}
                \put(2,103){\color{black}\large\textbf{C}}
            \end{overpic}
        \end{minipage}

        \vspace{1.5em}

        \begin{overpic}[width=\linewidth]{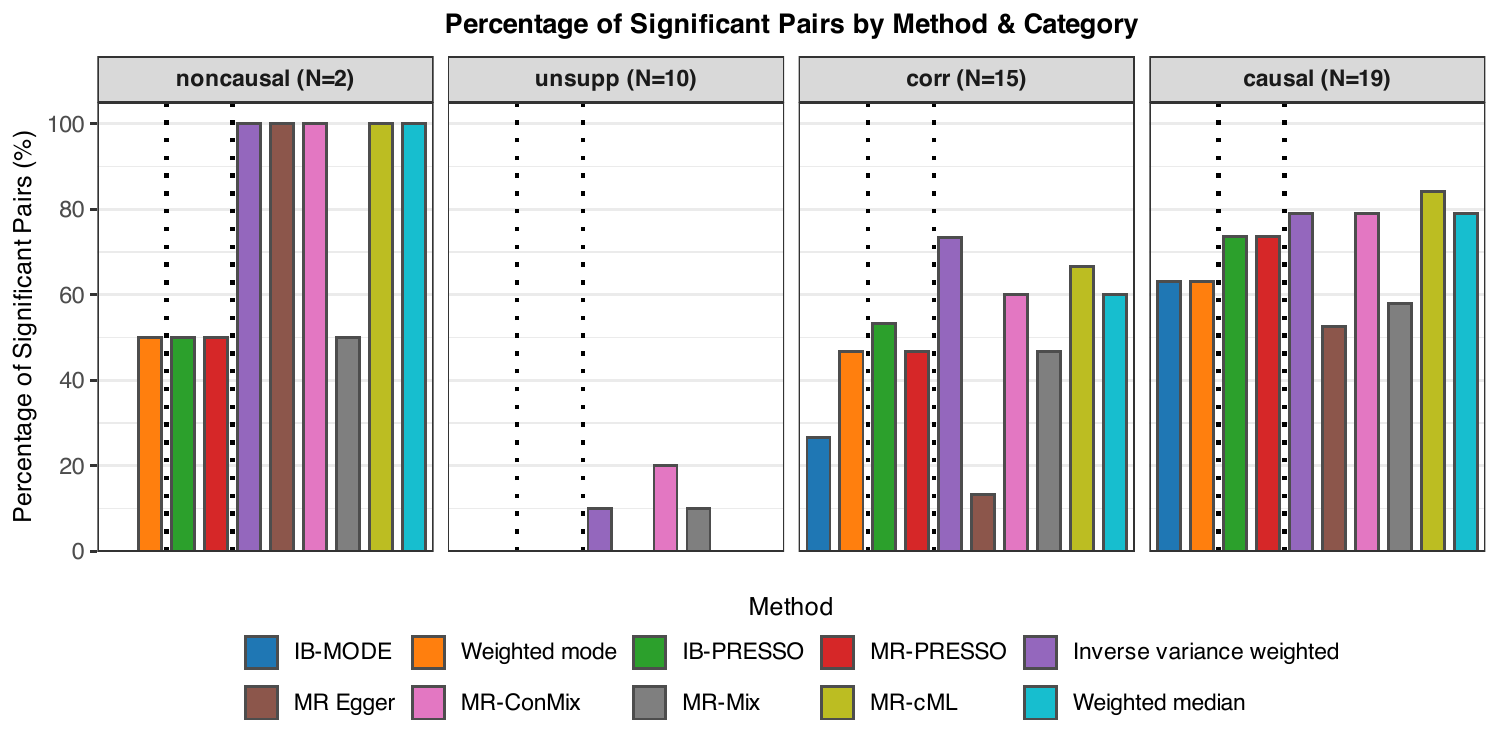}
            \put(2,52){\color{black}\large\textbf{B}}
        \end{overpic}
    \end{minipage}}

    \caption{
        \textbf{Data analysis results for MR analysis involving cardiometabolic traits and 46 exposure-outcome causal hypotheses}
        \textbf{(A)} \textit{Heatmaps representing values for coheterogeneity statistics for two representative exposures (BMI and LDL) across selected outcome traits. Statistically significant (p-value $<$ 0.001) are indicated by **.}
        \textbf{(B)} \textit{Proportion of positive findings (p-value $<$ 0.001) according to different MR methods within different categories of hypotheses. The selection of diseases and risk factors is adapted from \citep{morrison2020mendelian}.}
        \textbf{(C)} \textit{Efficiency gains or losses for IB methods compared to their non-IB counterparts are examined by squared ratio of Wald-statistics for hypotheses in the "causal" category.   Each data-point is labelled by its exposure–outcome pair and colored by ratio class (blue: IB stronger, red: reference stronger, orange: equal).}
    }
    \label{fig:main_multi_panel}
\end{figure}

We applied alternative MR methods, including IB-Mode and IB-PRESSO, to evaluate evidence of causal effects in 46 pairs of exposure-outcomes that had previously been categorized based on evidence of causality \citep{morrison2020mendelian}.  We observe that for two pairs ( HDL $\rightarrow$ CAD , HDL $\rightarrow$ Stroke) that are categorized as non-causal based on prior randomized trials, surprisingly, most existing MR methods identify them as causal (\cref{fig:main_multi_panel}B); a full breakdown of which methods support each exposure-outcome relationship, by category, is given in \cref{summ_method_morrison}. In particular, MR-Mix and MR-Mode each identify only one as causal (HDL $\rightarrow$ CAD and HDL $\rightarrow$ Stroke, respectively), while IB-Mode is the only method that does not identify either as causal. For 15 pairs of exposure-outcomes, which are categorized as ``Correlated but with unclear causality or conflicting evidence", many of the existing methods identified a high proportion of the relationships as causal, again raising suspicion about bias in these methods. Two methods, IB-Mode and Egger, identified a much lower fraction ($\leq$27\%) of these relationships as causal; notably, IB-Mode controlled this fraction better than its base method MR-Mode. For 10 pairs categorized as ``unsupported'', most methods found no causal evidence, except IVW (TG $\rightarrow$ Asthma), MR-Mix (Height $\rightarrow$ T2D) and MR-ConMix (Height $\rightarrow$ T2D), each detecting a significant effect. Overall, based on these ``negative control hypotheses", it appears that IB-Mode is the most robust method for controlling type-I error.

For the 19 pairs of exposure-outcome categorized as ``likely causal", several methods, including IVW, Weighted Median, MR-PRESSO, IB-PRESSO, MR-cML and MR-ConMix, identify nearly 80\% of the hypotheses as causal. In contrast, MR-Mode, IB-Mode, Egger, and MR-Mix identified 53-63$\%$ of the same hypotheses as causal. The number of significant findings identified at the specified significance threshold ($\alpha=0.001$) was similar between MR-Mode and IB-Mode, whereas IB-PRESSO detected one more significant finding than MR-PRESSO. Furthermore, analysis of the underlying chi-square statistics (\cref{fig:main_multi_panel}C) shows that IB-PRESSO improved efficiency over MR-PRESSO for six of nine exposures, with relative gains up to roughly 29\% (BMI), while IB-Mode showed substantial but exposure-dependent gains over MR-Mode, reaching up to roughly 348\% (Smoking) for several exposures (see \cref{causalgain}). We also evaluated the robustness of the selection of auxiliary traits. When using the trait with the second largest absolute statistic $|\hat{\rho}_{CH}^{(N)}|$ as the auxiliary result, the results remained qualitatively similar across all hypothesis categories (\cref{fig:sig_pair_IB2}).

Several methods, including IVW, MR-PRESSO, IB-PRESSO, MR-cML, MRMedian, and MR-ConMix, identified significant causal effects (p-value $<$ 0.05) of vitamin D level on multiple traits, including CAD, T2D, and eGFR (\cref{VITD_table}).  MR-Mix identified a significant effect of vitamin D only on T2D, whereas MR-Mode (Weighted mode) flagged nominally significant effects on several traits, most strongly HDL, further illustrating the susceptibility of standard methods to spurious detection. After multiple-testing correction, neither Egger nor IB-Mode detected evidence of a causal effect---IB-Mode showed only nominal associations with HDL and eGFR (smallest $p=0.017$) that did not survive correction---while Egger produced much wider confidence intervals, indicating the low power of the method. Given the lack of clear evidence of the causal effect of vitamin D on cardio-metabolic traits such as CAD and T2D in randomized trials \citep{pittas2019vitamin, manson2019vitamin}, this provides yet another example that many existing ``robust'' MR methods are still susceptible to bias in the presence of complex pleiotropic effects, and IB-Mode is a robust method to balance sensitivity (power) and specificity for detection of causal effects.   

\input{Tables}

%% file: 4maintextFigs.tex
\generateFigSubpanels{figureSuggestions}{figureSuggestions}{1}
    {Making good figures requires attention to detail.}
    {{
        {General principals of making and labelling a good figure. These should help facilitate understanding and reproducibility. },
        {Some more detailed technical specifics about making figures and their sizes. Inkscape is a great free vector figure editor and alternative to Adobe Illustrator.\captionTips}
    }}

%% file: Tables.tex
% Placeholder: the main-text tables are typeset inline in the Results section
% and in the Supplement. This file is intentionally kept non-empty so that arXiv
% does not strip a zero-byte file, which would break \input{Tables} at compile time.

\begin{table}[!htbp]
\centering
\caption{\label{VITD_table}\textit{Estimates (top) and 95\% confidence intervals (bottom, in parentheses) for the causal effects of \textbf{vitamin D} on various outcomes from different MR methods. Effect sizes represent the log-odds ratio per standard deviation increase in vitamin D for binary outcomes, and the change in standard deviation units for continuous outcomes. Estimates significant after Bonferroni correction across the eight outcomes ($p<0.05/8$) are shown in \textbf{bold}.}}
\resizebox{\textwidth}{!}{
\fontsize{9}{11}\selectfont
\begin{tabular}[t]{lcccccccc}
\toprule
\textbf{Method} & \textbf{HTN} & \textbf{HDL} & \textbf{LDL} & \textbf{Asthma} & \textbf{CAD} & \textbf{eGFR} & \textbf{Stroke} & \textbf{T2D}\\
\midrule
IVW & \makecell{-0.03 \\ (-0.08, 0.02)} & \makecell{-0.05 \\ (-0.25, 0.15)} & \makecell{-0.08 \\ (-0.19, 0.02)} & \makecell{-0.00 \\ (-0.06, 0.06)} & \makecell{\textbf{-0.16} \\ (-0.24, -0.09)} & \makecell{\textbf{-0.08} \\ (-0.12, -0.03)} & \makecell{-0.02 \\ (-0.08, 0.05)} & \makecell{-0.07 \\ (-0.13, -0.01)}\\
Egger & \makecell{0.07 \\ (-0.02, 0.15)} & \makecell{0.06 \\ (-0.25, 0.37)} & \makecell{0.02 \\ (-0.14, 0.18)} & \makecell{-0.00 \\ (-0.10, 0.09)} & \makecell{0.06 \\ (-0.05, 0.18)} & \makecell{-0.04 \\ (-0.11, 0.04)} & \makecell{0.02 \\ (-0.08, 0.12)} & \makecell{0.04 \\ (-0.05, 0.13)}\\
MR-ConMix & \makecell{0.06 \\ (0.02, 0.10)} & \makecell{\textbf{0.11} \\ (0.08, 0.13)} & \makecell{\textbf{0.02} \\ (0.01, 0.03)} & \makecell{0.02 \\ (-0.04, 0.08)} & \makecell{-0.01 \\ (-0.06, 0.05)} & \makecell{\textbf{-0.04} \\ (-0.05, -0.03)} & \makecell{0.03 \\ (-0.07, 0.13)} & \makecell{\textbf{-0.10} \\ (-0.15, -0.06)}\\
MR-Mix & \makecell{-0.05 \\ (-0.18, 0.08)} & \makecell{0.06 \\ (0.00, 0.12)} & \makecell{0.02 \\ (-0.03, 0.07)} & \makecell{-0.08 \\ (-0.23, 0.07)} & \makecell{-0.05 \\ (-0.12, 0.02)} & \makecell{-0.03 \\ (-0.05, -0.01)} & \makecell{-0.12 \\ (-0.29, 0.05)} & \makecell{\textbf{-0.16} \\ (-0.25, -0.07)}\\
MR-cML & \makecell{0.01 \\ (-0.03, 0.05)} & \makecell{\textbf{0.08} \\ (0.05, 0.12)} & \makecell{0.02 \\ (0.00, 0.04)} & \makecell{-0.02 \\ (-0.08, 0.04)} & \makecell{-0.08 \\ (-0.14, -0.02)} & \makecell{\textbf{-0.04} \\ (-0.05, -0.02)} & \makecell{-0.02 \\ (-0.08, 0.04)} & \makecell{-0.07 \\ (-0.12, -0.02)}\\
Wt. Median & \makecell{0.00 \\ (-0.06, 0.06)} & \makecell{\textbf{0.10} \\ (0.06, 0.13)} & \makecell{0.03 \\ (0.00, 0.06)} & \makecell{-0.05 \\ (-0.14, 0.03)} & \makecell{-0.01 \\ (-0.08, 0.07)} & \makecell{-0.03 \\ (-0.05, -0.00)} & \makecell{0.03 \\ (-0.07, 0.14)} & \makecell{\textbf{-0.10} \\ (-0.16, -0.04)}\\
Wt. Mode & \makecell{0.03 \\ (-0.02, 0.09)} & \makecell{\textbf{0.10} \\ (0.05, 0.15)} & \makecell{0.04 \\ (0.01, 0.08)} & \makecell{-0.04 \\ (-0.13, 0.05)} & \makecell{0.07 \\ (0.00, 0.14)} & \makecell{-0.03 \\ (-0.06, -0.01)} & \makecell{0.05 \\ (-0.05, 0.16)} & \makecell{-0.07 \\ (-0.12, -0.02)}\\
IB-Mode & \makecell{0.04 \\ (-0.02, 0.09)} & \makecell{0.17 \\ (0.01, 0.33)} & \makecell{0.04 \\ (-0.07, 0.15)} & \makecell{-0.04 \\ (-0.15, 0.06)} & \makecell{0.04 \\ (-0.03, 0.11)} & \makecell{-0.07 \\ (-0.13, -0.01)} & \makecell{0.06 \\ (-0.06, 0.18)} & \makecell{-0.03 \\ (-0.09, 0.03)}\\
MR-PRESSO & \makecell{-0.01 \\ (-0.06, 0.04)} & \makecell{\textbf{0.09} \\ (0.03, 0.15)} & \makecell{0.01 \\ (-0.04, 0.06)} & \makecell{-0.02 \\ (-0.07, 0.04)} & \makecell{\textbf{-0.07} \\ (-0.12, -0.02)} & \makecell{-0.05 \\ (-0.10, 0.00)} & \makecell{-0.04 \\ (-0.10, 0.02)} & \makecell{\textbf{-0.07} \\ (-0.11, -0.02)}\\
IB-PRESSO & \makecell{0.00 \\ (-0.05, 0.05)} & \makecell{\textbf{0.10} \\ (0.04, 0.16)} & \makecell{0.01 \\ (-0.04, 0.06)} & \makecell{0.00 \\ (-0.05, 0.06)} & \makecell{-0.06 \\ (-0.11, -0.01)} & \makecell{-0.05 \\ (-0.10, 0.00)} & \makecell{-0.04 \\ (-0.11, 0.02)} & \makecell{-0.06 \\ (-0.10, -0.01)}\\
\bottomrule
\end{tabular}}
\end{table}

%% file: 5discussion.tex
\section{Discussion}

In summary, we have introduced a new class of Mendelian Randomization (MR) methods, termed Instrument Borrowing (IB), which can take advantage of information from an auxiliary trait to strengthen causal inference for a primary outcome. The IB framework relies on the hypothesis that closely related outcome traits are likely to share common confounders in relation to various exposures. We introduce a novel correlation statistics to identify such overlap, and related asymptotic inference procedure that will help to select suitable auxiliary traits. Our data analysis using cardiometabolic traits reveals that shared confounding is a common phenomenon and has unique exposure-specific patterns. We propose extensions of two popular MR methods, MR-Mode and MR-PRESSO, that allow joint analysis of the primary and auxiliary outcome traits with respect to a single exposure. Extensive simulation studies and data analysis demonstrate that IB-based MR methods can lead to a reduction in bias, an improvement in power, or both.

We used recently available summary statistics from the MVP study to re-evaluate the causal underpinning of various exposure–outcome pairs in the realm of cardiometabolic traits. The analysis, which includes negative control hypotheses, reveals that almost all existing MR methods, including those designed specifically to handle bias in the presence of invalid instruments, are susceptible to major bias in the presence of complex pleiotropic effects of the underlying genetic instruments. These patterns were also corroborated in our extensive simulation studies involving violations of the InSIDE assumption. In general, both simulation studies and data analysis indicated that IB-Mode was the most robust method, able to control type-I error and bias across a wide set of scenarios while still producing reasonable power and precision.

The statistics $\hat\rho_{CH}^{(N)}$ we introduce, relating an exposure to a pair of outcome traits, are of interest beyond their use to select auxiliary traits for IB-based MR analysis. The measure can be used more broadly to identify shared unmeasured confounding structures between related traits in relation to specific exposures, and thus can have other applications, such as the development of causal directed acyclic graphs (DAGs). Furthermore, if certain exposures have heterogeneous causal effects \citep{kennedy2023towards, brand2013causal}, the coheterogeneity measure can also capture overlap among traits with respect to the same causal components of the exposure. For example, BMI may have heterogeneous causal effects due to the relative contributions of body fat versus muscle, or due to differences in fat deposition between internal organs. For such exposures, a high coheterogeneity measure between two outcome traits would indicate shared causal pathways and could lead to deeper insights into biological mechanisms.

The IB framework for MR has many potential extensions. In this article, we implement it in the context of two established methods, MR-Mode and MR-PRESSO, because of ease of implementation. However, a number of other recent methods, including MR-Mix, MR-ConMix, and MR-cML, can be adapted to the IB framework. For example, both MR-Mix and MR-ConMix, which assume normally distributed effect-size distributions for horizontal pleiotropy among invalid instruments, could potentially be extended by specifying a corresponding bivariate normal distribution for pleiotropic effects across a pair of outcome traits. In addition, the IB framework could potentially be extended to take advantage of more than two related outcome traits and thereby gain additional robustness. Given the abundance of shared confounding structure observed for cardiometabolic traits in our data application, higher-dimensional extensions of the IB framework merit further research.

The IB framework also has limitations. When studying two or more outcome traits simultaneously, the underlying MR model requires more complex estimation and modelling in higher dimensions. For example, compared with MR-Mode, IB-Mode requires estimation of a density function in two dimensions, and it is well known that nonparametric density estimation can quickly break down into higher dimensions. Even in two dimensions, we observe that IB-Mode can lead to loss of efficiency compared to MR-Mode when the sample size of the underlying GWAS is small. Thus, the increased dimensionality and complexity of the underlying estimation procedure for IB-based methods require careful balancing with the available sample sizes for the GWAS traits under consideration. We also note that the IB extension does not provide uniform advantages between different methods. In our application, IB-Mode improved on MR-Mode across both robustness and efficiency, \rev{whereas the corresponding improvement of IB-PRESSO over MR-PRESSO was more modest under the outlier rule used here.}

% End matter (Data availability, Acknowledgements, etc.) moved to
% 7backmatter.tex, input after the Appendix (see 0main.tex).

%% file: appendix.tex
% ============================================================
% Main-text Appendix for JRSS-B.
%   Section: Regularity conditions and proofs (one appendix).
%   A.1 Proof of Theorem 1  (conditions + main steps; details online)
%   A.2 Proof of Lemma 1    (statement + conditions + complete proof;
%                            carries the subsec:scale label)
% Remaining proofs (genome-wide theorem, weak-instrument lemma,
% auxiliary regularity lemma, exact-variance derivation) are in the
% online Supplementary Material.
% Wiring: \input{appendix} after \input{5discussion} in 0main.tex;
%         the leading \appendix in supplement/0suppl_text.tex is
%         removed so \appendix is invoked once.
% ============================================================

\appendix
\section{Regularity conditions and proofs}
In this section we state the regularity conditions for \cref{thm:main} and \cref{lem:dov_scaling}, outline the proof of \cref{thm:main}, and prove \cref{lem:dov_scaling} in full.
The complete proof of \cref{thm:main} and of the auxiliary regularity lemma (\cref{lem:aux_reg}) is given in the online Supplementary Material. Additionally, in the online supplement, we provide extensions of our results under the presence of weak instruments (\cref{lem:weak_instrument_consistency}),
and a limit theorem (\cref{thm:genome_wide_limit}) for the estimation of genome-wide correlation parameter.
\label{app:proofs}
\subsection{Regularity conditions and proof of \texorpdfstring{\cref{thm:main}}{Theorem 1}}
\label{app:thm1}

Write \(\hat\rho_{CH}^{(N)}=T_N(\hat{\mathbf b})\),
a smooth functional of the stacked GWAS summary statistics \(\hat{\mathbf b}=(\hat{\mathbf b}_k)_{k\le K_N}\)
with \(\hat{\mathbf b}_k=(\hat\beta_{X,k},\hat\beta_{Y_1,k},\hat\beta_{Y_2,k})^\top\) and data-dependent
precision weights \(\hat w_k\); it is convenient to write \(T_N=g\circ Q_N\), with
\(Q_N=(\hat C_{12},\hat\tau_1^2,\hat\tau_2^2)\) and \(g(C,u,v)=C/\sqrt{uv}\), and to decompose the GWAS
errors as \(\sqrt N(\hat{\mathbf b}_k-\mathbf b_k)=\mathbf Z_{k,N}+\boldsymbol\delta_{k,N}\). We impose the
following conditions on the selected triangular array \(\{(\hat{\mathbf b}_k,\mathbf b_k):k\le K_N\}\).
The instrument set \(\mathcal K_N\) is itself random; throughout we condition on its realised value, and
all statements below are to be read as holding along any sequence of realised sets
\(\{\mathcal K_N\}_{N\ge1}\) for which \textup{(C1)}--\textup{(C6)} hold.

\input{conditions}

The proof relies on the following auxiliary regularity lemma, whose proof is given in the
Supplementary Material.
\begin{lemma}[Auxiliary regularity]
\label{lem:aux_reg}
Under Conditions \textup{(C1)}--\textup{(C5)}, 
\begin{enumerate}[label=\textup{(\roman*)}, leftmargin=*, itemsep=1pt]
\item \(\max_{k\le K_N} w_k=O(K_N^{-1})\) and \(\max_{k\le K_N}\hat w_k=O_p(K_N^{-1})\);
\item \(\sum_{k=1}^{K_N}|\hat w_k-w_k|=o_p(1)\);
\item \(Q_N\) and \(g\), regarded as functions of the summary statistics, are twice continuously
differentiable in a neighbourhood of the oracle point \((C_{12},\tau_1^2,\tau_2^2)\).
\end{enumerate}
\end{lemma}
\paragraph{Proof (main steps).}
\begin{enumerate}[label=\emph{Step \arabic*.}, leftmargin=*, itemsep=3pt]
\item \emph{Consistency.} From \textup{(C1)}--\textup{(C2)}, \(\|\hat{\mathbf b}-\mathbf b\|^2=O_p(K_N/N)=o_p(1)\).
Substituting \(\hat\theta_{l,k}=\theta_l+\alpha_{l,k}+\epsilon_{l,k}\) into
\(\hat\tau_l^2=\sum_k\hat w_k(\Delta_{l,k}^2-\hat\sigma_{l,k}^2)\) and expanding around
\(\tau_l^2=\sum_k w_k\tilde\alpha_{l,k}^2\), the correction \(-\hat\sigma_{l,k}^2\) cancels the leading
\(\sigma_{l,k}^2\) term and leaves three contributions: a weight-estimation term
\(\sum_k(\hat w_k-w_k)\tilde\alpha_{l,k}^2\), which is \(o_p(1)\) by \cref{lem:aux_reg}(ii) and \textup{(C4)};
a linear stochastic term of variance \(O\{(K_NN)^{-1}\}\) by \cref{lem:aux_reg}(i), \textup{(C4)} and
\(\sigma_{l,k}^2=O(N^{-1})\); and a centred quadratic term of variance \(O\{(K_NN^2)^{-1}\}\) by the
fourth-moment bound \textup{(C2)}. The post-selection bias adds \(o(N^{-1})\). Hence
\(\hat\tau_l^2\to_p\tau_l^2\) and \(\hat C_{12}\to_p C_{12}\); positivity \(\tau_l\ge\tau_{\min}>0\)
\textup{(C5)} makes the positive-part operator inactive, so \(\hat\rho_{CH}^{(N)}\to_p\rho_{CH}^{(N)}\).
\item \emph{Linearisation.} With a growing instrument set each SNP coordinate enters \(Q_N\) with a single
weight \(w_k=O(K_N^{-1})\) (\cref{lem:aux_reg}(i)), so \(\|\nabla T_N(\mathbf b)\|=O(K_N^{-1/2})\) and
\(\|\nabla^2T_N(\mathbf b)\|_{\mathrm{op}}=O(K_N^{-1})\). A second-order expansion gives
\(\hat\rho_{CH}^{(N)}-\rho_{CH}^{(N)}=\nabla T_N(\mathbf b)^\top(\hat{\mathbf b}-\mathbf b)+R_N\) with
\(|R_N|\le\tfrac12\|\nabla^2T_N\|_{\mathrm{op}}\|\hat{\mathbf b}-\mathbf b\|^2=O_p(N^{-1})
=o_p\{(NK_N)^{-1/2}\}\) by \textup{(C1)}.
\item \emph{Central limit theorem.} The linear term is a weighted sum of \(K_N\) conditionally independent,
mean-zero contributions \textup{(C2)}, centred at \(\rho_{CH}^{(N)}\) because
\(\max_k\|\boldsymbol\delta_{k,N}\|=o(N^{-1/2})\). With \(\max_k\hat w_k=O_p(K_N^{-1})\)
(\cref{lem:aux_reg}(i)) and the fourth-moment bound \textup{(C2)} the Lindeberg condition holds, giving
\(\sqrt N(\hat\rho_{CH}^{(N)}-\rho_{CH}^{(N)})/\sigma_N\rightsquigarrow\mathcal N(0,1)\).
\item \emph{Scaling.} By \textup{(C6)}, \(K_N\sigma_N^2\to\sigma^{*2}\), so Step~3 gives the stated limit
\(\sqrt{NK_N}(\hat\rho_{CH}^{(N)}-\rho_{CH}^{(N)})\rightsquigarrow\mathcal N(0,\sigma^{*2})\); and
\cref{lem:aux_reg}(iii) with the mean-value theorem gives \(K_N\hat\sigma_N^2\to_p\sigma^{*2}\), whence
\(\sqrt N(\hat\rho_{CH}^{(N)}-\rho_{CH}^{(N)})/\hat\sigma_N\rightsquigarrow\mathcal N(0,1)\) and
\(\hat\rho_{CH}^{(N)}\pm z_{1-\alpha/2}\hat\sigma_N/\sqrt N\) is asymptotically valid.
\end{enumerate}
Because \(\hat w_k\) depends on \(\hat{\mathbf b}\), the gradient \(\nabla_kT_N(\mathbf b)\) in
\eqref{eq:exact_var} differentiates through both the Wald ratios and the weights, each channel moving
\(\hat\rho_{CH}^{(N)}\) by \(O_p\{(NK_N)^{-1/2}\}\); the fixed-weight variance
\(\sigma_{N,\mathrm{fix}}^2\) omits the second.
\subsection{Proof of \texorpdfstring{\cref{lem:dov_scaling}}{Lemma 1}}
\label{subsec:scale}
In the notation of the generative model \eqref{sim_frame} of \cref{sec:simulation}, in which
\(\phi_{a,m}\) is the effect of SNP \(m\) on the confounder \(U_a\) and \(\delta_{l,m}\) its intrinsic
direct effect on \(Y_l\), put \(\eta_{l,m}=\delta_{l,m}/\beta_{X,m}\) and
\(\psi_{a,m}=\phi_{a,m}/\beta_{X,m}\) for the ratio-scale direct and confounder effects, and write
\(\theta_{U,l}:=\theta_{U\to Y_l}\). For \(a\in\{0,1,2\}\) let \(\mathcal U_a\) index the selected
instruments that are invalid through \(U_a\), where \(U_0\) is the shared confounder and \(U_1,U_2\) the
outcome-specific ones of \cref{sec:simulation}. We parametrise the split of the invalid instruments as
\(q_0=D_{\mathrm{ov}}\) and \(q_l=(1-q_0)r_l\) with \(r_1+r_2=1\), which is the constraint
\(q_0+q_1+q_2=1\) of \cref{lem:dov_scaling} written so that \(D_{\mathrm{ov}}\) may be varied with
\((r_1,r_2)\) held fixed.
\paragraph{Regularity conditions.}
\begin{enumerate}[label=\textup{(D\arabic*)}, leftmargin=*, itemsep=1pt]
\item \(|\beta_{X,m}|\) is bounded away from zero and, after centring, \(\Var(\eta_{l,m})=\sigma_{\eta_l}^2\),
\(\Var(\psi_{a,m})=\sigma_{U,R}^2\), \(\Cov(\eta_{1,m},\eta_{2,m})=0\), and \(\Cov(\eta_{l,m},\psi_{a,m})=0\).
\item Near-uniform weighting: the precision weights satisfy \(\max_{m\le K}|K w_m-1|\xrightarrow{p}0\).
\item The limiting marginal variances of the ratio-scale pleiotropic components are positive.
\end{enumerate}
\begin{revblock}
\paragraph{Proof (main steps).} By \eqref{sim_frame} the ratio-scale pleiotropic component is
\(\alpha_{l,m}=\eta_{l,m}+\theta_{U,l}\psi_{0,m}\mathbf 1\{m\in\mathcal U_0\}
+\theta_{U,l}\psi_{l,m}\mathbf 1\{m\in\mathcal U_l\}\), so only the shared category \(\mathcal U_0\)
enters both \(\alpha_{1,m}\) and \(\alpha_{2,m}\). Under \textup{(D2)} the weighted moments entering
\(\rho_{CH}^{(N)}\) have the same limits as their unweighted counterparts, and \textup{(D1)} then gives
\[
\Cov(\alpha_{1,m},\alpha_{2,m})=\bar\pi q_0\,\theta_{U,1}\theta_{U,2}\,\sigma_{U,R}^2,
\qquad
\Var(\alpha_{l,m})=\sigma_{\eta_l}^2+\bar\pi(q_0+q_l)\,\theta_{U,l}^2\,\sigma_{U,R}^2 .
\]
Substituting into \(\rho_{CH}=\Cov(\alpha_1,\alpha_2)/\{\Var(\alpha_1)\Var(\alpha_2)\}^{1/2}\) yields
\eqref{eq:rho_dov_main}, the constants there being \(a=\theta_{U,1}\sigma_{U,R}\),
\(b=\theta_{U,2}\sigma_{U,R}\), \(c=\sigma_{\eta_1}\) and \(d=\sigma_{\eta_2}\); the denominators are
positive by \textup{(D3)}. Since \(ab=\theta_{U,1}\theta_{U,2}\sigma_{U,R}^2\), the sign of
\(\rho_{CH}\) is that of \(\theta_{U,1}\theta_{U,2}\). For the monotonicity, hold \((r_1,r_2)\) fixed and
write \(x=q_0\), so that \(\Var(\alpha_{l,m})=A_l+B_lx\) with \(A_l>0\) and \(B_l\ge0\); up to a factor of
fixed sign the limit equals \(\varrho(x)=\bar\pi x\{(A_1+B_1x)(A_2+B_2x)\}^{-1/2}\), and a direct
calculation gives
\[
\varrho'(x)=\frac{\bar\pi\{2A_1A_2+(A_1B_2+A_2B_1)x\}}{2\{(A_1+B_1x)(A_2+B_2x)\}^{3/2}}>0,
\qquad 0\le x\le 1 .
\]
Hence \(\rho_{CH}\) is strictly increasing in \(D_{\mathrm{ov}}\) when \(\theta_{U,1}\theta_{U,2}>0\) and
strictly decreasing when \(\theta_{U,1}\theta_{U,2}<0\). The full argument, including the derivation of
the limiting moments and of \(A_l,B_l\), is given in \cref{app:lemma1_full} of the Supplementary Material.
\end{revblock}

%% file: conditions.tex
% ============================================================
% CANONICAL regularity conditions (C1)--(C6) for Theorem 1.
%
% This file is the SINGLE SOURCE for the condition list. It is \input twice:
%   * appendix.tex,               Appendix A.1  (primary statement)
%   * supplement/0suppl_text.tex, Section S4    (restated for convenience)
% Never edit a copy of this list anywhere else -- the two printings are
% guaranteed identical only because they are the same file.
%
% Symbols assumed already defined by the surrounding text:
%   K_N, \mathcal K_N, \hat{\mathbf b}_k, \mathbf b_k, \mathbf Z_{k,N},
%   \boldsymbol\delta_{k,N}, \Omega_k, \tau_l, \alpha_{l,k}, \sigma_N^2.
% The reference \eqref{eq:exact_var} is defined in 6methods.tex; if the
% supplement is compiled as a separate document, xr-hyper resolves it.
% ============================================================
\begin{enumerate}[label=\textup{(C\arabic*)}, leftmargin=*, itemsep=2pt]
\item \emph{Growing instrument set.} The number of selected instruments satisfies
\(K_N\to\infty\) and \(K_N=o(N)\).
\item \emph{Independence, moments and post-selection bias.} Conditionally on the selected
instrument set \(\mathcal K_N\), the estimators \(\hat{\mathbf b}_1,\ldots,\hat{\mathbf b}_{K_N}\)
are independent across SNPs and admit the decomposition
\(\sqrt N(\hat{\mathbf b}_k-\mathbf b_k)=\mathbf Z_{k,N}+\boldsymbol\delta_{k,N}\), in which the
\(\boldsymbol\delta_{k,N}\) are non-random conditionally on \(\mathcal K_N\) and satisfy
\(\max_{k\le K_N}\|\boldsymbol\delta_{k,N}\|=o(N^{-1/2})\), while the noise components obey
\(\mathbb E\,\mathbf Z_{k,N}=0\), \(\Cov(\mathbf Z_{k,N})=\Omega_k\) and
\(\sup_{N\ge1}\max_{k\le K_N}\mathbb E\|\mathbf Z_{k,N}\|^4<\infty\).
\item \emph{Uniform instrument strength.} There exist constants \(0<c_x<C_x<\infty\) such that
\(c_x\le|\beta_{X,k}|\le C_x\) uniformly in \(N\) and \(k\).
\item \emph{Bounded effects.} The causal effects \(\theta_l\) are fixed finite constants, and there
exists \(C_\alpha<\infty\) such that \(|\alpha_{l,k}|\le C_\alpha\) uniformly in \(N,k,l\).
\item \emph{Non-degenerate heterogeneity.} There exists \(\tau_{\min}>0\) such that
\(\tau_l\ge\tau_{\min}\) for \(l=1,2\). Together with \textup{(C4)} this also gives
\(\tau_l\le\tau_{\max}<\infty\), since \(\tau_l^2=\sum_k w_k\tilde\alpha_{l,k}^2\le(2C_\alpha)^2\).
\item \emph{Variance normalisation.} For \(\sigma_N^2\) as defined in \eqref{eq:exact_var},
\(K_N\sigma_N^2\to\sigma^{*2}\in(0,\infty)\).
\end{enumerate}

%% file: 7backmatter.tex
% ============================================================
% JASA end matter.  Funding, author contributions and competing
% interests live in the author footnote on the title page (JASA
% has no standing Acknowledgements section), so only the data and
% code statements appear here.  Repository URLs identify the
% authors and are therefore suppressed when \blind=1.
%
% Set single-spaced: these two statements are reference material (a list of
% URLs), not running prose, and at \doublespacing they cost roughly a page --
% which is the difference between 35 and 36 pages against the JASA limit.
% ============================================================
\begin{singlespace}
\ifnum\blind=0
\section*{Acknowledgements}
The research was funded by NIH grant R01HG013137. During preparation the authors used AI-assisted tools for code refactoring and language editing; all statistical methodology, analyses and results were designed, carried out and verified by the authors. A.C.\ developed the methods, implemented the software, and carried out all simulation studies and data analyses; N.C.\ conceptualized the study and was responsible for funding; both authors drafted and revised the manuscript. The authors declare no competing interests.
\fi

\section*{Data Availability Statement}
All public GWAS summary statistics used in our analyses are freely accessible from well-established consortia. The BMI of the participants was obtained from the GIANT consortium (\url{https://portals.broadinstitute.org/collaboration/giant/index.php/GIANT_consortium_data_files#BMI_and_Height_GIANT_and_UK_BioBank_Meta-analysis_Summary_Statistics}); lipid trait data (HDL-C, LDL-C, TG) were obtained from the Global Lipids Genetics Consortium's 1.65M - participating multi-ancestry GWAS, available via \url{https://www.lipidgenetics.org/}; vitamin D summary statistics come from \url{https://cnsgenomics.com/data/revez_20/}; and variants of blood pressure GWAS (SBP, DBP) were taken from a recent UK Biobank meta-analysis of approximately one million European individuals. Additional exposures to GWAS included fasting glucose (\url{https://opengwas.io/datasets/ebi-a-GCST90002232}), height (\url{https://opengwas.io/datasets/ieu-a-89}), birth weight (\url{https://opengwas.io/datasets/ukb-a-198}), body fat percentage (\url{https://opengwas.io/datasets/ukb-b-8909}), being a smoker (\url{https://opengwas.io/datasets/ukb-b-20261}), and drinks per week (\url{https://opengwas.io/datasets/ieu-b-73}). All MR outcomes---CAD, T2D, MI, HTN, stroke, and eGFR, as well as the auxiliary outcomes---were drawn from MVP's European-ancestry cohort (see \cref{Description} for the definitions of traits and sample sizes). Data can be downloaded from \url{https://ftp.ncbi.nlm.nih.gov/dbgap/studies/phs002453/analyses/GIA/}.

\section*{Code and Reproducibility}
\ifnum\blind=1
The custom software developed for this study is available as an \texttt{R} package, together with all
simulation and real-data analysis code, in two public repositories. The repository URLs are withheld
here because they identify the authors; they will be supplied on request to the editor and stated in
full in the unblinded version.
\else
The custom software developed for this study is available as the R package \texttt{IBMR} at \url{https://github.com/achatto4/IBMR} and as a website at \url{https://achatto4.github.io/IBMR/index.html}. Other simulation and real data analysis code can be found in \url{https://github.com/achatto4/Robust_Mendelian_Randomization_via_Instrument_Borrowing}.
\fi
\end{singlespace}

%% file: supplement/0suppl_text.tex
% (\appendix now invoked in main appendix.tex)
\section{Supplemental Notes}
\label{Supp-Not}

\begin{algorithm}
\caption{Instrument Borrowing MR Procedure}
\label{algo}
\begin{algorithmic}[1]
  \State \textbf{Identify traits}: Select an exposure $X$, primary outcome $Y_1$, and candidate auxiliary outcomes $Y_2, Y_3, \ldots$ with GWAS summary statistics available for each.
  \State \textbf{Select instruments}: Choose $K$ independent SNPs strongly associated with $X$ (e.g., $p < 5\times 10^{-8}$ after LD clumping).
  \For{each SNP $k = 1, \ldots, K$}
    \For{each trait $Y_l$}
      \State Compute the Wald ratio $\hat\theta_{l,k} = \hat\beta_{Y_l,k}/\hat\beta_{X,k}$ and its delta-method standard error $\hat\sigma_{l,k}$.
    \EndFor
  \EndFor
  \For{each candidate auxiliary trait $Y_l$ ($l \geq 2$)}
    \State Compute the coheterogeneity statistic $\hat\rho_{CH}^{(N)}(Y_1, Y_l)$ between the primary outcome $Y_1$ and the candidate auxiliary $Y_l$, correcting for any sample overlap between their GWAS.
  \EndFor
  \State \textbf{Select auxiliary traits}: Rank auxiliary outcomes by $|\hat\rho_{CH}^{(N)}|$ and retain those for which the test of $H_0:\rho_{CH}^{(N)} = 0$ is statistically significant.
  \State \textbf{Joint MR analysis}: Apply an instrument-borrowing MR method (e.g., IB-Mode or IB-PRESSO) using the top-ranked auxiliary trait; report sensitivity analyses with the next-ranked auxiliaries.
\end{algorithmic}
\end{algorithm}

% longtable, not table[H]: at 44 rows this no longer fits on a single page and a
% non-breaking float overruns the bottom margin.  longtable breaks across pages
% and repeats the header, and stays correct as further symbols are added.
\begin{longtable}{@{}p{0.22\linewidth}p{0.72\linewidth}@{}}
\caption{Summary of the notation used throughout the paper.}
\label{tab:notation}\\
\toprule
Symbol & Description \\
\midrule
\endfirsthead
\multicolumn{2}{@{}l}{\tablename\ \thetable{} -- continued}\\
\toprule
Symbol & Description \\
\midrule
\endhead
\bottomrule
\endfoot
\multicolumn{2}{@{}l}{\emph{Traits and data}}\\
$X$ & Index exposure trait \\
$Y_1,\,Y_2$ & Primary and auxiliary (secondary) outcome traits \\
$Y_l,\ l\in\{1,2\}$ & Generic outcome-trait index \\
$G_1,\dots,G_{K_N}$ & Genetic instruments (SNPs) selected for $X$ \\
$k$ & Index of a selected instrument, $k\in\mathcal K_N$ \\
$m$ & Index of a genome-wide SNP, $m\in\mathcal M_N$ \\
$\hat\beta_{X,k},\,\hat\beta_{Y_l,k}$ & Standardized GWAS effect estimates of SNP $k$ on $X$ and $Y_l$ \\
$\hat\theta_{l,k}$ & Wald ratio for $Y_l$ at SNP $k$, $\hat\beta_{Y_l,k}/\hat\beta_{X,k}$ \\
$\hat{\boldsymbol\theta}^{(k)}$ & Bivariate ratio vector, $(\hat\theta_{1,k},\hat\theta_{2,k})^\top$ \\
\addlinespace
\multicolumn{2}{@{}l}{\emph{Sample sizes (asymptotic regime)}}\\
$N_X=N$ & Effective sample size of the $X$ GWAS; indexes the asymptotic regime \\
$N_l$ & GWAS sample size for outcome $Y_l$ \\
$\kappa_l$ & Outcome-to-exposure sample-size ratio, $N_l/N$ \\
$K_N$ (abbrev.\ $K$) & Number of selected instruments; grows with $N$ \\
\addlinespace
\multicolumn{2}{@{}l}{\emph{Validity and pleiotropy}}\\
$\mathcal I_1,\,\mathcal I_2$ & Sets of invalid instruments for $Y_1,\,Y_2$ \\
$D_{\mathrm{ov}}$ & Overlap of invalid instruments, $|\mathcal I_1\cap\mathcal I_2|/|\mathcal I_1\cup\mathcal I_2|$ \\
$\theta_l$ & Causal effect of $X$ on $Y_l$ \\
$\alpha_{l,k}$ & Horizontal-pleiotropic effect of SNP $k$ on $Y_l$ (on the Wald-ratio scale) \\
$\epsilon_{l,k}$ & Sampling error of the Wald ratio $\hat\theta_{l,k}$, with $\mathbb E(\epsilon_{l,k})=0$ \\
$\sigma_{l,k}^2$ & Sampling variance of the Wald ratio, $\Var(\epsilon_{l,k})$ \\
$\sigma_{X,k}^2,\ \sigma_{Y_l,k}^2$ & Sampling variances of the GWAS effect estimates $\hat\beta_{X,k}$ and $\hat\beta_{Y_l,k}$ (not on the Wald-ratio scale) \\
$\mathbf{r}^{(k)}$ & Bivariate leave-one-out residual at instrument $k$, with components $\hat\beta_{Y_l,k}-\hat\theta^{(-k)}_{\mathrm{IVW},l}\hat\beta_{X,k}$ \\
$\Sigma^{r}_k$ & Null covariance of $\mathbf{r}^{(k)}$, with entries $\sigma_{r_l,k}^2$ and $\sigma_{r_{12},k}$ \\
$\rho_{s}$ & Correlation between the sampling errors of the two outcome GWAS induced by shared individuals; estimated by the intercept of bivariate LD-score regression \\
\addlinespace
\multicolumn{2}{@{}l}{\emph{Coheterogeneity}}\\
$\bar\alpha_l$ & Weighted mean pleiotropic effect for $Y_l$, $\sum_k w_k\alpha_{l,k}$ \\
$\tilde\alpha_{l,k}$ & Centered pleiotropic effect, $\alpha_{l,k}-\bar\alpha_l$ \\
$\tau_l^2$ & Weighted heterogeneity of pleiotropy for $Y_l$, $\sum_k w_k\tilde\alpha_{l,k}^2$ \\
$C_{12}$ & Weighted cross-trait pleiotropy covariance, $\sum_k w_k\tilde\alpha_{1,k}\tilde\alpha_{2,k}$ \\
$\rho_{CH}^{(N)}$ & Coheterogeneity parameter for the selected instrument set (sample size $N$), $C_{12}/(\tau_1\tau_2)$ \\
$\hat\rho_{CH}^{(N)}$ & Estimator of $\rho_{CH}^{(N)}$ \\
$\rho_{CH}$ & Genome-wide limiting coheterogeneity parameter \\
$\tilde D_{l,k}$ & Coheterogeneity-adjusted pleiotropic deviation (fixed-weight variance), $\tilde\alpha_{l,k}-\rho_{CH}^{(N)}(\tau_l/\tau_{3-l})\tilde\alpha_{3-l,k}$ \\
\addlinespace
\multicolumn{2}{@{}l}{\emph{IB-Mode estimation}}\\
$\hat{\boldsymbol{\theta}}^{*}$ & Joint causal-effect estimate $(\hat\theta_1,\hat\theta_2)^\top$ from the bivariate KDE \\
$w_k$ & Instrument weight, $w_k\propto 1/(\sigma_{1,k}\sigma_{2,k})$ \\
$H=\phi H_0$ & KDE bandwidth matrix; $\phi$ is the tuning parameter \\
\addlinespace
\multicolumn{2}{@{}l}{\emph{Generative model (simulations)}}\\
$p$ & Fraction of SNPs associated with $X$ (instrument pool) \\
$\pi,\,1-\pi$ & Fraction of valid / invalid instruments \\
$q_0,q_1,q_2$ & Partition of invalid instruments by confounder ($U_0$ shared; $U_1,U_2$ outcome-specific); $D_{\mathrm{ov}}=q_0$ \\
\end{longtable}

\subsection{Estimation of the per-instrument sampling covariances}
\label{sec:cov_est}

\rev{Each of the two instrument-borrowing methods requires a $2\times 2$ sampling covariance
at every instrument $k$, though on a different scale: the coheterogeneity estimator requires
the covariance of the Wald ratios, $\sigma_{12,k}$, whereas IB-PRESSO requires the covariance
of the bivariate residual $\mathbf{r}^{(k)}$ on the scale of the effect estimates,
$\Sigma^{r}_k$. We derive each in turn.}

The bias-corrected coheterogeneity estimator $\hat\rho_{CH}^{(N)}$ requires, for each
instrument $k$, the sampling variances $\sigma_{l,k}^2=\Var(\epsilon_{l,k})$ and the
cross-trait sampling covariance $\sigma_{12,k}=\Cov(\epsilon_{1,k},\epsilon_{2,k})$. The
ratio-scale variances are estimated from the reported exposure and outcome GWAS standard
errors using the delta method. The covariance
$\sigma_{12,k}$ is nonzero only when the two outcome GWAS are estimated on overlapping
samples. Let $N_1,N_2$ denote the outcome GWAS sample sizes and $N_s$ the number of shared
individuals. A delta-method expansion of the Wald ratios
$\hat\theta_{l,k}=\hat\beta_{Y_l,k}/\hat\beta_{X,k}$ gives
\begin{equation}
\sigma_{12,k}
= \frac{1}{\beta_{X,k}^2}\,\frac{N_s\,\rho_{12}}{N_1 N_2}
\;+\; \frac{\beta_{Y_1,k}\,\beta_{Y_2,k}}{N\,\beta_{X,k}^4}
\;+\; o(N^{-1}),
\label{eq:sigma12_est}
\end{equation}
where $\rho_{12}$ is the phenotypic correlation between $Y_1$ and $Y_2$ in the overlapping
samples and the second term is a negligible weak-instrument contribution. The leading term
identifies the sampling-error correlation as
\begin{equation}
\rho_s := \frac{N_s\,\rho_{12}}{\sqrt{N_1 N_2}},
\label{eq:rho_s_def}
\end{equation}
which is the constant cross-trait term estimated by the intercept of bivariate
(cross-trait) LD-score regression \citep{bulik2015atlas}: regressing the products of the
$Y_1$ and $Y_2$ $z$-scores on the LD score identifies the genetic covariance through its
slope and the overlap-induced covariance through its intercept. A raw cross-trait
intercept $\hat I_{12}$ lies on the correlation scale only when the univariate intercepts
$\hat I_{11}$ and $\hat I_{22}$ equal one, and it can be inflated by population
stratification shared between the two studies even in the absence of overlap
\citep{yengo2018expectation}. We therefore estimate $\rho_s$ by the normalised intercept
\begin{equation}
\hat\rho_s = \frac{\hat I_{12}}{\sqrt{\hat I_{11}\hat I_{22}}},
\label{eq:rho_s_hat}
\end{equation}
following the construction used for sample-overlap correction in local genetic-correlation
analysis \citep{werme2022integrated}; this is the quantity used in our implementation.
Substituting $\hat\rho_s$, together with the GWAS effect-size and standard-error estimates,
into \eqref{eq:sigma12_est} gives the plug-in $\hat\sigma_{12,k}$.

\begin{revblock}
For IB-PRESSO the same construction applies to the bivariate residual
$\mathbf{r}^{(k)}$, whose $l$-th component is
$r^{(k)}_l=\hat\beta_{Y_l,k}-\hat\theta^{(-k)}_{\mathrm{IVW},l}\hat\beta_{X,k}$.
Because $\hat\theta^{(-k)}_{\mathrm{IVW},l}$ is computed from the remaining $K-1$
instruments it is independent of the data at instrument $k$ and may be treated as
fixed, and the exposure and outcome GWAS are non-overlapping, so
$\Cov(\hat\beta_{X,k},\hat\beta_{Y_l,k})=0$. Hence
\begin{equation}
\sigma_{r_l,k}^2
= \sigma_{Y_l,k}^2 + \big(\hat\theta^{(-k)}_{\mathrm{IVW},l}\big)^2\sigma_{X,k}^2,
\qquad
\sigma_{r_{12},k}
= \rho_{s}\,\sigma_{Y_1,k}\sigma_{Y_2,k}
+ \hat\theta^{(-k)}_{\mathrm{IVW},1}\hat\theta^{(-k)}_{\mathrm{IVW},2}\,\sigma_{X,k}^2,
\label{eq:sigma_r_est}
\end{equation}
or equivalently
$\Sigma^{r}_k=\Sigma_{Y,k}+\sigma_{X,k}^2\,\hat{\boldsymbol\theta}^{(-k)}\big(\hat{\boldsymbol\theta}^{(-k)}\big)^\top$,
where $\Sigma_{Y,k}$ is the sampling covariance of
$(\hat\beta_{Y_1,k},\hat\beta_{Y_2,k})^\top$ and
$\hat{\boldsymbol\theta}^{(-k)}=(\hat\theta^{(-k)}_{\mathrm{IVW},1},\hat\theta^{(-k)}_{\mathrm{IVW},2})^\top$.
The second term in each entry arises because the residual is taken about a fitted
slope, so the exposure's sampling error enters both components; unlike the first
term in $\sigma_{r_{12},k}$, it does not vanish when the outcome GWAS are
non-overlapping.
\end{revblock}

\subsection{\texorpdfstring{Proof of Theorem~\ref{thm:main}}{Proof of Theorem}}
\label{app:proof_main}

Let \(\mathcal M_N\) denote the finite set of candidate genome-wide SNPs
available at sample size \(N\). The selected instrument set is
\(\mathcal K_N\subseteq\mathcal M_N\), with \(K_N:=|\mathcal K_N|\). Throughout
this proof, we condition on the realized set \(\mathcal K_N\) and re-index its
elements as \(k=1,\ldots,K_N\). For each selected SNP \(k\), write
\(\mathbf b_k:=(\beta_{X,k},\beta_{Y_1,k},\beta_{Y_2,k})^\top\) for the vector of
true marginal effects of SNP \(k\) on the exposure \(X\) and the two outcomes
\(Y_1,Y_2\), and \(\hat{\mathbf b}_k:=(\hat\beta_{X,k},\hat\beta_{Y_1,k},\hat\beta_{Y_2,k})^\top\)
for the corresponding vector of GWAS effect-size estimators. We prove the result under
Conditions \textup{(C1)}--\textup{(C6)} of Appendix~\ref{app:thm1}, imposed on the selected
triangular array \(\{(\hat{\mathbf b}_k,\mathbf b_k):1\le k\le K_N\}\). For convenience these
are restated here verbatim.

\input{conditions}
The auxiliary smoothness and weight-regularity properties used throughout the
proof are collected in the auxiliary regularity lemma stated in the main-text
appendix (\cref{lem:aux_reg}); we prove it here.

\begin{proof}[Proof of \cref{lem:aux_reg} (auxiliary regularity)]
By \textup{(C3)}--\textup{(C4)}, \(|\beta_{X,k}|\ge c_x>0\) and
\(r_{l,k}:=\beta_{Y_l,k}/\beta_{X,k}=\theta_l+\alpha_{l,k}=O(1)\) uniformly
in \(N,k,l\). Hence the delta-method variance formula gives
\(\sigma_{l,k}^2\asymp N^{-1}\) uniformly in \(k\). Since
\(w_k\propto(\sigma_{1,k}^2\sigma_{2,k}^2)^{-1/2}\), this implies
\(\max_k w_k=O(K^{-1})\).

Next, by \textup{(C2)} and \textup{(C1)},
\(\max_{k\le K}\|\hat{\mathbf b}_k-\mathbf b_k\|=o_p(1)\); indeed, the
fourth-moment bound gives
\(\sum_{k=1}^K\mathbb E\|\hat{\mathbf b}_k-\mathbf b_k\|^4=O(K/N^2)=o(1)\).
Because the delta-method variance estimators are smooth plug-in functions of
the GWAS beta estimates, with denominators bounded away from zero by
\textup{(C3)}, it follows that
\[
\max_{l=1,2}\max_{k\le K}
\left|
\frac{\hat\sigma_{l,k}^2}{\sigma_{l,k}^2}-1
\right|
=o_p(1).
\]
Writing \(a_k=(\sigma_{1,k}^2\sigma_{2,k}^2)^{-1/2}\) and
\(\hat a_k=(\hat\sigma_{1,k}^2\hat\sigma_{2,k}^2)^{-1/2}\), we therefore have
\(\hat a_k=a_k\{1+o_p(1)\}\) uniformly in \(k\). Hence
\[
\hat w_k
=
\frac{\hat a_k}{\sum_{j=1}^K\hat a_j}
=
\frac{a_k\{1+o_p(1)\}}
{\sum_{j=1}^K a_j\{1+o_p(1)\}}
=
w_k\{1+o_p(1)\}
\]
uniformly in \(k\). This gives \(\max_k\hat w_k=O_p(K^{-1})\), and
\[
\sum_{k=1}^K|\hat w_k-w_k|
\le
\max_{k\le K}
\left|
\frac{\hat w_k}{w_k}-1
\right|
\sum_{k=1}^K w_k
=o_p(1).
\]

Finally, \textup{(C3)} ensures smoothness of the Wald-ratio maps,
\textup{(C4)} gives uniform boundedness of the ratio-scale pleiotropic
components, and \textup{(C5)} gives \(\tau_l\ge\tau_{\min}>0\). Therefore
the weighted covariance/variance maps and the correlation map
\(g(C,u,v)=C/\sqrt{uv}\) are twice continuously differentiable on a
neighborhood of the oracle point.
\end{proof}

Let
\(\hat{\mathbf b}=(\hat{\mathbf b}_1^\top,\ldots,\hat{\mathbf b}_K^\top)^\top\)
and
\(\mathbf b=(\mathbf b_1^\top,\ldots,\mathbf b_K^\top)^\top\), and set
\(\mathbf Z=(\mathbf Z_{1,N}^\top,\ldots,\mathbf Z_{K,N}^\top)^\top\),
\(\boldsymbol\delta=(\boldsymbol\delta_{1,N}^\top,\ldots,\boldsymbol\delta_{K,N}^\top)^\top\),
so that \(\sqrt N(\hat{\mathbf b}-\mathbf b)=\mathbf Z+\boldsymbol\delta\).
Throughout, we work with the smooth version of the plug-in map
\(\hat{\mathbf b}\mapsto\hat\rho_{CH}^{(N)}\) guaranteed by
Lemma~\ref{lem:aux_reg}(iii), still denoted \(T_N\).

We first show consistency of the bias-corrected covariance and variance
estimators. By \textup{(C2)},
\[
\|\hat{\mathbf b}-\mathbf b\|^2
=
\sum_{k=1}^K\|\hat{\mathbf b}_k-\mathbf b_k\|^2
=
O_p(K/N)+o(K/N^2)
=
O_p(K/N)
=
o_p(1),
\]
where the \(o(K/N^2)\) term collects the deterministic bias contribution
\(\|\boldsymbol\delta\|^2/N\le K\max_k\|\boldsymbol\delta_{k,N}\|^2/N=o(K/N^2)\).
Recall the Wald-ratio decomposition
\(\hat\theta_{l,k}=\theta_l+\alpha_{l,k}+\epsilon_{l,k}\) with
\(\mathbb E(\epsilon_{l,k})=0\) and
\(\sigma_{l,k}^2=\Var(\epsilon_{l,k})\). Substituting this decomposition into
\(\hat\tau_l^2=\sum_k\hat w_k(\Delta_{l,k}^2-\hat\sigma_{l,k}^2)\) and
expanding around the oracle counterpart
\(\tau_l^2=\sum_k w_k\tilde\alpha_{l,k}^2\) gives
\[
\hat\tau_l^2-\tau_l^2
=
\sum_{k=1}^K(\hat w_k-w_k)\tilde\alpha_{l,k}^2
+
2\sum_{k=1}^K\hat w_k\tilde\alpha_{l,k}(\epsilon_{l,k}-\bar\epsilon_l)
+
\sum_{k=1}^K\hat w_k
\left[
(\epsilon_{l,k}-\bar\epsilon_l)^2-\sigma_{l,k}^2
\right]
+
o_p(1),
\]
where \(\bar\epsilon_l=\sum_k\hat w_k\epsilon_{l,k}\) and the bias
correction \(-\hat\sigma_{l,k}^2\) cancels the leading
\(\sigma_{l,k}^2\) term. The first term is \(o_p(1)\) by \textup{(C4)} and
Lemma~\ref{lem:aux_reg}(ii). The second has variance bounded by
\(
\max_k\hat w_k\sum_k\hat w_k\tilde\alpha_{l,k}^2\sigma_{l,k}^2
=O_p(K^{-1}N^{-1})=O_p((KN)^{-1}),
\)
using Lemma~\ref{lem:aux_reg}(i), \textup{(C4)}, and
\(\sigma_{l,k}^2=O(N^{-1})\) under \textup{(C2)}--\textup{(C3)}. The third
has variance \(O((KN^2)^{-1})\) by the fourth-moment bound on
\(\mathbf Z_{k,N}\) in \textup{(C2)}. The deterministic bias
\(\boldsymbol\delta\) contributes additional cross- and quadratic-terms of
order \(\max_k\|\boldsymbol\delta_{k,N}\|/\sqrt N=o(N^{-1})\) and
\(\max_k\|\boldsymbol\delta_{k,N}\|^2=o(N^{-1})\), respectively, both
absorbed into the \(o_p(1)\) remainder. Hence
\(\hat\tau_l^2\to_p\tau_l^2\), and the same argument gives
\(\hat C_{12}\to_p C_{12}\). 

Since \(\tau_l^2\ge\tau_{\min}^2>0\) and
\(\hat\tau_l^2\to_p\tau_l^2\), the positive-part operator in
\(\hat\tau_l^2\) is asymptotically inactive. Since \(\tau_1\) and \(\tau_2\) are bounded
away from zero by \textup{(C5)}, the continuous mapping theorem gives
\[
\hat\rho_{CH}^{(N)}-\rho_{CH}^{(N)}\xrightarrow{p}0.
\]

We next establish asymptotic normality. Let
\(\tilde\rho^{(N)}:=T_N(\mathbf b)\). A second-order Taylor expansion yields
\[
\hat\rho_{CH}^{(N)}-\tilde\rho^{(N)}
=
\nabla T_N(\mathbf b)^\top(\hat{\mathbf b}-\mathbf b)
+
\frac12
(\hat{\mathbf b}-\mathbf b)^\top
\nabla^2T_N(\mathbf b^*)
(\hat{\mathbf b}-\mathbf b),
\]
where \(\mathbf b^*\) lies between \(\hat{\mathbf b}\) and \(\mathbf b\).
Each component of \(Q_N\) is a weighted sum in which the \(k\)-th SNP
coordinate enters multiplied by a single weight \(w_k\); by
Lemma~\ref{lem:aux_reg}(i) and \textup{(C3)}--\textup{(C5)}, this gives
\[
\left\|
\frac{\partial Q_{N,i}}{\partial \mathbf b_j}
\right\|=O(K^{-1}),
\qquad
\left\|
\frac{\partial^2 Q_{N,i}}{\partial \mathbf b_j^2}
\right\|=O(K^{-1}),
\qquad
\left\|
\frac{\partial^2 Q_{N,i}}{\partial \mathbf b_j\partial \mathbf b_m}
\right\|=O(K^{-2}),\quad j\neq m.
\]
Composition with \(g\), whose derivatives are bounded on a neighborhood of
\(Q_N(\mathbf b)\) by Lemma~\ref{lem:aux_reg}(iii), gives
\(\|\nabla T_N(\mathbf b)\|=O(K^{-1/2})\) and
\(\|\nabla^2T_N(\mathbf b)\|_{\mathrm{op}}=O(K^{-1})\). The Taylor remainder
is therefore \(O_p(K^{-1})\cdot O_p(K/N)=O_p(N^{-1})\), which is
\(o_p((NK)^{-1/2})\) by \textup{(C1)}.

The leading term decomposes into noise and bias contributions. Substituting
\(\sqrt N(\hat{\mathbf b}-\mathbf b)=\mathbf Z+\boldsymbol\delta\):
\[
\sqrt{NK}\,\nabla T_N(\mathbf b)^\top(\hat{\mathbf b}-\mathbf b)
=
\sum_{k=1}^K\xi_k + \Lambda_N,
\qquad
\xi_k=\sqrt K\,\nabla_kT_N(\mathbf b)^\top\mathbf Z_{k,N},
\qquad
\Lambda_N=\sqrt K\,\nabla T_N(\mathbf b)^\top\boldsymbol\delta.
\]
The bias term satisfies
\[
|\Lambda_N|\le\sqrt K\cdot K\cdot\max_k\|\nabla_kT_N(\mathbf b)\|\cdot\max_k\|\boldsymbol\delta_{k,N}\|
=\sqrt K\cdot O(1)\cdot o(N^{-1/2})
=o\!\bigl(\sqrt{K/N}\bigr)=o(1)
\]
by \textup{(C1)} and the bias rate in \textup{(C2)}. Conditional on
\(\mathcal K_N\), the \(\xi_k\)'s are independent and mean zero by
\textup{(C2)}. By \textup{(C6)},
\(\sum_{k=1}^K\Var(\xi_k)=K\sigma_N^2\to\sigma^{*2}\). The fourth-moment
bound on \(\mathbf Z_{k,N}\) in \textup{(C2)}, together with
\(\|\nabla_kT_N(\mathbf b)\|=O(K^{-1})\), gives the Lyapunov condition. Hence
\(\sum_k\xi_k\rightsquigarrow \mathcal{N}(0,\sigma^{*2})\). The deterministic
difference \(T_N(\mathbf b)-\rho_{CH}^{(N)}\) is \(O(N^{-1})\), arising only
from finite-\(N\) bias-correction terms in \(Q_N(\mathbf b)\). Therefore
\(\sqrt{NK}\{T_N(\mathbf b)-\rho_{CH}^{(N)}\}=O(\sqrt{K/N})=o(1)\) by
\textup{(C1)}, and Slutsky's theorem completes the argument.

The exact variance \eqref{eq:exact_var} is the delta-method sandwich
\(\sum_k\nabla_k T_N(\mathbf b)^\top \Omega_k\, \nabla_k T_N(\mathbf b)\), with
\(\nabla_k T_N\) the gradient of \(T_N=C_{12}/(\tau_1\tau_2)\) through both the
Wald ratios and the data-dependent weights. It is instructive to isolate the two
channels. Chain-ruling only the Wald-ratio channel,
\(d\hat\theta_{l,k} = \beta_{X,k}^{-1}(d\hat\beta_{Y_l,k} - r_{l,k}\,d\hat\beta_{X,k})\)
with \(r_{l,k}=\beta_{Y_l,k}/\beta_{X,k}\) and holding the weights fixed, yields by
direct algebra the known-weights variance
\begin{equation}
\sigma_{N,\mathrm{fix}}^2
=\frac{1}{\tau_1^2\tau_2^2}\sum_{k=1}^K\Big(\frac{w_k}{\beta_{X,k}}\Big)^2
\Big[\big(r_{1,k}\tilde D_{2,k}+r_{2,k}\tilde D_{1,k}\big)^2
+\frac{\tilde D_{2,k}^2}{\kappa_1}+\frac{\tilde D_{1,k}^2}{\kappa_2}
+\frac{2\rho_{s}\,\tilde D_{1,k}\tilde D_{2,k}}{\sqrt{\kappa_1\kappa_2}}\Big],
\label{eq:oracle_var_closed_form}
\end{equation}
where \(\tilde D_{1,k}=\tilde\alpha_{1,k}-\rho_{CH}^{(N)}(\tau_1/\tau_2)\tilde\alpha_{2,k}\),
\(\tilde D_{2,k}=\tilde\alpha_{2,k}-\rho_{CH}^{(N)}(\tau_2/\tau_1)\tilde\alpha_{1,k}\),
\(\kappa_l=N_l/N\) are the outcome-to-exposure sample-size ratios, and \(\rho_{s}\) is the
correlation induced by shared individuals between the two outcome GWAS. Retaining the dependence of \(\hat w_k\) on \(\hat{\mathbf b}\)
contributes a weight-estimation channel of the \emph{same} order \((NK)^{-1/2}\);
consequently \eqref{eq:oracle_var_closed_form} omits a first-order term and is not,
in general, the asymptotic variance; only \eqref{eq:exact_var} is. The plug-in
estimator is obtained by replacing oracle quantities with their empirical
counterparts; uniform smoothness from Lemma~\ref{lem:aux_reg}(iii) and the
mean-value theorem give \(K|\hat\sigma_N^2 - \sigma_N^2| = o_p(1)\), and
together with \(K\sigma_N^2 \to \sigma^{*2}\) this proves
\(K\hat\sigma_N^2 \to_p \sigma^{*2}\). This completes the proof of Theorem~\ref{thm:main}.

\begin{remark}[Bias correction]
The exact leading bias correction for the centered quadratic term contains
\(\hat w_k(1-\hat w_k)\), whereas the estimator uses \(\hat w_k\). The
difference satisfies
\[
\sum_{k=1}^K\hat w_k^2\hat\sigma_{l,k}^2
\le
\max_k\hat w_k
\sum_{k=1}^K\hat w_k\hat\sigma_{l,k}^2
=
O_p\{(KN)^{-1}\}
=
o_p\{(NK)^{-1/2}\},
\]
hence negligible at the scale of the limit in Theorem~\ref{thm:main}.
\end{remark}

\subsubsection{Behavior under weak instruments and robustness}

Theorem~\ref{thm:main} relies on the strong-instrument condition \textup{(C3)}.
Here we complement it with results describing the estimator's behaviour when that
condition is relaxed, and clarifying the role of centering and weighting. These
are not required for Theorem~\ref{thm:main} itself but characterize the
robustness of \(\hat\rho_{CH}^{(N)}\) in regimes of practical interest.

\begin{slemma}[Consistency and asymptotic normality under weak instruments]
\label{lem:weak_instrument_consistency}
Suppose conditions \textup{(C1)}, \textup{(C2)}, \textup{(C4)}, and
\textup{(C5)} hold, and that \textup{(C3)} is replaced by
\begin{equation}
\label{eq:weak_inst_cond}
|\beta_{X,k}|\ge c_N^*
\quad\text{uniformly in }k,\qquad
\frac{\sqrt N\,c_N^*}{K^{1/4}}\to\infty,
\qquad
\max_{1\le k\le K} w_k\to0 .
\end{equation}
Then the following hold.
\begin{enumerate}
\item[\textup{(i)}] \emph{(Consistency)}
\(\hat\rho_{CH}^{(N)}-\rho_{CH}^{(N)}\xrightarrow{p}0.\)
\item[\textup{(ii)}] \emph{(Asymptotic normality)}
Define \(\mathcal V_N:=\sum_{k=1}^K w_k^2/\beta_{X,k}^2\),
\(\mathcal G_N:=\sum_{k=1}^K w_k/\beta_{X,k}^2\), and
\(\mathcal D_N:=\sum_{k=1}^K w_k/|\beta_{X,k}|\). If, in addition,
\(\sigma_N^2/\mathcal V_N\to\sigma_\dagger^2\in(0,\infty)\),
\[
\frac{\max_{k\le K}w_k^2/\beta_{X,k}^2}{\mathcal V_N}\to0,
\qquad
\frac{\mathcal G_N}{\sqrt{N\mathcal V_N}}\to0,
\qquad
\frac{\mathcal D_N\max_{k\le K}\|\boldsymbol\delta_{k,N}\|}{\sqrt{\mathcal V_N}}\to0,
\]
then
\[
\sqrt{\frac{N}{\mathcal V_N}}
\left(
\hat\rho_{CH}^{(N)}-\rho_{CH}^{(N)}
\right)
\xrightarrow{d}
\mathcal{N}(0,\sigma_\dagger^2).
\]
In particular, the weak-instrument estimation error is of order \((\mathcal V_N/N)^{1/2}\);
when \(w_k\asymp K^{-1}\) and \(|\beta_{X,k}|\asymp c_N^*\), this becomes
\(\{c_N^*\sqrt{NK}\}^{-1}\).
\end{enumerate}
\end{slemma}

\begin{proof}
The proof follows the same strategy as the consistency part of
Theorem~\ref{thm:main}. By
\textup{(C2)}, the fourth-moment bound and a union bound give, for every
\(\varepsilon>0\),
\[
P\left(
\max_{k\le K}\|\hat{\mathbf b}_k-\mathbf b_k\|>\varepsilon c_N^*
\right)
\le
C\frac{K}{N^2(c_N^*)^4}+o(1)
=
C\left(\frac{K^{1/4}}{\sqrt N\,c_N^*}\right)^4+o(1)
\to0 .
\]
Thus
\[
\max_{k\le K}
\frac{\|\hat{\mathbf b}_k-\mathbf b_k\|}{c_N^*}=o_p(1).
\]
Since \(|\beta_{X,k}|\ge c_N^*\), the Wald-ratio and delta-method variance
maps remain uniformly smooth over the selected instruments. Hence
\[
\max_{l=1,2}\max_{k\le K}
\left|
\frac{\hat\sigma_{l,k}^2}{\sigma_{l,k}^2}-1
\right|
=o_p(1).
\]
Writing \(a_k=(\sigma_{1,k}^2\sigma_{2,k}^2)^{-1/2}\) and
\(\hat a_k=(\hat\sigma_{1,k}^2\hat\sigma_{2,k}^2)^{-1/2}\), it follows that
\(\hat a_k=a_k\{1+o_p(1)\}\) uniformly in \(k\), and therefore
\[
\sum_{k=1}^K|\hat w_k-w_k|=o_p(1),
\qquad
\max_{k\le K}\hat w_k=o_p(1),
\]
where the second relation uses \(\max_k w_k\to0\).

Now expand \(\hat\tau_l^2\) around its oracle counterpart:
\[
\hat\tau_l^2-\tau_l^2
=
\sum_{k=1}^K(\hat w_k-w_k)\tilde\alpha_{l,k}^2
+
2\sum_{k=1}^K\hat w_k\tilde\alpha_{l,k}\epsilon_{l,k}
+
\sum_{k=1}^K\hat w_k
\left[
(\epsilon_{l,k}-\bar\epsilon_l)^2
-\mathbb E\{(\epsilon_{l,k}-\bar\epsilon_l)^2\}
\right]
+o_p(1).
\]
The first term is \(o_p(1)\) by \textup{(C4)} and
\(\sum_k|\hat w_k-w_k|=o_p(1)\). For the linear stochastic term,
\[
\Var\left(2\sum_k\hat w_k\tilde\alpha_{l,k}\epsilon_{l,k}\right)
\le
C\sum_k\hat w_k^2\sigma_{l,k}^2
\le
\frac{C\max_k\hat w_k}{N(c_N^*)^2}
=o_p(1),
\]
using \(\sigma_{l,k}^2=O\{1/(N(c_N^*)^2)\}\) uniformly and
\(\sum_k\hat w_k^2\le\max_k\hat w_k\). For the quadratic stochastic term,
the fourth-moment bound in \textup{(C2)} gives
\[
\Var\left(
\sum_k\hat w_k
\left[
(\epsilon_{l,k}-\bar\epsilon_l)^2
-\mathbb E\{(\epsilon_{l,k}-\bar\epsilon_l)^2\}
\right]
\right)
\le
C\sum_k\hat w_k^2\frac{1}{N^2(c_N^*)^4}
\le
\frac{C\max_k\hat w_k}{N^2(c_N^*)^4}
=o_p(1).
\]
Hence \(\hat\tau_l^2\to_p\tau_l^2\). The same argument, applied to the
corresponding expansion of \(\hat C_{12}\), gives
\(\hat C_{12}\to_p C_{12}\). Since \(\tau_l\ge\tau_{\min}>0\) by
\textup{(C5)}, the continuous mapping theorem yields
\[
\hat\rho_{CH}^{(N)}
=
\frac{\hat C_{12}}{\hat\tau_1\hat\tau_2}
\to_p
\frac{C_{12}}{\tau_1\tau_2}
=
\rho_{CH}^{(N)},
\]
which proves part~(i).

\smallskip\noindent
For part~(ii), the expansions of \(\hat\tau_l^2\) and \(\hat C_{12}\) above
decompose the estimation error into weight-estimation, bias, and linear and
quadratic stochastic terms. Under the additional conditions, the
weight-estimation and quadratic terms are \(o_p\{(\mathcal V_N/N)^{1/2}\}\) by the same
bounds as in part~(i). For the bias term, \textup{(C2)} contributes
\(O\{\|\boldsymbol\delta_{k,N}\|/(\sqrt N|\beta_{X,k}|)\}\) to \(\hat\theta_{l,k}\), so weighting by
\(w_k\) gives \(O\{\mathcal D_N\max_{k\le K}\|\boldsymbol\delta_{k,N}\|/\sqrt N\}
=o\{(\mathcal V_N/N)^{1/2}\}\) by the stated condition. The normalized
linear term \(\sqrt{N/\mathcal V_N}\sum_{k}w_k\tilde\alpha_{l,k}\epsilon_{l,k}\) is a sum
of independent, mean-zero contributions \(\chi_{k,N}\) whose total variance
converges to \(\sigma_\dagger^2\). Its Lyapunov ratio satisfies
\[
\sum_{k=1}^K\mathbb E|\chi_{k,N}|^4
=O\!\Big(\mathcal V_N^{-2}\sum_{k=1}^K w_k^4/\beta_{X,k}^4\Big)
=O\!\Big(\frac{\max_{k\le K} w_k^2/\beta_{X,k}^2}{\mathcal V_N}\Big),
\]
using the elementary bound
\(\sum_k w_k^4/\beta_{X,k}^4\le(\max_k w_k^2/\beta_{X,k}^2)\sum_k w_k^2/\beta_{X,k}^2
=(\max_k w_k^2/\beta_{X,k}^2)\,\mathcal V_N\).
This tends to zero by the uniform fourth-moment bound in \textup{(C2)} and the
negligibility condition \(\max_{k\le K}(w_k^2/\beta_{X,k}^2)/\mathcal V_N\to0\) above, so
the Lyapunov---and hence Lindeberg--Feller---central limit theorem applies. Applying the
delta method to \((\hat C_{12},\hat\tau_1,\hat\tau_2)\mapsto
\hat C_{12}/(\hat\tau_1\hat\tau_2)\), as in the asymptotic-normality part of
Theorem~\ref{thm:main} with \(\mathcal V_N\) replacing the corresponding variance
functional, then gives
\[
\sqrt{\frac{N}{\mathcal V_N}}
\left(\hat\rho_{CH}^{(N)}-\rho_{CH}^{(N)}\right)
\xrightarrow{d}
\mathcal{N}(0,\sigma_\dagger^2),
\]
which proves part~(ii).
\end{proof}

\begin{remark}[Robustness via centering and weighting]
\label{rmk:robustness_centering_weighting}
The estimator has two stabilizing features. First, \(T_N\) depends on the
Wald ratios only through the centered quantities
\(\Delta_{l,k}=\hat\theta_{l,k}-\sum_j\hat w_j\hat\theta_{l,j}\). Hence
adding any SNP-common shift to \(\hat\theta_{l,k}\) leaves \(T_N\) unchanged.
Consequently, the component of post-selection or winner's-curse bias that is
approximately common across selected SNPs is removed by centering; only its
SNP-specific residual contributes to the leading bias term
\(\Lambda_N=\sqrt K\,\nabla T_N(\mathbf b)^\top\boldsymbol\delta\), which is
controlled by \textup{(C2)}.

Second, precision weighting reduces weak-instrument amplification. Since
\(\sigma_{l,k}^2\asymp\{N\beta_{X,k}^2\}^{-1}\), SNPs with smaller
\(|\beta_{X,k}|\) have larger ratio-scale variance and receive smaller
precision weight. In the weak-instrument regime
\(\inf_{k\le K}|\beta_{X,k}|\asymp c_N^*\), the weighted stochastic terms are
of order \(O\{1/(NK(c_N^*)^2)\}\) for the linear term and
\(O\{1/(KN^2(c_N^*)^4)\}\) for the quadratic term, as in
Lemma~\ref{lem:weak_instrument_consistency}. Thus centering attenuates
SNP-common selection bias, while weighting attenuates heteroskedastic
weak-instrument noise.
\end{remark}

\subsection{\texorpdfstring{Genome-wide Finite-Population Limit}{Genome-wide Finite-Population Limit}}
\label{app:proof_genome_wide_limit}

Theorem~\ref{thm:main} gives a CLT for \(\hat\rho_{CH}^{(N)}\) targeting the
selected-set oracle \(\rho_{CH}^{(N)}\), conditional on the selected instrument
set \(\mathcal K_N\). To relate this target to a genome-wide finite-population
parameter, write
\begin{equation}
\label{eq:error_decomposition}
\hat\rho_{CH}^{(N)}-\rho_{CH}
=
\{\hat\rho_{CH}^{(N)}-\rho_{CH}^{(N)}\}
+
\{\rho_{CH}^{(N)}-\rho_{CH}\}.
\end{equation}
The first term is controlled by Theorem~\ref{thm:main}; below we study the
second term, treating \(\mathcal K_N\) as random.

Let \(\mathcal M_{X,N}:=\{m\in\mathcal M_N:\beta_{X,m}\neq0\}\) and
\(M_{X,N}:=|\mathcal M_{X,N}|\). For \(m\in\mathcal M_{X,N}\), treat
\((\beta_{X,m},\alpha_{1,m},\alpha_{2,m})\) as fixed. Let \(S_{m,N}\) denote
selection, \(J_{m,N}:=1-S_{m,N}\), and \(q_{m,N}:=\Pr(J_{m,N}=1)\). Let
\(h_m>0\) be the unnormalized population-scale precision weight, and set
\(H_N:=\sum_{m\in\mathcal M_{X,N}}h_m\) and
\(H_N^{\mathcal K}:=\sum_{m\in\mathcal M_{X,N}}S_{m,N}h_m\).

For \(\mathcal F:=\{\alpha_1,\alpha_2,\alpha_1^2,\alpha_2^2,\alpha_1\alpha_2\}\),
define, for \(f\in\mathcal F\),
\[
\Pi_{N,f}:=\frac{\sum_m h_mf_m}{H_N},
\qquad
\Pi_{N,f}^{\mathcal K}:=\frac{\sum_m S_{m,N}h_mf_m}{H_N^{\mathcal K}},
\]
where all sums are over \(m\in\mathcal M_{X,N}\). Write
\(\boldsymbol\Pi_N:=(\Pi_{N,f})_{f\in\mathcal F}\) and
\(\boldsymbol\Pi_N^{\mathcal K}:=(\Pi_{N,f}^{\mathcal K})_{f\in\mathcal F}\), and define
\[
\rho_{CH}:=\rho(\boldsymbol\Pi_N),\qquad
\rho_{CH}^{(N)}:=\rho(\boldsymbol\Pi_N^{\mathcal K}),\qquad
\rho(t):=
\frac{t_{\alpha_1\alpha_2}-t_{\alpha_1}t_{\alpha_2}}
{\{(t_{\alpha_1^2}-t_{\alpha_1}^2)(t_{\alpha_2^2}-t_{\alpha_2}^2)\}^{1/2}}.
\]
Under the weight factorization used in Theorem~\ref{thm:main},
\(\rho(\boldsymbol\Pi_N^{\mathcal K})\) coincides with the selected-set oracle
target denoted \(\rho_{CH}^{(N)}\) there.

For \(f\in\mathcal F\), set
\(a_{m,N}^{(f)}:=h_m(f_m-\Pi_{N,f})\) and
\(a_{m,N}:=(a_{m,N}^{(f)})_{f\in\mathcal F}\). Then
\begin{equation}
\label{eq:zero-sum}
\sum_m a_{m,N}^{(f)}=0,\qquad f\in\mathcal F.
\end{equation}
Define \(B_N^{(f)}:=\sum_m q_{m,N}a_{m,N}^{(f)}\),
\(\mathcal S_N^{(f)}:=\sum_m (J_{m,N}-q_{m,N})a_{m,N}^{(f)}\),
\(\psi_N:=\nabla\rho(\boldsymbol\Pi_N)\),
\(\ell_{m,N}:=\psi_N^\top a_{m,N}\), and
\(\bar q_N^{(h)}:=H_N^{-1}\sum_m h_mq_{m,N}\). Finally, set
\[
\Sigma_N:=\sum_m q_{m,N}(1-q_{m,N})a_{m,N}a_{m,N}^\top,
\qquad
V_N:=\psi_N^\top\Sigma_N\psi_N
=\sum_m q_{m,N}(1-q_{m,N})\ell_{m,N}^2.
\]

We assume the following finite-population selection conditions.

\begin{enumerate}[label=\textup{(C*\arabic*)}, leftmargin=*]

\item \emph{Independent selection.}
Conditional on the finite population
\(\{(\beta_{X,m},\alpha_{1,m},\alpha_{2,m}):m\in\mathcal M_{X,N}\}\),
the variables \(\{J_{m,N}:m\in\mathcal M_{X,N}\}\) are mutually independent,
with \(J_{m,N}\sim\mathrm{Bernoulli}(q_{m,N})\).

\item \emph{Population-level regularity.}
There exist constants \(0<c<C<\infty\) such that, for all sufficiently large
\(N\), \(|\Pi_{N,f}|\le C\) and
\(H_N^{-1}\sum_m h_m(f_m-\Pi_{N,f})^2\le C\) for every \(f\in\mathcal F\), while
\(\Pi_{N,\alpha_l^2}-\Pi_{N,\alpha_l}^2\ge c\) for \(l=1,2\). Thus \(\rho\) is
twice continuously differentiable in a neighborhood of \(\boldsymbol\Pi_N\), with
uniformly bounded first and second derivatives.

\item \emph{Rare-miss regime.}
\(\bar q_N^{(h)}\to0\) and \(M_{X,N}\bar q_N^{(h)}\to\infty\).

\item \emph{Bias, scale, and CLT regularity.}
For each \(f\in\mathcal F\), \(B_N^{(f)}=o(\sqrt{V_N})\). Moreover,
\[
\nu_N^2:=\frac{M_{X,N}V_N}{H_N^2\bar q_N^{(h)}}\to\nu^2\in(0,\infty),
\qquad
\frac{\max_m\|a_{m,N}\|^2}{V_N}\to0,
\qquad
\operatorname{tr}(\Sigma_N)=O(V_N).
\]

\end{enumerate}

\begin{theorem}[Population approximation for the selected-set target]
\label{thm:genome_wide_limit}
Under Conditions \textup{(C*1)}--\textup{(C*4)},
\[
\rho_{CH}^{(N)}-\rho_{CH}\to_p0,
\qquad
\sqrt{\frac{M_{X,N}}{\bar q_N^{(h)}}}
\bigl(\rho_{CH}^{(N)}-\rho_{CH}\bigr)
\rightsquigarrow \mathcal{N}(0,\nu^2).
\]
\end{theorem}

\begin{proof}
Throughout, sums are over \(m\in\mathcal M_{X,N}\). By
\eqref{eq:zero-sum} and \(S_{m,N}=1-J_{m,N}\), for each \(f\in\mathcal F\),
\begin{equation}
\label{eq:T-decomp}
\Pi_{N,f}^{\mathcal K}-\Pi_{N,f}
=
\frac{\sum_m S_{m,N}a_{m,N}^{(f)}}{H_N^{\mathcal K}}
=
-\frac{\sum_m J_{m,N}a_{m,N}^{(f)}}{H_N^{\mathcal K}}
=
-\frac{\mathcal S_N^{(f)}+B_N^{(f)}}{H_N^{\mathcal K}}.
\end{equation}

First, \(H_N^{\mathcal K}/H_N\to_p1\). Indeed,
\(\mathbb E(H_N^{\mathcal K}/H_N)=1-\bar q_N^{(h)}\to1\), and since \(0<h_m\le H_N\),
\[
\Var(H_N^{\mathcal K}/H_N)
=
H_N^{-2}\sum_m h_m^2q_{m,N}(1-q_{m,N})
\le
H_N^{-1}\sum_m h_mq_{m,N}
=
\bar q_N^{(h)}
\to0.
\]

Next, for each \(f\in\mathcal F\),
\[
\Var(\mathcal S_N^{(f)})
=
\sum_m q_{m,N}(1-q_{m,N})(a_{m,N}^{(f)})^2
\le
\operatorname{tr}(\Sigma_N)
=
O(V_N).
\]
Together with \(B_N^{(f)}=o(\sqrt{V_N})\), \eqref{eq:T-decomp}, and
\(H_N^{\mathcal K}/H_N\to_p1\), this gives
\(\Pi_{N,f}^{\mathcal K}-\Pi_{N,f}=O_p(\sqrt{V_N}/H_N)\). Since
\(V_N/H_N^2=\nu_N^2\bar q_N^{(h)}/M_{X,N}=o(1)\), we have
\(\boldsymbol\Pi_N^{\mathcal K}-\boldsymbol\Pi_N=o_p(1)\). By \textup{(C*2)}, \(\rho\)
is continuous near \(\boldsymbol\Pi_N\), and hence
\(\rho_{CH}^{(N)}-\rho_{CH}\to_p0\).

It remains to establish the limit distribution. Taylor expansion gives
\[
\rho(\boldsymbol\Pi_N^{\mathcal K})-\rho(\boldsymbol\Pi_N)
=
\psi_N^\top(\boldsymbol\Pi_N^{\mathcal K}-\boldsymbol\Pi_N)+R_N,
\qquad
|R_N|\le C\|\boldsymbol\Pi_N^{\mathcal K}-\boldsymbol\Pi_N\|^2
\]
with probability tending to one. Let
\(\boldsymbol{\mathcal S}_N:=(\mathcal S_N^{(f)})_{f\in\mathcal F}\) and
\(\mathbf B_N:=(B_N^{(f)})_{f\in\mathcal F}\). By \eqref{eq:T-decomp},
\(\boldsymbol\Pi_N^{\mathcal K}-\boldsymbol\Pi_N=-(\boldsymbol{\mathcal S}_N+\mathbf B_N)/H_N^{\mathcal K}\).
Since \(\mathbb E\|\boldsymbol{\mathcal S}_N\|^2=\operatorname{tr}(\Sigma_N)=O(V_N)\) and
\(\|\mathbf B_N\|=o(\sqrt{V_N})\), we obtain
\[
\|\boldsymbol\Pi_N^{\mathcal K}-\boldsymbol\Pi_N\|^2
=
O_p(V_N/H_N^2),
\qquad
R_N=o_p(\sqrt{V_N}/H_N).
\]
Furthermore, using \(\|\psi_N\|=O(1)\),
\[
\psi_N^\top(\boldsymbol\Pi_N^{\mathcal K}-\boldsymbol\Pi_N)
=
-\frac{\sum_m (J_{m,N}-q_{m,N})\ell_{m,N}}{H_N^{\mathcal K}}
+
o_p(\sqrt{V_N}/H_N).
\]
Therefore
\begin{equation}
\label{eq:linear-limit-reduction}
\frac{H_N}{\sqrt{V_N}}
\bigl(\rho_{CH}^{(N)}-\rho_{CH}\bigr)
=
-\frac{\sum_m (J_{m,N}-q_{m,N})\ell_{m,N}}{\sqrt{V_N}}
+
o_p(1).
\end{equation}

Let \(\xi_{m,N}:=(J_{m,N}-q_{m,N})\ell_{m,N}\). Conditional on the finite
population, the \(\xi_{m,N}\)'s are independent, mean-zero, and satisfy
\[
\sum_m\Var(\xi_{m,N})
=
\sum_m q_{m,N}(1-q_{m,N})\ell_{m,N}^2
=
V_N.
\]
Moreover, \(|\xi_{m,N}|\le|\ell_{m,N}|\le\|\psi_N\|\,\|a_{m,N}\|\), so by
\textup{(C*2)} and \textup{(C*4)},
\[
\frac{\max_m\xi_{m,N}^2}{V_N}
\le
O(1)\frac{\max_m\|a_{m,N}\|^2}{V_N}
\to0.
\]
The Lindeberg condition follows, and hence
\(\sum_m\xi_{m,N}/\sqrt{V_N}\rightsquigarrow \mathcal{N}(0,1)\). Together with
\eqref{eq:linear-limit-reduction},
\[
\frac{H_N}{\sqrt{V_N}}
\bigl(\rho_{CH}^{(N)}-\rho_{CH}\bigr)
\rightsquigarrow \mathcal{N}(0,1).
\]
Finally, since
\[
\sqrt{\frac{M_{X,N}}{\bar q_N^{(h)}}}\frac{\sqrt{V_N}}{H_N}
=
\nu_N\to\nu,
\]
Slutsky's theorem gives
\[
\sqrt{\frac{M_{X,N}}{\bar q_N^{(h)}}}
\bigl(\rho_{CH}^{(N)}-\rho_{CH}\bigr)
\rightsquigarrow \mathcal{N}(0,\nu^2).
\]
\end{proof}

\begin{remark}[Combined regimes and rates]
\label{rmk:regimes_combined}
The relative magnitudes of the two error sources in the decomposition~\eqref{eq:error_decomposition} depend on instrument strength. Under adaptive Bonferroni
selection with
\(q_{m,N}=\Phi(c_N-\sqrt N\,\beta_{X,m})+\Phi(-c_N-\sqrt N\,\beta_{X,m})\),
there are two useful regimes:
\begin{enumerate}[label=(\alph*),leftmargin=*]
\item \emph{Strong-instrument regime}
(\(|\sqrt N\,\beta_{X,m}|-c_N\to\infty\) uniformly): \(q_{m,N}\) is
exponentially small, \(M_{X,N}\bar q_N^{(h)}\to0\), the selection error
vanishes faster than any polynomial rate, and Theorem~\ref{thm:main} alone
delivers genome-wide inference at rate \(\sqrt{NK}\).

\item \emph{Rare-miss regime}
(\(\bar q_N^{(h)}\to0\), \(M_{X,N}\bar q_N^{(h)}\to\infty\)): the missed-SNP
fraction vanishes but the effective number of missed SNPs diverges, yielding a
non-degenerate Gaussian selection-error limit at rate
\(\sqrt{M_{X,N}/\bar q_N^{(h)}}\) by
Theorem~\ref{thm:genome_wide_limit}.
\end{enumerate}
\end{remark}
\subsection{\texorpdfstring{Proof of \cref{lem:dov_scaling}}{Proof of Lemma 1}}
\label{app:lemma1_full}
In the notation of the generative model \eqref{sim_frame} of \cref{sec:simulation}, in which \(\phi_{a,m}\) is the effect of SNP \(m\) on the confounder \(U_a\) and \(\delta_{l,m}\) its intrinsic direct effect on \(Y_l\), evaluated at the selected instruments, put
\(\eta_{l,m}=\delta_{l,m}/\beta_{X,m}\) and \(\psi_{a,m}=\phi_{a,m}/\beta_{X,m}\) for the ratio-scale
direct and confounder effects, and write \(\theta_{U,l}:=\theta_{U\to Y_l}\) for the confounder-to-outcome effect of \cref{sec:simulation}. For \(a\in\{0,1,2\}\) let
\(\mathcal U_a\) index the selected instruments that are invalid through \(U_a\), where \(U_0\) is the shared confounder and \(U_1,U_2\) the outcome-specific ones of \cref{sec:simulation}. We parametrise the split
of the invalid instruments as
\[
q_0=D_{\mathrm{ov}},\qquad q_1=(1-q_0)\,r_1,\qquad q_2=(1-q_0)\,r_2,\qquad r_1+r_2=1,
\]
which is the constraint \(q_0+q_1+q_2=1\) of \cref{lem:dov_scaling} written so that \(D_{\mathrm{ov}}\) may
be varied while $(r_1,r_2)$ held fixed. 
\paragraph{Regularity conditions.}
\begin{enumerate}[label=\textup{(D\arabic*)}, leftmargin=*, itemsep=1pt]
\item \(|\beta_{X,m}|\) is bounded away from zero and, after centring, \(\Var(\eta_{l,m})=\sigma_{\eta_l}^2\),
\(\Var(\psi_{a,m})=\sigma_{U,R}^2\), \(\Cov(\eta_{1,m},\eta_{2,m})=0\), and \(\Cov(\eta_{l,m},\psi_{a,m})=0\).
\item Near-uniform weighting: the precision weights satisfy \(\max_{m\le K}|K w_m-1|\xrightarrow{p}0\).
\item The limiting marginal variances of the ratio-scale pleiotropic components are positive.
\end{enumerate}
\paragraph{Proof.} The ratio-scale pleiotropic component
\(\alpha_{l,m}=(\beta_{Y_l,m}-\theta_{X\to Y_l}\beta_{X,m})/\beta_{X,m}\) satisfies, by \eqref{sim_frame},
\[
\alpha_{1,m}=\eta_{1,m}+\theta_{U,1}\psi_{0,m}\mathbf 1\{m\in\mathcal U_0\}+\theta_{U,1}\psi_{1,m}\mathbf 1\{m\in\mathcal U_1\},
\]
\[
\alpha_{2,m}=\eta_{2,m}+\theta_{U,2}\psi_{0,m}\mathbf 1\{m\in\mathcal U_0\}+\theta_{U,2}\psi_{2,m}\mathbf 1\{m\in\mathcal U_2\},
\]
so only the shared category \(\mathcal U_0\) enters both \(\alpha_{1,m}\) and \(\alpha_{2,m}\). By
\textup{(D2)}, weighted empirical moments have the same limits as unweighted ones: for each of the five functions \(f(z_1,z_2)\in\{z_1,z_2,z_1^2,z_2^2,z_1z_2\}\),
\[
\Big|\sum_{m=1}^K\big(w_m-\tfrac1K\big)f(\alpha_{1,m},\alpha_{2,m})\Big|
\le\max_m|Kw_m-1|\cdot K^{-1}\sum_{m=1}^K|f(\alpha_{1,m},\alpha_{2,m})|=o_p(1),
\]
so the weighted covariance and variances entering \(\rho_{CH}^{(N)}\) converge in probability to the
corresponding ratio-scale mixture moments under the data-generating process. Since only \(\mathcal U_0\), of unconditional proportion \(\bar\pi q_0\), contributes cross-outcome
covariance \textup{(D1)},
\[
\Cov(\alpha_{1,m},\alpha_{2,m})=\bar\pi q_0\,\theta_{U,1}\theta_{U,2}\,\sigma_{U,R}^2,
\qquad
\Var(\alpha_{l,m})=\sigma_{\eta_l}^2+\bar\pi(q_0+q_l)\,\theta_{U,l}^2\,\sigma_{U,R}^2,
\qquad l=1,2.
\]
Substituting into \(\rho_{CH}=\Cov(\alpha_1,\alpha_2)/\{\Var(\alpha_1)\Var(\alpha_2)\}^{1/2}\) gives
\eqref{eq:rho_dov_main}, the constants there being
\(a=\theta_{U,1}\sigma_{U,R}\), \(b=\theta_{U,2}\sigma_{U,R}\), \(c=\sigma_{\eta_1}\) and
\(d=\sigma_{\eta_2}\); the denominators are positive by \textup{(D3)}. In particular \(ab=\theta_{U,1}\theta_{U,2}\sigma_{U,R}^2\), so whenever
\(\bar\pi q_0>0\) the sign of \(\rho_{CH}\) is the sign of \(\theta_{U,1}\theta_{U,2}\).

For the monotonicity, hold \((r_1,r_2)\) fixed and set \(x=q_0=D_{\mathrm{ov}}\),
\[
A_l=\sigma_{\eta_l}^2+\bar\pi r_l\,\theta_{U,l}^2\sigma_{U,R}^2>0,
\qquad
B_l=\bar\pi(1-r_l)\,\theta_{U,l}^2\sigma_{U,R}^2\ge0,
\]
so that \(q_0+q_l=r_l+x(1-r_l)\) and hence \(\Var(\alpha_{l,m})=A_l+B_lx\). Up to the constant factor
\(\theta_{U,1}\theta_{U,2}\sigma_{U,R}^2\), whose sign is fixed, the limit \eqref{eq:rho_dov_main} equals
\[
\varrho(x)=\frac{\bar\pi x}{\{(A_1+B_1x)(A_2+B_2x)\}^{1/2}},\qquad 0\le x\le1 .
\]
With \(u(x)=(A_1+B_1x)(A_2+B_2x)\) we have \(\varrho'(x)=\tfrac12\bar\pi\,u(x)^{-3/2}\{2u(x)-xu'(x)\}\)
and \(2u(x)-xu'(x)=2A_1A_2+(A_1B_2+A_2B_1)x\), so
\[
\varrho'(x)=\frac{\bar\pi\{2A_1A_2+(A_1B_2+A_2B_1)x\}}{2\{(A_1+B_1x)(A_2+B_2x)\}^{3/2}}>0,
\qquad 0\le x\le 1 .
\]
Hence \(\rho_{CH}\) is strictly increasing in \(D_{\mathrm{ov}}\) when \(\theta_{U,1}\theta_{U,2}>0\) and
strictly decreasing when \(\theta_{U,1}\theta_{U,2}<0\).

\subsection{Simulation Settings}
\label{sec:sim_settings}

The generative model for the simulations is described in the main text (Section~\ref{sec:simulation}); here we provide the specific parameter values and implementation details. The true causal effects were set as $\theta_{X\to Y_1} \in \{-0.2, -0.1, 0, 0.1, 0.2\}$ and $\theta_{X\to Y_2} \in \{-0.3, 0, 0.3\}$, and unmeasured confounding was parameterized as $\theta_{U\to X}, \theta_{U\to Y_1}, \theta_{U\to Y_2} \in \{0.3, 0.5\}$, reflecting moderate to strong confounding commonly observed in complex trait studies.

To reflect realistic GWAS settings with varying instrument strength, exposure sample sizes varied as $N_X \in \{50k, 80k, 100k, 150k, 200k, 500k, 1000k\}$. Outcome sample sizes were set as $N_{Y_1} = 0.5 \times N_X$ and $N_{Y_2} = N_X$, representing scenarios where one outcome study is smaller than the exposure study while an auxiliary outcome has comparable sample size. We set the proportion of exposure-associated SNPs to $p = 0.02$ (approximately $4k$ variants), and among these, the proportion of invalid instruments was varied as $1-\pi \in \{0.1, 0.3, 0.5, 0.7\}$, capturing settings from weak to severe instrument invalidity. The overlap of invalid instruments between outcomes was controlled through $D_{\mathrm{ov}} \in \{0.5, 0.75, 1.0\}$, representing moderate to complete overlap. Additionally, 1\% of SNPs had direct pleiotropic effects on outcomes ($\delta_{lm} \neq 0$) but no association with the exposure.

Direct genetic effect sizes were drawn from normal distributions with variances calibrated to produce realistic instrument strengths: $\sigma_X^2 = 5 \times 10^{-5}$ for valid instruments, $\sigma_U^2 = 1 \times 10^{-4}$ for confounder effects, and $\sigma_{Y_1}^2 = \sigma_{Y_2}^2 = 5 \times 10^{-5}$ for direct pleiotropic effects. We considered both balanced pleiotropy ($\mu_{Y_1} = \mu_{Y_2} = 0$) and directional pleiotropy ($\mu_{Y_1} = 0.005$, $\mu_{Y_2} = 0.003$), where pleiotropic pathways systematically bias effects in one direction.

Replicates with fewer than three significant instruments were excluded to ensure estimator stability. For each parameter configuration, we performed 100 replicates. We compared our proposed IB methods (IB-Mode and IB-PRESSO) against eight existing MR estimators: IVW, MR-Egger, weighted median, weighted mode-based estimation, contamination mixture, MR-Mix, MR-cML, and MR-PRESSO. For both IB-Mode and weighted mode, we used the weighting scheme with $\phi=1$ and computed bootstrap standard errors with 100 bootstrap samples. For MR-PRESSO and IB-PRESSO, we used 10,000 bootstrap samples for outlier detection with significance threshold 0.05. Outlier-corrected estimates were obtained by refitting the weighted inverse-variance regression on the retained instruments. For MR-Mix, we evaluated the causal effect over a grid of candidate values from $-0.5$ to $0.5$ with step size 0.01, and initialized the algorithm with invalid instrument proportion 0.6 and pleiotropy variance $10^{-5}$. The contamination mixture method provides 95\% confidence intervals derived from profile likelihood ratio test inversion rather than standard errors; we approximated standard errors as the confidence interval width divided by $2 \times 1.96$ for comparison purposes,  assuming approximate normality. 

As a secondary analysis, we additionally evaluated the Bayesian multi-outcome method MR2. For MR2, the two outcome summary statistics were modelled jointly using a Bayesian framework with $\texttt{EVgamma}=0.5$. We ran 4,000 MCMC iterations with a burn-in of 1,000 and thinning interval of 5, and obtained posterior summaries from the processed MCMC output. The posterior mean was used as the point estimate, with 95\% credible intervals reported. Statistical significance was assessed using the posterior inclusion probability (PIP), with a rejection threshold calibrated under the null to control the type I error rate.

All simulations were implemented in \texttt{R} version 4.3.0 using the \texttt{MendelianRandomization} package for existing methods and custom implementations for IB-Mode and IB-PRESSO.

\newpage
\subsection{Additional Real Data Analysis}

Exposure--outcome pairs from \citet{morrison2020mendelian} for which no candidate auxiliary trait showed significant coheterogeneity with the primary outcome, as required at the auxiliary-selection step of \cref{algo}, were excluded from the analysis.

\textbf{Primary exposure BMI.} IB-Mode demonstrated strong agreement with other robust MR methods for outcomes where causal effects are well-established, including CAD and T2D, with efficiency comparable to weighted mode.
 For BMI, the principal benefit of instrument borrowing is a more precise separation of causal from pleiotropic effects rather than a large efficiency gain, consistent with the exposure-dependent efficiency pattern in \cref{causalgain}. These patterns reflect BMI's complex genetic architecture: many BMI-associated variants exhibit horizontal pleiotropy through metabolic and inflammatory pathways that differentially affect cardiometabolic outcomes. This provides a compelling case where, by leveraging auxiliary outcomes, IB methods enable a more precise separation of causal from pleiotropic effects than standard MR approaches.

\textbf{Primary exposure blood pressure traits.} IB-Mode successfully identified established causal effects of DBP and SBP on CAD and stroke while correctly returning null results for asthma, where no causal relationship is supported \citep{morrison2020mendelian}.
Blood pressure genetics are characterized by highly pleiotropic variants affecting vascular remodeling, sodium homeostasis, and sympathetic tone \citep{singh2021systematic, bertorello2015increased}. These mechanisms create structured pleiotropy where invalid instruments for one cardiovascular outcome tend to be invalid for related outcomes. IB-Mode's ability to leverage this coheterogeneity structure through auxiliary outcomes yields substantial efficiency gains for DBP and more modest gains for SBP relative to weighted mode (\cref{causalgain}).

\textbf{Primary exposure cholesterol traits.} The cholesterol analyses provide perhaps the most compelling evidence for IB-Mode's balanced performance.
 While IB-Mode correctly identified LDL's causal effect on CAD, consistent with decades of trial evidence, it simultaneously avoided false positives for HDL across all outcomes. This is particularly noteworthy because HDL represents a well-characterized example of a non-causal biomarker: observational studies consistently show protective associations, yet randomized trials of HDL-raising interventions have uniformly failed to reduce cardiovascular events \citep{morrison2020mendelian}. In contrast, weighted mode incorrectly identified a significant HDL effect on stroke. HDL genetics are notoriously complex, with variants in lipid metabolism genes (e.g., CETP, LIPC) exhibiting extensive pleiotropy across multiple lipid fractions and inflammatory markers \citep{millwood2018association, kullo2005pleiotropic}. HDL-raising variants often affect triglycerides and LDL simultaneously, creating confounding that mimics a causal effect. IB-Mode avoids this trap by using auxiliary traits to identify instruments whose effects are consistent with direct causation versus those reflecting broader metabolic dysregulation, thereby maintaining appropriate Type I error control where weighted mode fails.

%% file: supplement/1suppl_figs.tex
\newpage
\section{Supplemental Figures}
\label{Supp-Fig}
\setcounter{figure}{0}
\renewcommand{\thefigure}{S\arabic{figure}}

\vspace{1em}

\begin{figure}[H]
\centering
\begin{mdframed}[
    linewidth=1.5pt,
    roundcorner=8pt,
    innertopmargin=10pt,
    innerbottommargin=10pt,
    innerleftmargin=0pt,
    innerrightmargin=10pt,
    backgroundcolor=white
]

\begin{tikzpicture}[remember picture]
\node[inner sep=0pt] (tabular) {
\setlength{\tabcolsep}{8pt}
\renewcommand{\arraystretch}{1}
\begin{tabular}{cc}
\hspace{-0.8cm}
\begin{minipage}[c][5cm][t]{0.44\textwidth}
\caption*{\textbf{(A)}}
    \centering
    \tikz[remember picture] \node[inner sep=0pt] (panelA) {
    \begin{tikzpicture}[node distance=0.8cm and 1.5cm, every node/.style={small}, bend angle=15, scale = 0.6]
    \begin{scope}[shift={(-5,0)}]
      \draw[thick] (-1, 2) rectangle (1,-2);
      \draw[thick] (-1, 0) -- (1,0);
      \fill[green!30] (-1, 2) rectangle (1,0);
      \fill[red!30]   (-1, 0) rectangle (1,-2);
      \node[above, font=\bfseries] at (0,2.2) {Candidate \\ instruments for $X$};
      \node[small, align=center] at (0, 0.8) {Valid};
      \node[small, align=center] at (0,-0.8) {Invalid};
      \coordinate (ValidTop) at (0,1.3);
      \coordinate (InvalidMid) at (0,-1.3);
    \end{scope}
    \node[exposure] (X) at (0,0) {Exposure\\(X)};
    \node[outcome] (Y1) at ([xshift=5cm, yshift=1.5cm]X) {Outcome\\$Y_{\text{main}}$};
    \node[outcome] (Y2) at ([xshift=5cm, yshift=-1.5cm]X) {Outcome\\$Y_{\text{aux}}$};
    \node[confounder, below=of X, yshift=-0.5cm] (U) {Confounder\\(U)};
    \draw[arrow, thick, draw=black!70!black] (ValidTop) to[bend left=15] (X);
    \draw[arrow, thick, draw=black!70!black] (InvalidMid) to[bend right=15] (X);
    \draw[violation, thick, draw=red!70!black] (InvalidMid) to[bend right=30] (U);
    \draw[violation, thick, draw=red!70!black] (U) -- (X);
    \draw[violation, thick, draw=red!70!black] (U) to (Y1);
    \draw[violation, thick, draw=red!70!black] (U) to (Y2);
    \draw[arrow] (X) -- (Y1) node[midway, above, small] {$\theta_1$};
    \draw[arrow] (X) -- (Y2) node[midway, below, small] {$\theta_2$};
    \draw[dotted, thick] (Y2) to[bend right=15] (Y1);
    \end{tikzpicture}
    };
\end{minipage}
&
\vspace{1cm}
\hspace{0.3cm}
\begin{minipage}[c][5cm][t]{0.42\textwidth}
\caption*{\textbf{(B)}}
    \centering
    \vspace{2mm}
    \tikz[remember picture] \node[inner sep=0pt] (panelB) {
    \begingroup
    \setlength{\tabcolsep}{3.8pt}
    \renewcommand{\arraystretch}{1.07}
    \begin{tabular}{|c|c|c|c|c|c|}
      \hline 
      \textbf{SNP} & $\hat{\beta}_X$ & $\hat{\beta}_{Y_{main}}$ & $\hat{\beta}_{Y_{aux}}$ & $\hat{\theta}_1$ & $\hat{\theta}_2$\\
      \hline
      \textit{SNP}1 & 0.05 & 0.025 & 0.015 & 0.5 & 0.3 \\
      \textit{SNP}2 & 0.10 & 0.07 & 0.05 & 0.7 & 0.5 \\
      \textit{SNP}3 & 0.20 & 0.05 & 0.10 & 0.25 & 0.5 \\
      \textit{SNP}4 & 0.10 & 0.02 & 0.09 & 0.2 & 0.9 \\
      \textit{SNP}5 & 0.10 & 0.05 & 0.10 & 0.5 & 1 \\
      $\cdots$ & $\cdots$ & $\cdots$ & $\cdots$ & $\cdots$ & $\cdots$ \\
      \textit{SNP}K & $\beta_X^{(N)}$ & $\beta_{Y_{main}}^{(N)}$ & $\beta_{Y_{aux}}^{(N)}$ & $\frac{\beta_{Y_{main}}^{(N)}}{\beta_X^{(N)}}$ & $\frac{\beta_{Y_{aux}}^{(N)}}{\beta_X^{(N)}}$ \\
      \hline
    \end{tabular}
    \endgroup
    };
\end{minipage}
\\[10pt]
\begin{minipage}[c][5cm][t]{0.44\textwidth}
\vspace{-1cm}
\caption*{\textbf{(D)}}
    \centering
    \tikz[remember picture] \node[inner sep=0pt] (panelD) {
    \includegraphics[width=0.97\linewidth]{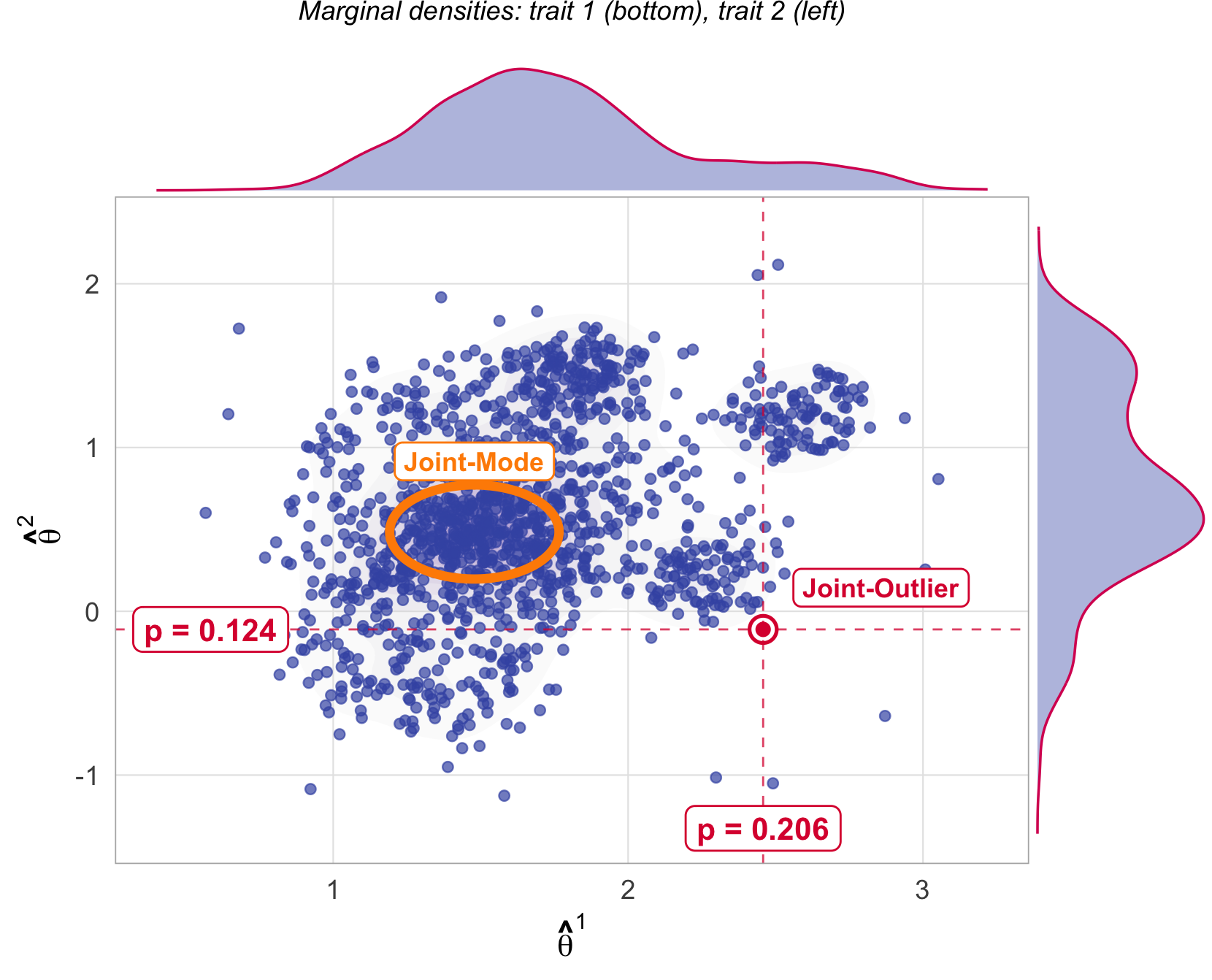}
    };
\end{minipage}
&
\hspace{-1.0cm}
\vspace{0.5cm}
\begin{minipage}[c][5cm][t]{0.40\textwidth}
\caption*{\textbf{(C)}}
    \centering
    \tikz[remember picture] \node[inner sep=0pt] (panelC) {
    \begin{tikzpicture}[x=1.30cm, y=0.8cm]
        \node[anchor=south, font=\small] at (0,5.05) {Trait 1};
        \node[anchor=south, font=\small] at (1,5.05) {Trait 2};
        \foreach \row/\name in {4/Aux 1,3/Aux 2,2/Aux 3,1/Aux 4,0/Aux 5} {
            \node[anchor=east, font=\small] at (-0.5,\row) {\name};
        }
        \foreach \row/\v in {4/0.15,3/0.82,2/0.35,1/0.60,0/0.05} {
            \pgfmathsetmacro{\rf}{255-(255-160)*\v}
            \pgfmathsetmacro{\gf}{240-(240-0)*\v}
            \pgfmathsetmacro{\bf}{240-(240-0)*\v}
            \definecolor{cellcol}{RGB}{\rf,\gf,\bf}
            \node[fill=cellcol, draw=black, minimum width=0.95cm, minimum height=0.9cm, rounded corners=2pt] at (0,\row) {};
            \node[font=\scriptsize] at (0,\row) {\pgfmathprintnumber[fixed,precision=2]{\v}};
        }
        \foreach \row/\v in {4/0.40,3/0.20,2/0.90,1/0.50,0/0.75} {
            \pgfmathsetmacro{\rf}{255-(255-160)*\v}
            \pgfmathsetmacro{\gf}{240-(240-0)*\v}
            \pgfmathsetmacro{\bf}{240-(240-0)*\v}
            \definecolor{cellcol}{RGB}{\rf,\gf,\bf}
            \node[fill=cellcol, draw=black, minimum width=0.95cm, minimum height=0.9cm, rounded corners=2pt] at (1,\row) {};
            \node[font=\scriptsize] at (1,\row) {\pgfmathprintnumber[fixed,precision=2]{\v}};
        }
        \draw[thick, orange!85!black, rounded corners=2pt] ($(0,3)+(-0.46,-0.46)$) rectangle ($(0,3)+(0.46,0.46)$);
        \draw[thick, blue!80!black, rounded corners=2pt]   ($(1,2)+(-0.46,-0.46)$) rectangle ($(1,2)+(0.46,0.46)$);
        \begin{scope}[shift={(1.75,0)}]
            \draw[thick] (0,0) rectangle (0.32,2.7);
            \foreach \yy in {0,0.07,...,2.63}{
                \pgfmathsetmacro{\t}{\yy/2.7}
                \pgfmathsetmacro{\r}{255-(255-160)*\t}
                \pgfmathsetmacro{\g}{240-(240-0)*\t}
                \pgfmathsetmacro{\b}{240-(240-0)*\t}
                \definecolor{barc}{RGB}{\r,\g,\b}
                \draw[fill=barc, draw=none] (0,\yy) rectangle (0.32,{\yy+0.07});
            }
            \node[anchor=west, font=\tiny] at (0.33,0) {0};
            \node[anchor=west, font=\tiny] at (0.33,2.7) {1};
            \node[anchor=west, font=\small] at (0.32,1.35) {$\hat\rho_{CH}^{(N)}$};
        \end{scope}
    \end{tikzpicture}
    };
\end{minipage}
\end{tabular}
};

\begin{scope}[overlay]
    \draw[->, very thick, black, line width=2pt] 
        (panelB.south) -- (panelC.north);
    \draw[->, very thick, black, line width=2pt] 
        (panelB.south west) to[bend left=12] (panelD.north east);
    \draw[->, very thick, black, line width=2pt] 
        (panelC.west) -- (panelD.east);
\end{scope}

\end{tikzpicture}
\end{mdframed}
\vspace{-1.5ex}
\caption{
  \textbf{Graphical Summary:} \textit{ \textbf{The Instrument Borrowing approach systematically leverages genetic associations with related traits, enabling improved identification and estimation of causal effects in Mendelian randomization analyses.} \textbf{(A)} DAG illustrating sharing of valid/invalid instruments across two outcome traits due to an underlying common confounder.  \textbf{(B)} GWAS summary statistics and corresponding ratio-estimates for two outcome traits in relationship to an exposure ($X$). \textbf{(C)} Selection of appropriate auxiliary traits for instrumental borrowing is guided by the proposed coheterogeneity statistic $\hat\rho_{CH}^{(N)}$. \textbf{(D)} Scatterplot of the ratio-estimates across the two outcome traits reveals how mode identification and outlier detection are improved through joint analysis.}
}
\label{graph_summ}
\end{figure}

\begin{figure}[h]
\centering
\begin{tikzpicture}[scale=1.15, every node/.style={scale=0.9}]

\draw[thick, line width=1.2pt] (0, 0.5) rectangle (1.8, 10);
\node[align=center, font=\bfseries, rotate=90] at (0.9, 5.25) {All SNPs ($M = 200k$)};

\draw[thick, fill=blue!10, line width=1pt] (2.0, 7.8) rectangle (3.8, 10);
\node[align=center, font=\small] at (2.9, 8.9) {Valid IVs\\[2pt]($\gamma_m \neq 0$)};

\draw[thick, fill=red!10, line width=1pt] (2.0, 5.3) rectangle (3.8, 7.6);
\node[align=center, font=\small] at (2.9, 6.45) {Invalid IVs\\[2pt]($\phi_{lm} \neq 0$)};

\draw[thick, fill=green!10, line width=1pt] (2.0, 2.8) rectangle (3.8, 5.1);
\node[align=center, font=\small] at (2.9, 3.95) {Pleiotropic\\[2pt]($\delta_{lm} \neq 0$)};

\draw[thick, fill=gray!10, line width=1pt] (2.0, 0.5) rectangle (3.8, 2.6);
\node[align=center, font=\small] at (2.9, 1.55) {Null\\[2pt]variants};

\draw[->, thick, line width=1.2pt] (3.8, 8.9) -- (4.2, 8.9);
\draw[->, thick, line width=1.2pt] (3.8, 6.45) -- (4.2, 6.45);
\draw[->, thick, line width=1.2pt] (3.8, 3.95) -- (4.2, 3.95);

\draw[thick, fill=blue!20, line width=1pt] (4.4, 7.8) rectangle (6.5, 10);
\node[align=center, font=\small] at (5.45, 8.9) {Effect on $X$ only\\[4pt]No pleiotropy};

\draw[thick, fill=red!20, line width=1pt] (4.4, 6.7) rectangle (6.5, 7.6);
\node[align=center, font=\small] at (5.45, 7.15) {Shared\\[2pt]via $U_0$};

\draw[thick, fill=red!20, line width=1pt] (4.4, 5.9) rectangle (6.5, 6.6);
\node[align=center, font=\small] at (5.45, 6.25) {$Y_1$-specific\\[2pt]via $U_1$};

\draw[thick, fill=red!20, line width=1pt] (4.4, 5.1) rectangle (6.5, 5.8);
\node[align=center, font=\small] at (5.45, 5.45) {$Y_2$-specific\\[2pt]via $U_2$};

\draw[thick, fill=green!20, line width=1pt] (4.4, 3.9) rectangle (6.5, 4.8);
\node[align=center, font=\small] at (5.45, 4.35) {Both\\[2pt]outcomes};

\draw[thick, fill=green!20, line width=1pt] (4.4, 3.1) rectangle (6.5, 3.8);
\node[align=center, font=\small] at (5.45, 3.45) {$Y_1$ only};

\draw[thick, fill=green!20, line width=1pt] (4.4, 2.3) rectangle (6.5, 3.0);
\node[align=center, font=\small] at (5.45, 2.65) {$Y_2$ only};

\node[draw, thick, fill=cyan!15, ellipse, minimum width=2.4cm, minimum height=1.6cm, line width=1.2pt] (X) at (10.5, 8.5) {$\boldsymbol{X}$};

\node[draw, thick, fill=yellow!15, ellipse, minimum width=2.2cm, minimum height=1.5cm, line width=1.2pt] (Y1) at (8.5, 5) {$\boldsymbol{Y_1}$};
\node[draw, thick, fill=yellow!15, ellipse, minimum width=2.2cm, minimum height=1.5cm, line width=1.2pt] (Y2) at (12.5, 5) {$\boldsymbol{Y_2}$};

\node[draw, thick, fill=orange!15, rounded corners=3pt, minimum width=1.7cm, minimum height=1.4cm, line width=1pt] (U1) at (7, 1.2) {\small $U_1$};
\node[draw, thick, fill=purple!15, rounded corners=3pt, minimum width=1.7cm, minimum height=1.4cm, line width=1pt] (U0) at (10.5, 1.2) {\small $U_0$};
\node[draw, thick, fill=brown!15, rounded corners=3pt, minimum width=1.7cm, minimum height=1.4cm, line width=1pt] (U2) at (14, 1.2) {\small $U_2$};

\draw[->, very thick, dotted, blue!60!black, line width=1.6pt] (6.5, 9.2) .. controls (8, 9.5) and (9, 9.2) .. (X.north west);

\draw[->, thick, dotted, red!60!black, line width=1.3pt] (6.5, 7.4) .. controls (7.5, 8.8) and (8.5, 8.9) .. (X.west);
\draw[->, thick, dotted, red!60!black, line width=1.3pt] (6.5, 6.25) .. controls (7.8, 7.8) and (9, 8.2) .. (X.south west);
\draw[->, thick, dotted, red!60!black, line width=1.3pt] (6.5, 5.45) .. controls (8, 6.8) and (9.2, 7.5) .. (X.south);

\draw[->, thick, dotted, purple!60!black, line width=1.3pt] (6.5, 7.15) .. controls (8, 5.5) and (9.5, 3) .. (U0.north);

\draw[->, thick, dotted, orange!60!black, line width=1.3pt] (6.5, 6.25) .. controls (6.8, 4.5) and (6.8, 3) .. (U1.north);

\draw[->, thick, dotted, brown!60!black, line width=1.3pt] (6.5, 5.45) .. controls (10.5, 4.5) and (13, 3) .. (U2.north);

\draw[->, thick, dotted, green!60!black, line width=1.3pt] (6.5, 4.35) .. controls (8.5, 3.5) and (9.5, 2.5) .. (U0.north west);

\draw[->, thick, dotted, green!60!black, line width=1.3pt] (6.5, 3.45) .. controls (6.8, 2.5) .. (U1.north);

\draw[->, thick, dotted, green!60!black, line width=1.3pt] (6.5, 2.65) .. controls (11, 2) .. (U2.north);

\draw[->, very thick, dashed, black, line width=2pt] (X.south west) -- (Y1.north);
\node[left, font=\normalsize\bfseries] at (9.2, 7) {$\theta_{X\to Y_1}$};

\draw[->, very thick, dashed, black, line width=2pt] (X.south east) -- (Y2.north);
\node[right, font=\normalsize\bfseries] at (11.8, 7) {$\theta_{X\to Y_2}$};

\draw[->, thick, orange!60!black, line width=1.4pt] (U1.north) .. controls (7.5, 5) and (8.5, 7) .. (X.south west);
\draw[->, thick, purple!60!black, line width=1.4pt] (U0.north) .. controls (10.5, 4.5) .. (X.south);
\draw[->, thick, brown!60!black, line width=1.4pt] (U2.north) .. controls (13.5, 5) and (12.5, 7) .. (X.south east);

\draw[->, thick, orange!60!black, line width=1.4pt] (U1.north) -- (Y1.south);

\draw[->, thick, purple!60!black, line width=1.4pt] (U0.north) .. controls (9, 3) .. (Y1.south);
\draw[->, thick, purple!60!black, line width=1.4pt] (U0.north) .. controls (12, 3) .. (Y2.south);

\draw[->, thick, brown!60!black, line width=1.4pt] (U2.north) -- (Y2.south);

\end{tikzpicture}
\caption{Simulation structure showing SNP categories, confounding pathways, and trait relationships. SNPs are partitioned into valid IVs (direct effect on $X$ only via $\gamma_m$), invalid IVs (effect on $X$ and confounders via $\gamma_m$ and $\phi_{lm}$), pleiotropic variants (direct effect on the outcomes via $\delta_{lm}$, with no association with $X$), and null variants. Invalid instruments are subpartitioned by confounder association: shared confounder $U_0$ (affects both outcomes and $X$), or outcome-specific confounders $U_1$ and $U_2$ (affect $Y_1$ or $Y_2$ respectively, and $X$). Pleiotropic variants affect the outcomes directly and are not associated with $X$. The overlap parameter $D_{\mathrm{ov}}$ controls the proportion of invalid instruments associated with $U_0$. Dotted arrows show direct genetic effects; solid arrows show confounding pathways; dashed arrows show causal effects.}
\label{fig:simulation_structure}
\end{figure}

\begin{figure}[!htbp]
    \centering
    \includegraphics[width=\linewidth, trim=0 0 0 55, clip]{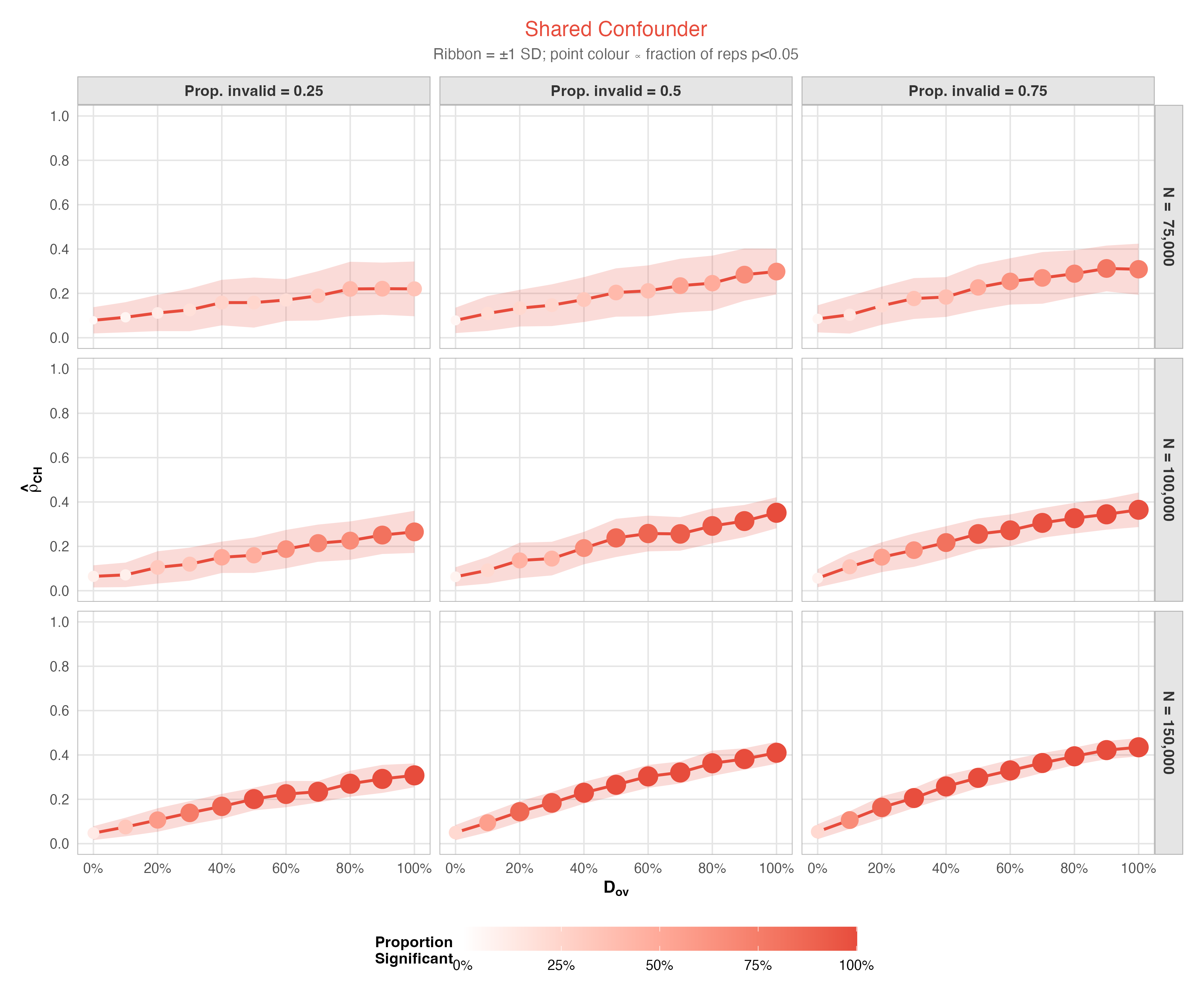}
    \caption{\textit{\textbf{Coheterogeneity statistics tracks overlap in shared confounders.}
Each panel shows the mean absolute coheterogeneity statistic ($|\hat{\rho}_{CH}^{(N)}|$) across 100 simulations, plotted against the overlap in invalid instruments ($D_{\mathrm{ov}}$) between two outcomes. In this scenario, invalid instruments share a common confounder but exhibit no directional pleiotropy. Point size and color reflect the fraction of replicates with significant coheterogeneity ($p < 0.05$).}}
    \label{fig:cohet-Dcov}
\end{figure}

\begin{figure}[!htbp]
    \centering
    \includegraphics[width=\linewidth, trim=0 0 0 55, clip]{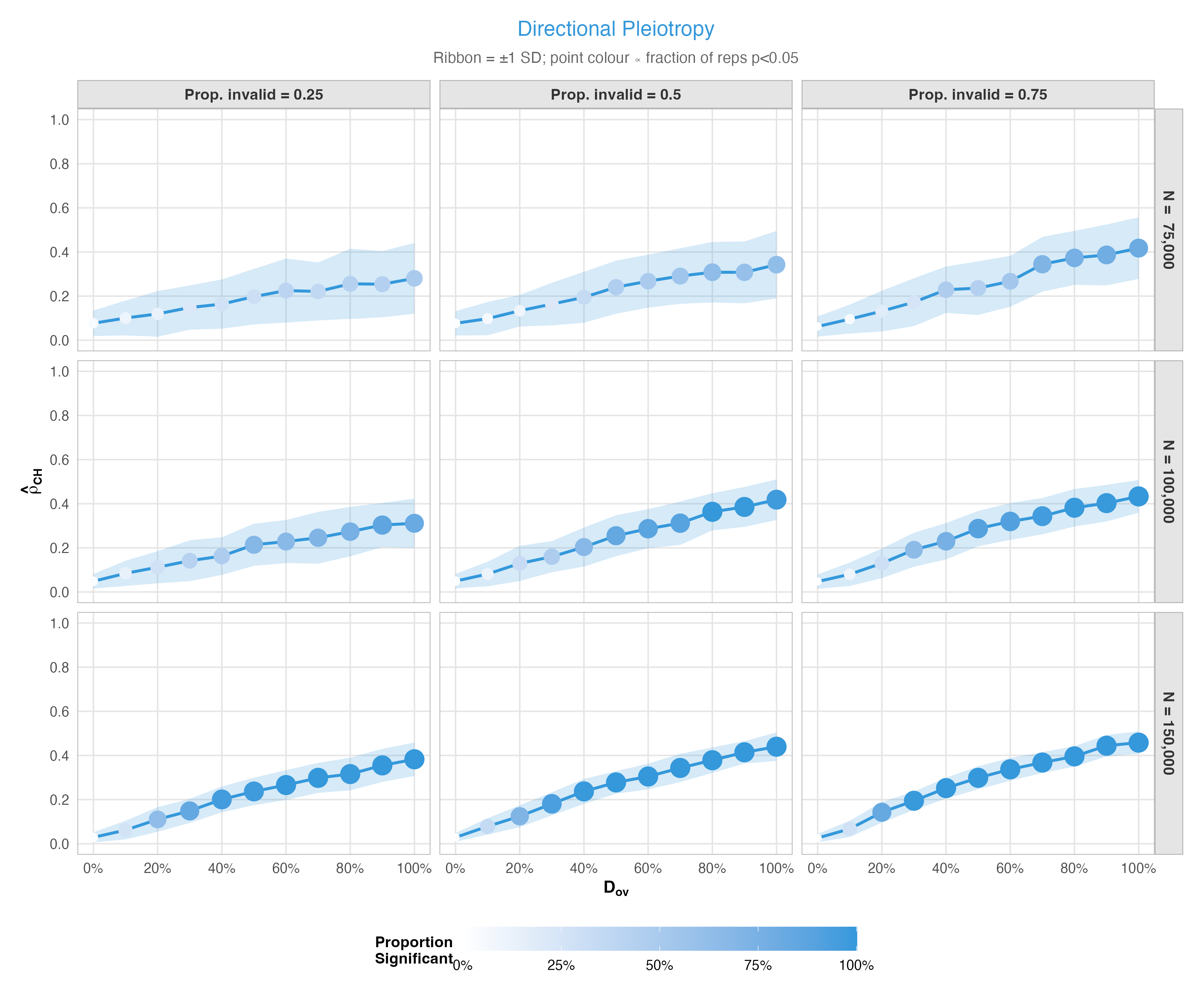}
    \caption{\textit{\textbf{Coheterogeneity statistic detects overlap in directional pleiotropic effects in the absence of shared confounders.}
Each panel shows the mean absolute coheterogeneity statistic ($|\hat{\rho}_{CH}^{(N)}|$) across 100 simulations, plotted against the overlap in invalid instruments ($D_{\mathrm{ov}}$) between two outcomes. In this scenario, invalid instruments exhibit shared directional pleiotropic effects but no shared confounders. Point size and color reflect the fraction of replicates with significant coheterogeneity ($p < 0.05$).}}
    \label{fig:cohet-Dirpleit}
\end{figure}

\begin{figure}
    \centering
    \includegraphics[width=0.49\linewidth]{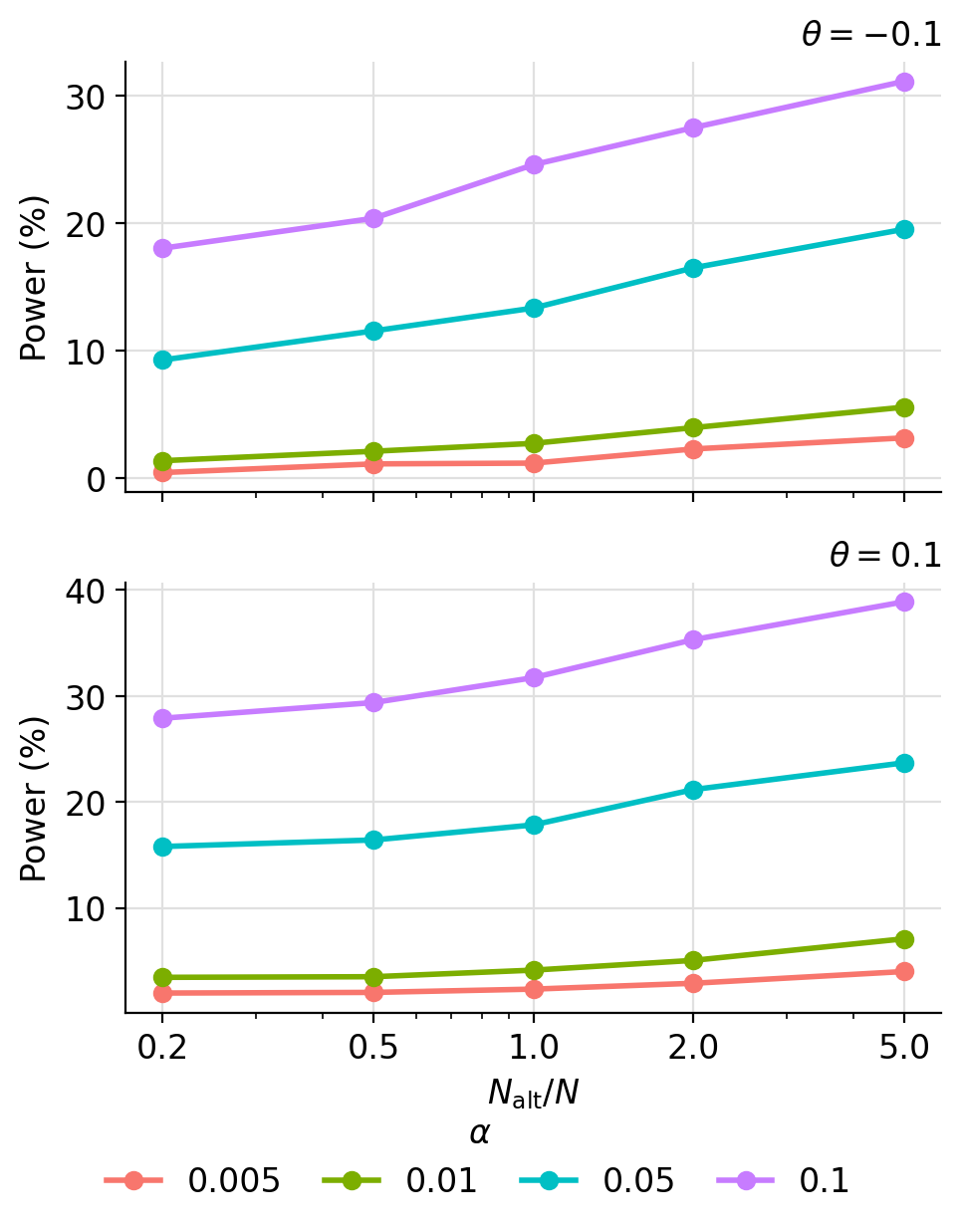}\hfill
    \includegraphics[width=0.49\linewidth]{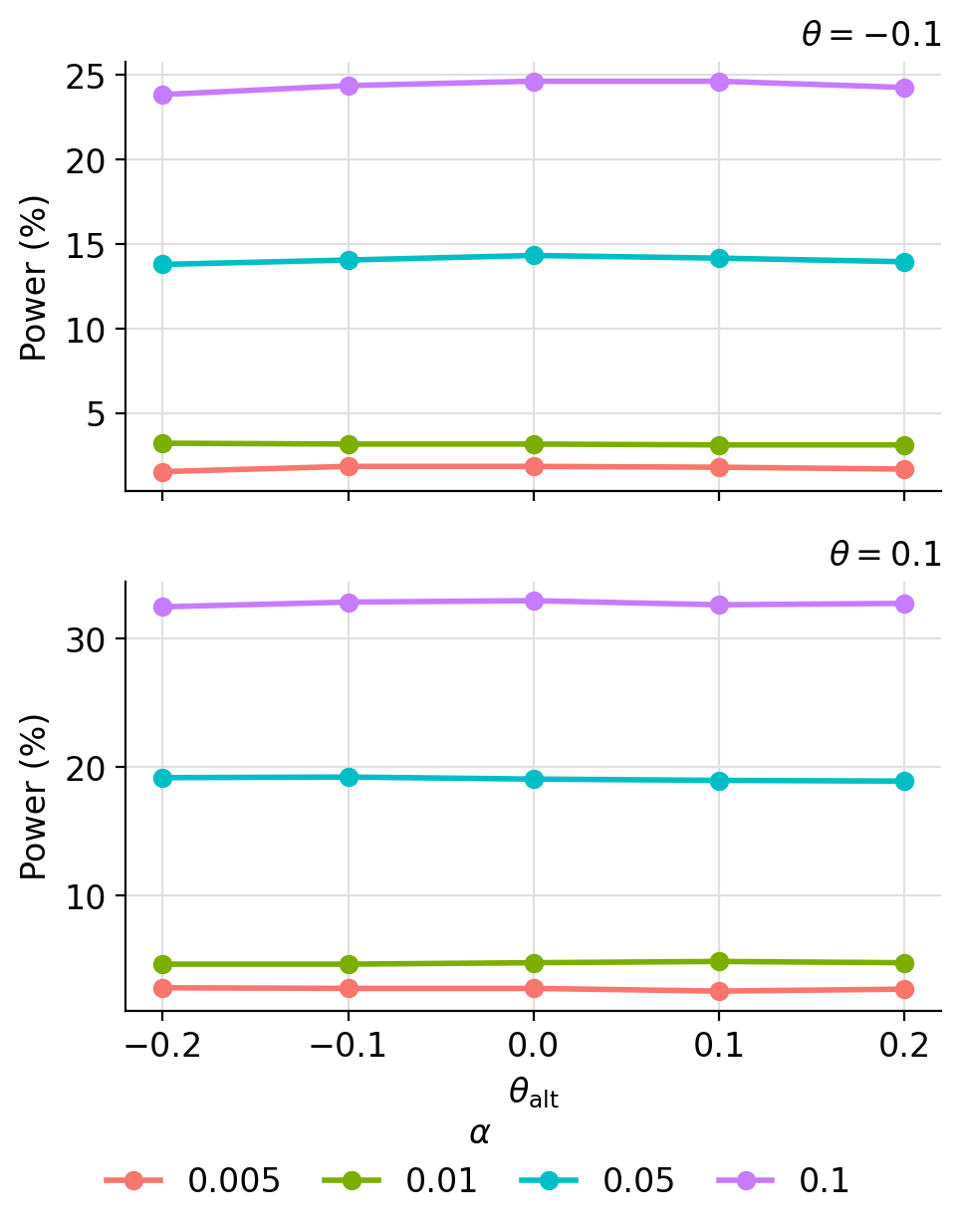}
    \caption{\textit{\textbf{IB-Mode power sensitivity to auxiliary sample size and secondary-trait effect.} Left: power versus auxiliary outcome sample-size ratio \(N_{\text{alt}}/N\) (log-scaled x-axis), stratified by \(\theta\), with curves for nominal \(\alpha\) levels; settings are fixed at \(N=10^5\), invalid-IV proportion \(0.5\), and overlap \(0.75\). Right: power versus secondary-trait causal effect \(\theta_{\text{alt}}\), stratified by \(\theta\), with curves for the same \(\alpha\) levels and the same fixed settings; results are averaged over auxiliary sample-size ratio. Both panels correspond to the InSIDE-violated (DN) scenario.}}
    \label{fig:power_combined}
\end{figure}

\begin{figure}
    \centering
    \includegraphics[width=0.5\linewidth]{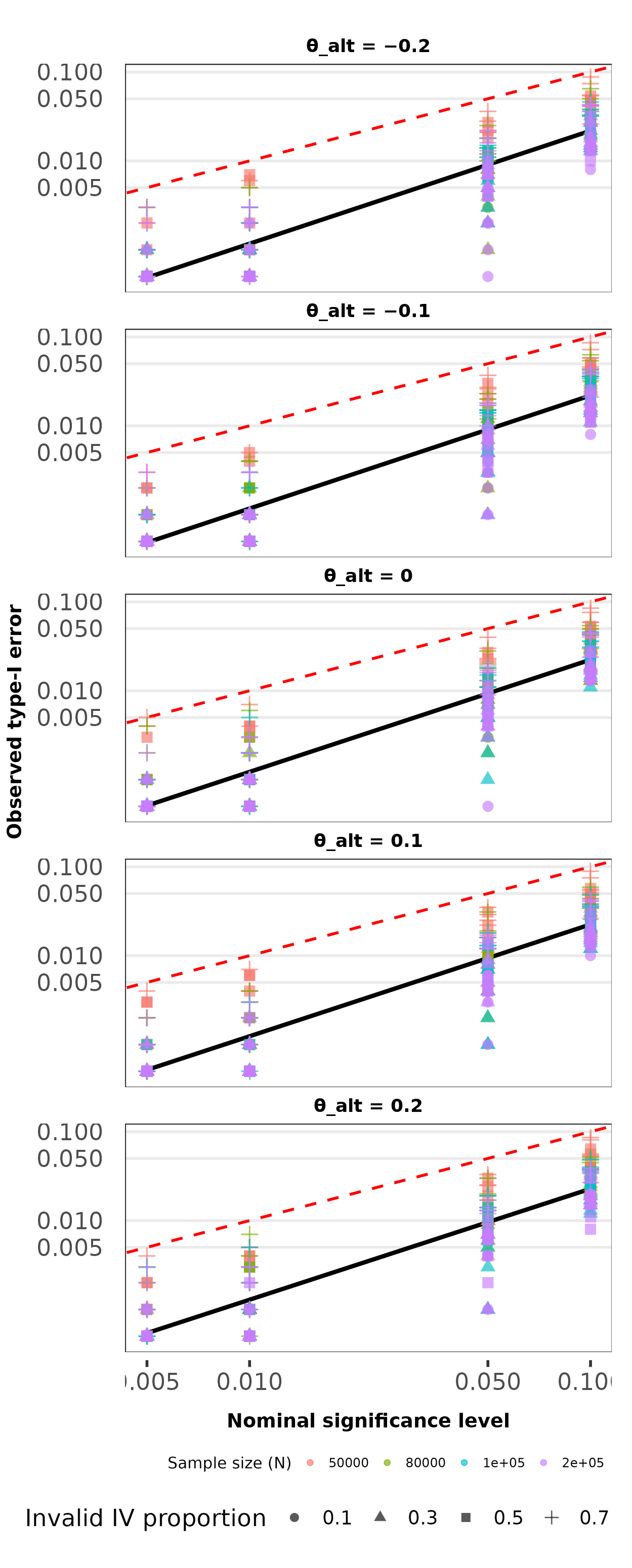}
    \caption{\textit{\textbf{Type-I error calibration of IB-Mode under the null target effect.} This figure evaluates calibration when the target-trait causal effect is null (\(\theta=0\)) and IB-Mode is run at \(\phi=1\). For each panel, the x-axis is the nominal significance level and the y-axis is the observed type-I error (empirical null rejection rate), with both axes shown on a \(\log_{10}\) scale to improve separation at small \(\alpha\). Results are shown for the InSIDE-violated (DN) scenario---directional pleiotropy with instrument invalidity arising from shared confounders---at an invalid-instrument overlap of $D_{\mathrm{ov}}=0.75$ (matching the scenario of main-text Figure~3). Panels correspond to \(\theta_{\text{alt}} \in \{-0.2,-0.1,0,0.1,0.2\}\), point color indicates sample size \(N\), and point shape indicates invalid-IV proportion. The red dashed line is the identity line (\(y=x\), exact calibration), and the black line is the panel-wise linear trend. Nominal levels shown correspond to \(\alpha \in \{0.005,0.01,0.05,0.1\}\).}
}
    \label{fig:type-1-error-calibration}
\end{figure}

\begin{figure}
    \centering
    \includegraphics[width=0.66\linewidth]{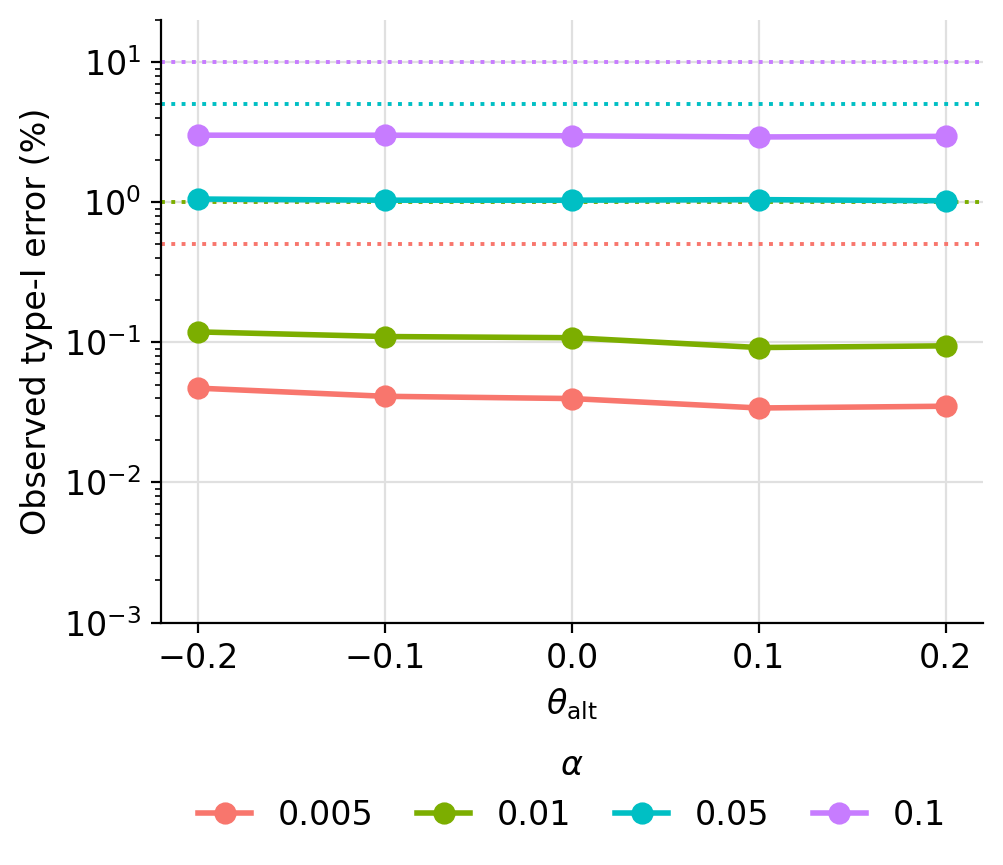}
    \caption{\textit{\textbf{Robustness of IB-Mode type-I error to secondary-trait effect size.} This figure shows observed type-I error of IB-Mode (y-axis, \%) versus \(\theta_{\text{alt}} \in \{-0.2,-0.1,0,0.1,0.2\}\) (x-axis), with separate curves for each nominal \(\alpha\), shown for the InSIDE-violated (DN) scenario. Points and solid lines are means across remaining settings (sample size \(N\), invalid-IV proportion, and overlap); dotted horizontal lines mark the corresponding nominal \(\alpha\) values. The y-axis is on a \(\log_{10}\) scale for better visualization.}
}
    \label{fig:robust_aux_trait}
\end{figure}

\begin{figure}[!h]
  \centering
  \begin{subfigure}[b]{0.9\textwidth}
    \centering
    \includegraphics[width=\textwidth]{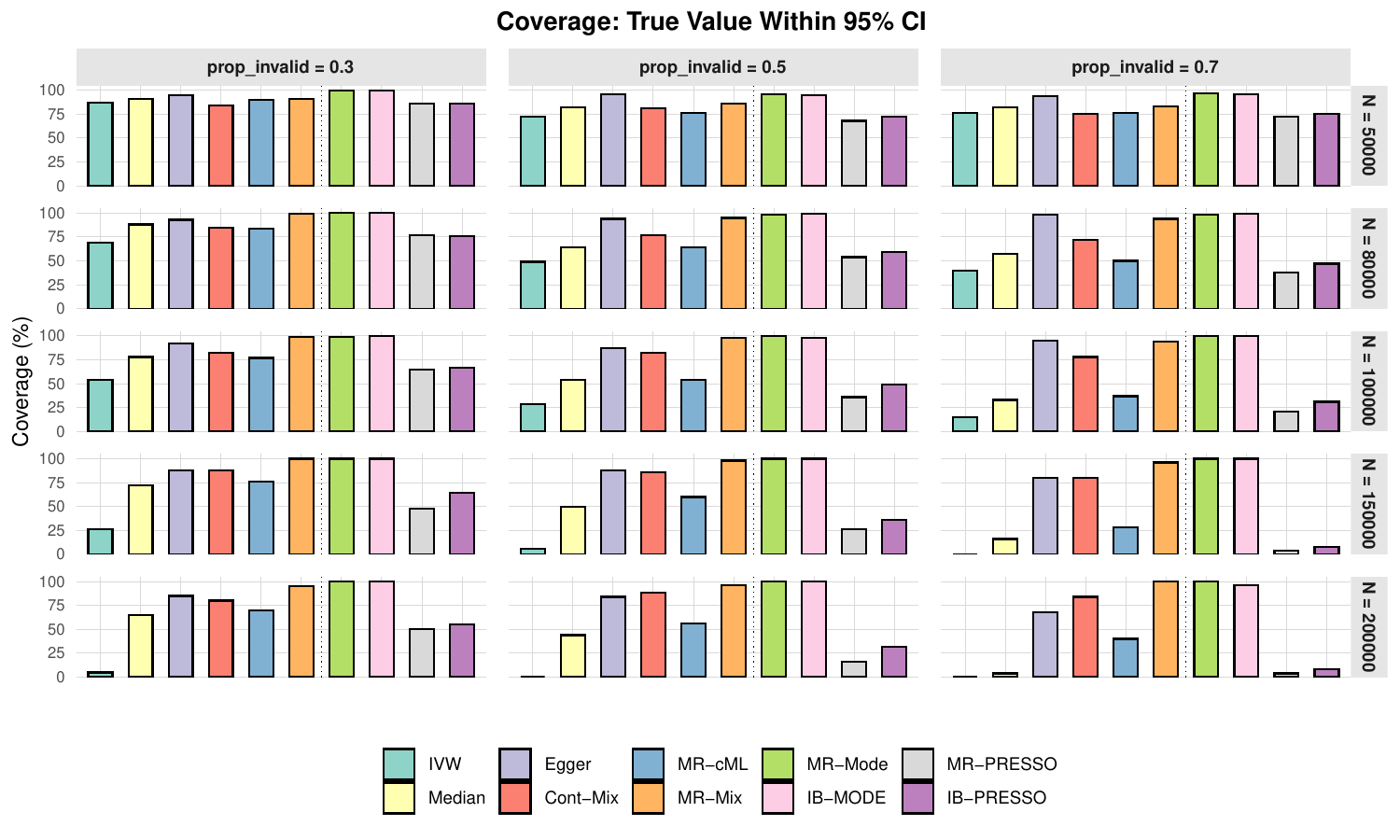}
  \end{subfigure}

  \vspace{0.5cm}

  \begin{subfigure}[b]{0.9\textwidth}
    \centering
    \includegraphics[width=\textwidth]{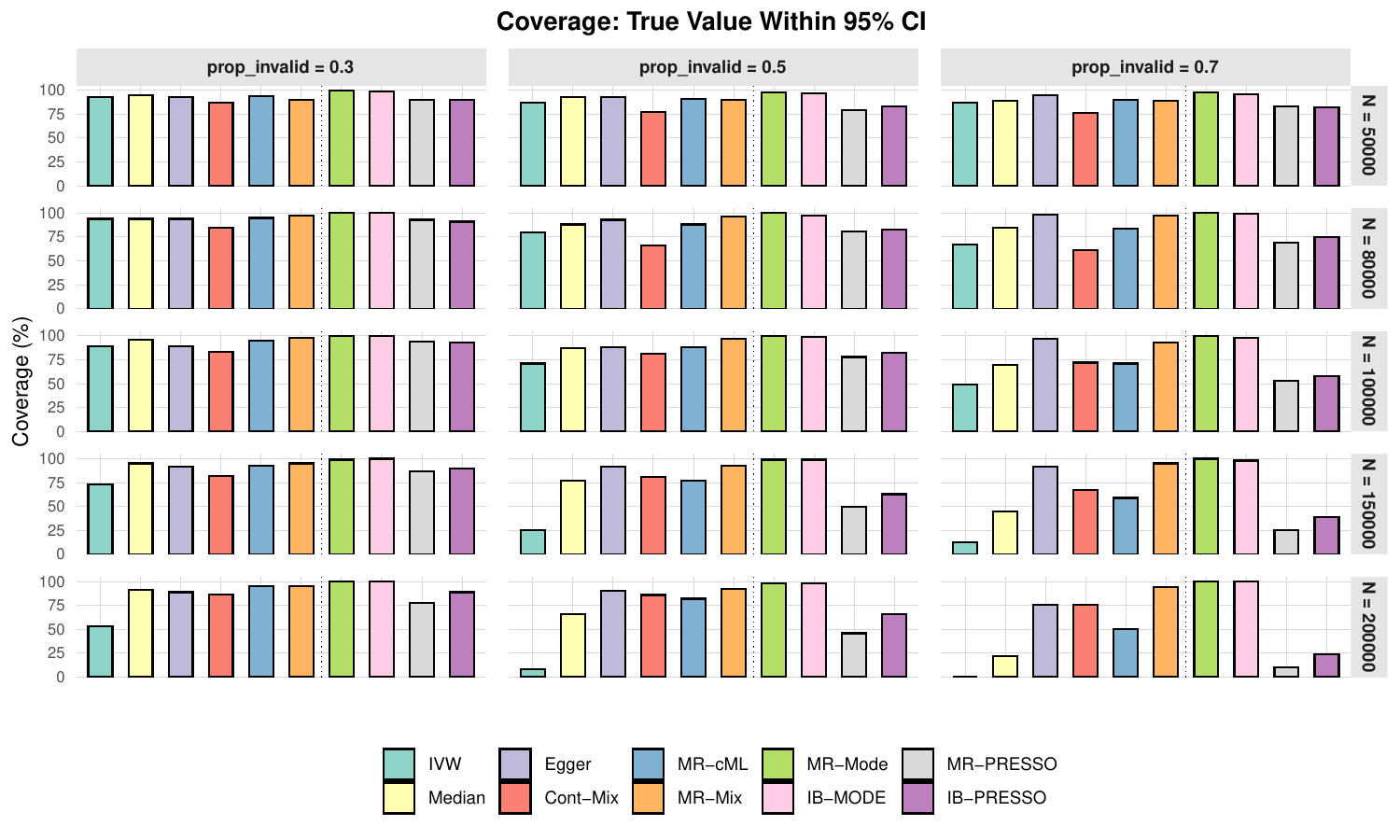}
    \caption{Coverage probability of 95\% confidence intervals.}
  \end{subfigure}

  \caption{
    \textbf{Coverage of 95\% confidence intervals for different MR methods.} Simulations are conducted under models where invalid instruments arise due to hidden heritable confounders leading to violation of the InSIDE assumption. The overlap in underlying invalid instruments across two outcome traits due to shared confounders is assumed to be $D_{\mathrm{ov}}=0.75$. Additionally, invalid instruments are allowed to have directional pleiotropic effects. The exposure has a causal effect of $\theta_{X\to Y_2}=0.3$ on the secondary (auxiliary) outcome trait. \textbf{(Top)} Coverage when the true causal effect for the primary outcome trait is $-0.1$. \textbf{(Bottom)} Coverage when the true causal effect for the primary outcome trait is $0.1$.
  }
  \label{fig:Main_simulation_results_coverage}
\end{figure}

\begin{figure}[!htp]
    \centering
    \includegraphics[width=0.7\linewidth]{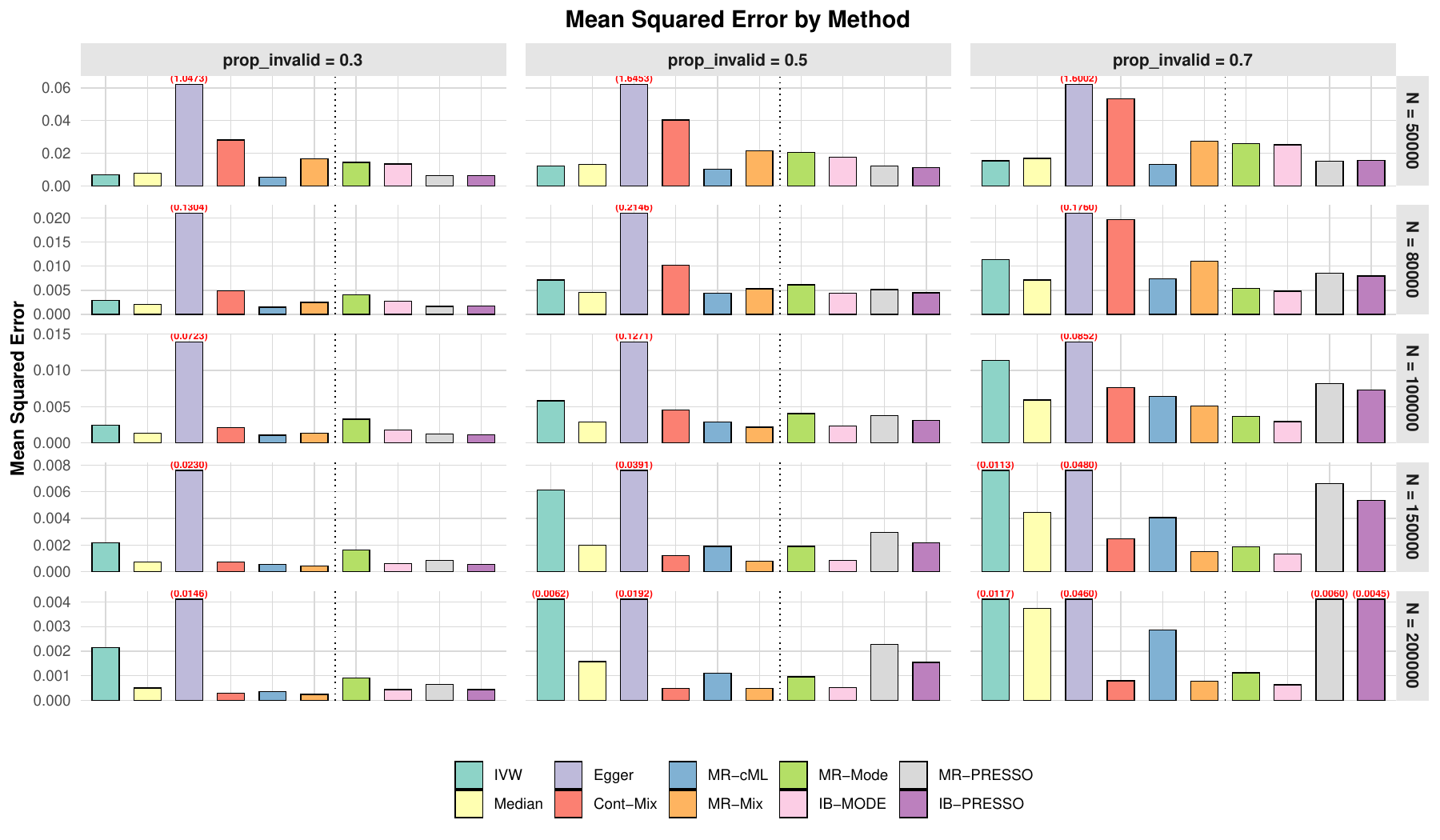}\\[0.5em]
    \includegraphics[width=0.7\linewidth]{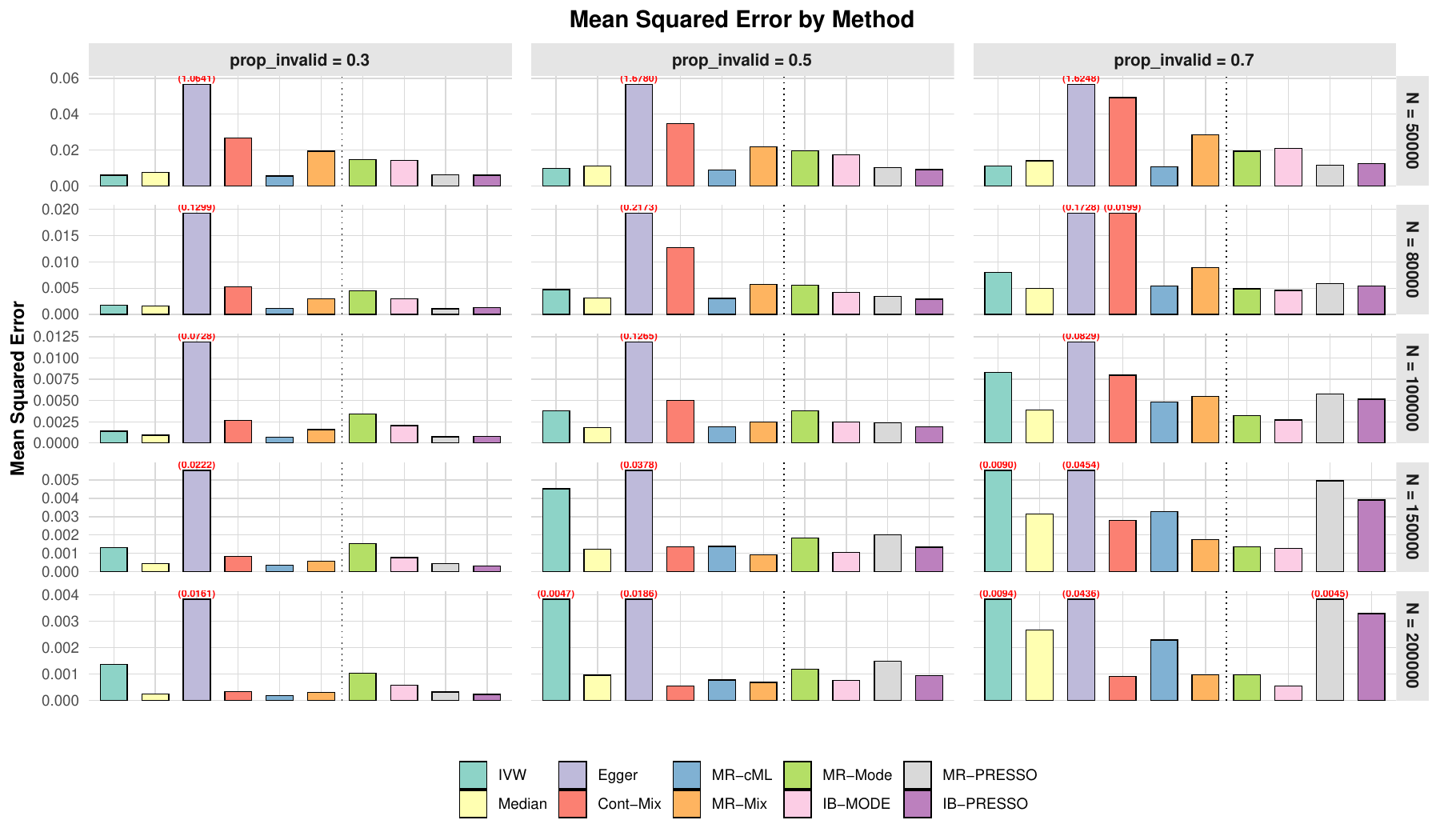}
    
    \caption{
        \textbf{Mean squared error (MSE) of different MR methods.} Simulations are conducted under models where invalid instruments arise due to hidden heritable confounders leading to violation of the InSIDE assumption. The overlap in underlying invalid instruments across two outcome traits due to shared confounders is assumed to be $D_{\mathrm{ov}}=0.75$. Additionally, invalid instruments are allowed to have directional pleiotropic effects. The exposure has a causal effect of $\theta_{X\to Y_2}=0.3$ on the secondary (auxiliary) outcome trait. \textbf{(Top)} MSE when the true causal effect for the primary outcome trait is 0. \textbf{(Bottom)} MSE when the true causal effect for the primary outcome trait is 0.1. Bars are truncated at $4\times$ the row-median MSE; actual values for truncated methods are shown in red.
        }
    \label{figsup:panel2x2_simulation}
\end{figure}

\begin{figure}
    \centering
    \includegraphics[width=1\linewidth]{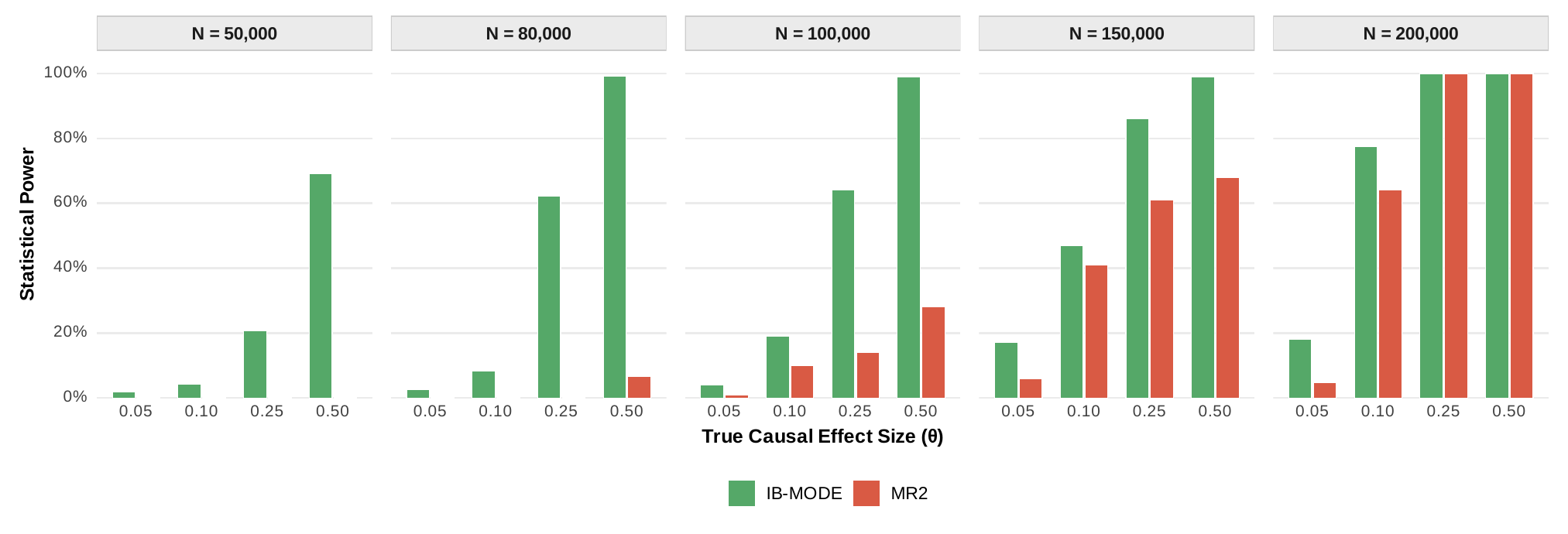}
    \caption{\textit{\textbf{Statistical Power Comparison: IB-Mode vs $\mathrm{MR}^2$.} We compared IB-Mode (green) against $\mathrm{MR}^2$ (red) across five sample sizes ($N \in \{50{,}000$--$200{,}000\}$) and four causal effect sizes ($\theta \in \{0.05, 0.10, 0.25, 0.50\}$) using 100 simulated datasets per scenario. Data were generated under a pleiotropy model with 200,000 variants (70\% invalid instruments), GWAS sample sizes $N_X = N$ and $N_y = N/2$ with 50\% overlap, and IV selection at $p < 5 \times 10^{-8}$. $\mathrm{MR}^2$ was calibrated to maintain Type I error $\leq 0.001$ via posterior inclusion probability thresholding.}}
    \label{fig:power_comparison_0.001}
\end{figure}

\begin{figure}[H]
    \centering
        \includegraphics[width=0.95\linewidth]{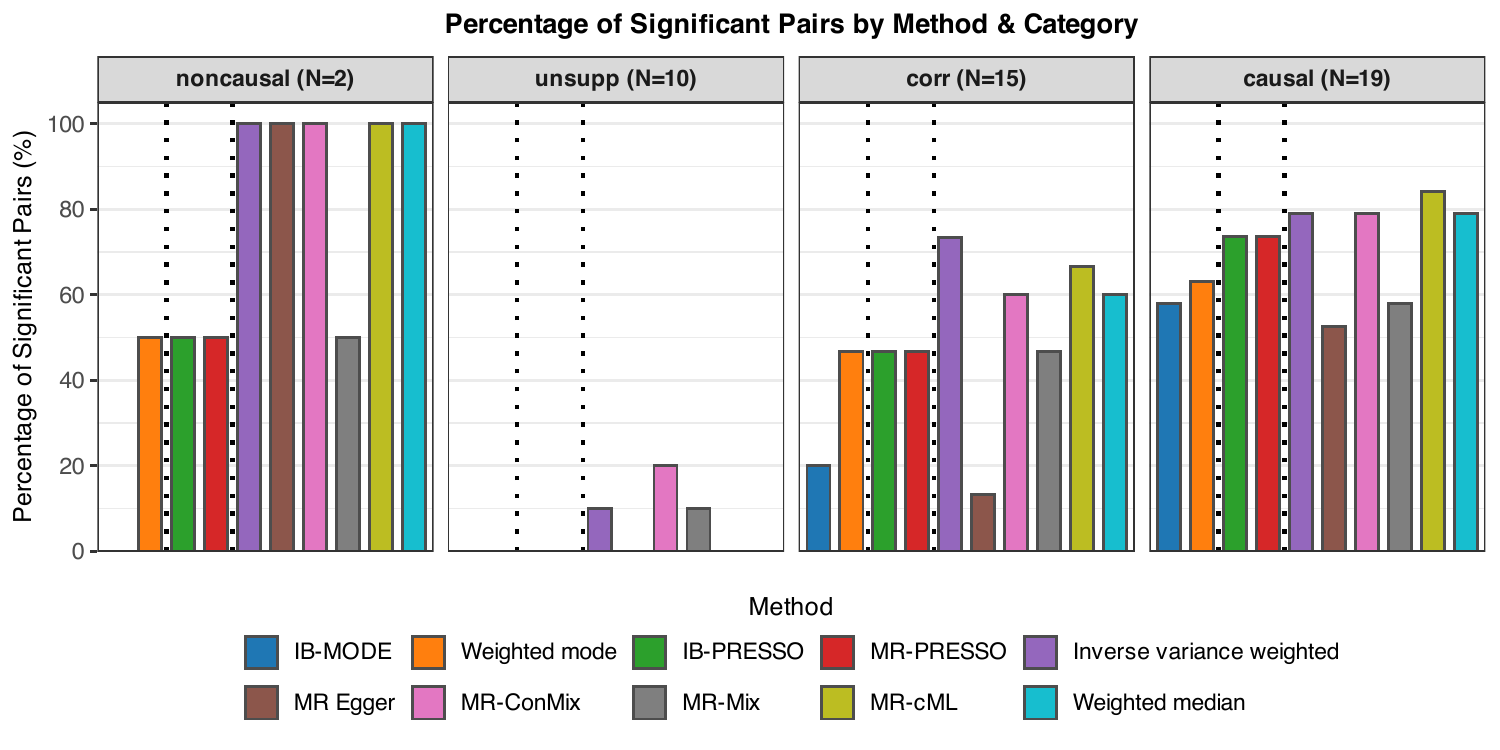}
      
    \caption{\textit{Re-analysis of hypothesized exposure-outcome relationships using  IB-based methods where the auxiliary-trait is selected as the one which has the second largest absolute value of the coheterogeneity statistic ($|\hat\rho_{CH}^{(N)}|$) with the primary outcome trait.  The corresponding results when the auxiliary traits were chosen to have the largest absolute value $|\hat\rho_{CH}^{(N)}|$ are shown in \cref{fig:main_multi_panel}B.}}
    \label{fig:sig_pair_IB2}
\end{figure}

\clearpage

%% file: supplement/2suppl_tabs.tex
\subsection{Supplemental Tables}
\label{Supp-Tab}
\setcounter{table}{0}
\renewcommand{\thetable}{S\arabic{table}}

\begin{table}[htbp]
\centering
\caption{\textit{MVP Outcome Traits: Trait Type, Description, and Sample Size (Effective SS for binary trait)}}
\begin{tabular}{lll}
\hline
binary & Coronary Artery Disease (CAD) & 177837.55 \\
binary & Stroke (Str/STK) & 89278.09 \\
binary & Transient Ischemic Attack (TIA) & 59316.20 \\
binary & Type 1 Diabetes (T1DM) & 66104.45 \\
binary & Type 2 Diabetes (T2DM) & 402913.45 \\
binary & Hyperlipidemia (HLD) & 291947.58 \\
binary & Sleep Apnea (OSA) & 48678.00 \\
binary & Hypertension (HTN) & 324905.98 \\
binary & Myocardial Infarction (MI) & 144151.05 \\
binary & Angina Pectoris (AP) & 138086.02 \\
binary & Pulmonary Heart Disease (PHD) & 89282.74 \\
binary & Pulmonary Embolism (PE) & 59142.70 \\
binary & Cardiomegaly (CMG) & 38275.76 \\
binary & Cardiomyopathy (CM) & 79783.72 \\
binary & Cardiac Conduction Disorders (CCD) & 220964.07 \\
binary & Cardiac Dysrhythmias (CD) & 374944.77 \\
binary & Congestive Heart Failure (CHF) & 205419.82 \\
binary & Atherosclerosis (AS) & 105886.66 \\
binary & Peripheral Vascular Disease (PVD) & 203570.30 \\
qt & Hemoglobin A1c (A1C) & 98735 \\
qt & Body Mass Index (BMI) & 118993 \\
qt & Brain Natriuretic Peptide (BNP) & 5833 \\
qt & Blood Urea Nitrogen (BUN\_BSP) & 117876 \\
qt & Creatine Kinase-MB (CKMB\_Abs) & 11686 \\
qt & Creatinine (Creat\_BSP) & 116707 \\
qt & C-Reactive Protein (CRP\_L) & 5051 \\
qt & Diastolic Blood Pressure (Diastolic) & 119332 \\
qt & Estimated GFR (eGFR) & 110856 \\
qt & Glucose (Glucose) & 118233 \\
qt & HDL Cholesterol (HDLC) & 113085 \\
qt & LDL Cholesterol (LDLC) & 113213 \\
qt & Systolic Blood Pressure (Systolic) & 119331 \\
qt & Triglycerides (Trig) & 107730 \\
qt & Troponin I (TroponinI) & 10375 \\
\hline
\end{tabular}
\label{Description}
\end{table}

\begin{table}[ht]
\centering
\caption{Summary of exposures and outcomes analysed. Traits from \citet{morrison2020mendelian} are listed separately from additional traits introduced in this study. Negative controls are indicated.}
\label{tab:traits_summary_compact}
\footnotesize
\begin{tabularx}{\linewidth}{lXX}
\toprule
Source & Exposures & Outcomes \\
\midrule
From \citet{morrison2020mendelian} &
\textit{TG}, \textit{HDL}, \textit{LDL}, alcohol, smoke, \textit{BF}, \textit{BW}, \textit{BMI}, height, \textit{FG}, \textit{SBP}, \textit{DBP} &
\textit{CAD}, stroke, \textit{T2D}, asthma\textsuperscript{\dag} \\
Additional (this study) &
\textit{vitamin D} &
\textit{HTN}, \textit{HDL-C}, \textit{LDL-C}, \textit{eGFR} \\
\bottomrule
\end{tabularx}

\vspace{2pt}
\raggedright\footnotesize
\textsuperscript{\dag} Negative control outcome.
\end{table}

\begin{table}[ht]
\centering
\caption{GWAS sources and sample sizes for exposures}
\begin{tabular}{l l l}
\hline
Exposure                   & GWAS Source                  & Sample Size ($N$) \\
\hline
BMI                        & GIANT \citep{locke2015genetic}       & 339,224 \\
HDL-C, LDL-C, TG           & GLGC \citep{dron2025breadth}         & 1,654,960 \\
SBP, DBP                   & \citep{keaton2024genome}             & 1,028,980 \\
vitamin D                  & \citep{revez2020genome}              & 417,580 \\
Fasting glucose            & \citep{chen2021trans}                & 200,622 \\
Height                     & GIANT \citep{wood2014defining}       & 253,288 \\
Birth weight               & UK Biobank \citep{NealeLab2017_UKBiobank} & 193,063 \\
Body fat percentage        & UK Biobank \citep{Elsworth2018_BodyFatPercentage} & 454,633 \\
Smoking status             & UK Biobank \citep{Elsworth2018_EverSmoked} & 461,066 \\
Alcoholic drinks per week  & \citep{liu2019association}            & 335,394 \\
\hline
\end{tabular}
\end{table}

\begin{table}
\centering
\caption{\textit{Summary of methods supporting causal or correlated relationships between exposures and outcomes across association categories. The selection of diseases, risk factors and the classification of their relationship are adapted from \citep{morrison2020mendelian} (\checkmark: support from method)}}
\begin{adjustbox}{max width=\textwidth}
\begin{tabular}{l l l c c c c c c c c c c}
\hline
Exposure & Outcome & Category & IB-Mode & IB-PRESSO & Inverse variance weighted & MR Egger & MR-ConMix & MR-Mix & MR-PRESSO & MR-cML & Weighted median & Weighted mode \\
\hline
Smoking & Asthma & causal & & & & & & & & & & \\
BF & CAD & causal & \checkmark & \checkmark & \checkmark & & \checkmark & \checkmark & \checkmark & \checkmark & \checkmark & \checkmark \\
BMI & CAD & causal & \checkmark & \checkmark & \checkmark & \checkmark & \checkmark & \checkmark & \checkmark & \checkmark & \checkmark & \checkmark \\
DBP & CAD & causal & \checkmark & \checkmark & \checkmark & \checkmark & \checkmark & \checkmark & \checkmark & \checkmark & \checkmark & \checkmark \\
Height & CAD & causal & \checkmark & \checkmark & \checkmark & & \checkmark & & \checkmark & \checkmark & \checkmark & \\
LDL & CAD & causal & \checkmark & \checkmark & \checkmark & \checkmark & \checkmark & \checkmark & \checkmark & \checkmark & \checkmark & \checkmark \\
SBP & CAD & causal & \checkmark & \checkmark & \checkmark & \checkmark & \checkmark & \checkmark & \checkmark & \checkmark & \checkmark & \checkmark \\
Smoking & CAD & causal & & \checkmark & & & & & & \checkmark & & \\
TG & CAD & causal & \checkmark & \checkmark & \checkmark & \checkmark & \checkmark & \checkmark & \checkmark & \checkmark & \checkmark & \checkmark \\
BF & Stroke & causal & & \checkmark & \checkmark & & \checkmark & & \checkmark & \checkmark & \checkmark & \\
BMI & Stroke & causal & & & \checkmark & & \checkmark & \checkmark & & \checkmark & \checkmark & \\
DBP & Stroke & causal & \checkmark & \checkmark & \checkmark & \checkmark & \checkmark & & \checkmark & \checkmark & \checkmark & \checkmark \\
LDL & Stroke & causal & & \checkmark & \checkmark & \checkmark & \checkmark & & \checkmark & \checkmark & \checkmark & \checkmark \\
SBP & Stroke & causal & \checkmark & \checkmark & \checkmark & \checkmark & \checkmark & & \checkmark & \checkmark & \checkmark & \checkmark \\
Smoking & Stroke & causal & & & & & & & & & & \\
BF & T2D & causal & \checkmark & \checkmark & \checkmark & \checkmark & \checkmark & \checkmark & \checkmark & \checkmark & \checkmark & \checkmark \\
BMI & T2D & causal & \checkmark & \checkmark & \checkmark & \checkmark & \checkmark & \checkmark & \checkmark & \checkmark & \checkmark & \checkmark \\
FG & T2D & causal & \checkmark & \checkmark & \checkmark & & \checkmark & \checkmark & \checkmark & \checkmark & \checkmark & \checkmark \\
Smoking & T2D & causal & & & & & \checkmark & \checkmark & & & & \\
BF & Asthma & correlated & & \checkmark & \checkmark & & \checkmark & & \checkmark & \checkmark & \checkmark & \\
BMI & Asthma & correlated & & \checkmark & \checkmark & & \checkmark & & \checkmark & \checkmark & \checkmark & \\
BW & CAD & correlated & & \checkmark & \checkmark & & & & \checkmark & & & \\
FG & CAD & correlated & \checkmark & \checkmark & \checkmark & & \checkmark & & \checkmark & \checkmark & \checkmark & \\
Alcohol & Stroke & correlated & & & & & & \checkmark & & & & \\
BW & Stroke & correlated & & & & & & & & & & \\
Height & Stroke & correlated & & & \checkmark & & & & & \checkmark & & \\
TG & Stroke & correlated & & \checkmark & \checkmark & & \checkmark & & \checkmark & \checkmark & & \\
Alcohol & T2D & correlated & \checkmark & & & & \checkmark & \checkmark & & & \checkmark & \checkmark \\
BW & T2D & correlated & \checkmark & \checkmark & \checkmark & & \checkmark & \checkmark & \checkmark & \checkmark & \checkmark & \checkmark \\
DBP & T2D & correlated & \checkmark & \checkmark & \checkmark & & \checkmark & \checkmark & \checkmark & \checkmark & \checkmark & \checkmark \\
HDL & T2D & correlated & & \checkmark & \checkmark & & \checkmark & \checkmark & \checkmark & \checkmark & \checkmark & \checkmark \\
LDL & T2D & correlated & & & & \checkmark & & & & & & \\
SBP & T2D & correlated & & \checkmark & \checkmark & & \checkmark & \checkmark & \checkmark & \checkmark & \checkmark & \checkmark \\
TG & T2D & correlated & & \checkmark & \checkmark & \checkmark & \checkmark & \checkmark & & \checkmark & \checkmark & \checkmark \\
Alcohol & Asthma & unsupported & & & & & & & & & & \\
BW & Asthma & unsupported & & & & & & & & & & \\
DBP & Asthma & unsupported & & & & & & & & & & \\
FG & Asthma & unsupported & & & & & & & & & & \\
HDL & Asthma & unsupported & & & & & & & & & & \\
Height & Asthma & unsupported & & & & & & & & & & \\
LDL & Asthma & unsupported & & & & & & & & & & \\
SBP & Asthma & unsupported & & & & & & & & & & \\
TG & Asthma & unsupported & & & \checkmark & & & & & & & \\
Height & T2D & unsupported & & & & & \checkmark & \checkmark & & & & \\
HDL & CAD & noncausal & & \checkmark & \checkmark & \checkmark & \checkmark & \checkmark & \checkmark & \checkmark & \checkmark & \\
HDL & Stroke & noncausal & & & \checkmark & \checkmark & \checkmark & & & \checkmark & \checkmark & \checkmark \\
\hline
\end{tabular}
\label{summ_method_morrison}
\end{adjustbox}
\end{table}

\begin{table}[!h]
\centering
\caption{\textit{Percentage efficiency gain (\%) of IB-Mode (vs Weighted mode) and IB-PRESSO (vs MR-PRESSO) across exposures for the causal pairs. IB-Mode gains are reported at the default bandwidth $\phi=1$ and at the more conservative $\phi=0.5$.}}
\resizebox{\ifdim\width>\linewidth\linewidth\else\width\fi}{!}{
\begin{tabular}{lrrr}
\toprule
Exposure & IB-Mode ($\phi{=}1$, \%) & IB-Mode ($\phi{=}0.5$, \%) & IB-PRESSO (\%) \\
\midrule
FG       & -43.0  & -31.7  &  -4.2 \\
TG       & -42.5  &  -4.4  &   3.6 \\
BF       &  58.5  & 208.0  &  10.2 \\
BMI      &   3.6  &  98.0  &  29.2 \\
DBP      &  92.6  &  95.5  & -11.4 \\
Height   &  66.3  &  37.8  &  22.2 \\
LDL      & -56.1  &  52.7  &   1.5 \\
SBP      &   9.7  &  15.2  &   0.2 \\
Smoking  & 348.1  & 233.0  &  -6.2 \\
\bottomrule
\end{tabular}}
\label{causalgain}
\end{table}

\clearpage